\documentclass[final,onefignum,onetabnum]{siamonline190516}

\usepackage[mathscr]{euscript}
\usepackage{color,amsmath,amssymb,mathrsfs}
\usepackage{graphicx}
\usepackage{algorithm}
\usepackage{algpseudocode}
\usepackage{tikz,pgfplots}
\pgfplotsset{compat=1.18}
\usepackage{tabularx}
\makeatother
\usepackage{upgreek}
\usepackage{cite}
\usepackage{multirow}
\usepackage{yfonts}
\usepackage{mathtools}
\usepackage[normalem]{ulem}
\usepackage{latexsym}
\DeclareMathAlphabet{\mathpzc}{OT1}{pzc}{m}{it}
\usepackage{multirow}

\usepackage{url} 
\usepackage[margin=1in]{geometry}
\usepackage{tikz} 
\usepackage{comment}
\usepackage{dsfont}
\newcommand{\norm}[2][]{\left\Vert #2\right\Vert_{#1}}

\newsiamremark{rmk}{Remark}

\newcommand{\GGamma}{{\Gamma}}
\newcommand{\zsp}{\mathcal{P}}
\newcommand{\msp}{\mathcal{M}}

\newcommand{\e}{{e}}

\newcommand{\f}{\mathscr{G}}

\newcommand{\dd}{d_1}
\newcommand{\ddd}{d_2}
\newcommand{\ee}{{e_1}}
\newcommand{\eee}{e_2}
\newcommand{\z}{p}
\newcommand{\w}{m}
\newcommand{\SPD}{\mathbb{S}^{+}_n}
\newcommand{\CON}{\mathbb{A}}
\newcommand{\data}{{d}}

\newcommand{\CC}{{C}}
\newcommand{\DD}{{D}}

\newcommand{\s}{{s}}

\newcommand{\LL}{{L}}

\newcommand{\tr}[1]{{\rm tr}\left(\ensuremath{#1}\right)}
\newcommand{\diag}{\mathop{\mathrm{diag}}}

\begin{document}

\def\addressRO{Department of Engineering Science, University of Auckland, Auckland 1010, New Zealand}
\def\addressJ{Department of Technical Physics, University of Eastern Finland, Kuopio Campus, 70211 Kuopio, Finland}
\def\addressM{Rocsole Limited, Kartanonkatu 2,
70700 Kuopio, Finland}

\author{Ruanui Nicholson\footnotemark[1]
  \and Matti Niskanen\footnotemark[2]
  \and  Oliver J. Maclaren\footnotemark[1]
    \and Jari P. Kaipio\footnotemark[3]}
\renewcommand{\thefootnote}{\fnsymbol{footnote}}
\footnotetext[1]{\addressRO.}
\footnotetext[2]{\addressM.}
\footnotetext[3]{\addressJ.}

\title{Beyond Independence: on Jointly Normal Priors in Bayesian Inversion}

\maketitle

\begin{abstract}
We consider joint inversion for two or more unknown parameters from observational data in the Bayesian framework. Standard approaches often either treat the parameters as independent or impose structural similarity through regularisation terms that can be difficult to interpret statistically. We instead construct jointly Gaussian prior models with prescribed Gaussian marginals, so that correlation between the parameters can be incorporated without altering the marginal prior distributions. We propose a joint covariance construction that preserves the marginals, allows spatially varying cross-correlation, and supports uncertainty and inference in the correlation itself. The construction is valid for any strict contraction encoding the desired cross-correlation and is optimal in a canonical correlation sense under the principal square root factorisation. We demonstrate the method using prior sampling and several inference examples: a low-dimensional illustrative example and two higher-dimensional examples, including a PDE-constrained problem. The examples highlight both the potential pitfalls of ignoring or neglecting uncertainty in the correlation as well as reinforcing a key principle of the Bayesian paradigm: unknown quantities included in a model should be treated as random variables.
\end{abstract}

\section{Introduction}


We consider joint inversion of distributed parameters in the Bayesian framework. The aim of joint inversion is to estimate two or more {\em inversion parameters}, denoted here by $\z\in\zsp$ and $\w\in\msp$,  based on noisy observed data. In various settings there may be measured data from two independent problems. In such cases the setup is often modelled as
\begin{equation}
\label{eq: sepjoint}
\begin{aligned}
\dd&=\mathscr{F}_1(\z)+\ee,\\
\ddd&=\mathscr{F}_2(\w)+\eee,
\end{aligned} 
\end{equation}
with $\dd\in\mathbb{R}^{q_1}$ and $\ddd\in\mathbb{R}^{q_2}$ denoting the vectors of measured data, $\mathscr{F}_1:\zsp\to\mathbb{R}^{q_1}$ and $\mathscr{F}_2:\msp\to\mathbb{R}^{q_2}$ denoting the associated parameter-to-observable mappings (i.e., the forward problems) and $\ee\in\mathbb{R}^{q_1}$ and $\eee\in\mathbb{R}^{q_2}$ denoting measurement errors in each of the problems. Such a setup is commonly referred to as the {\em multi-physics} case in the literature. The more general setup, however, is of the form
\begin{align}\label{eq: genjoint}
\data&=\f(\z,\w)+\e
\end{align} 
where $\data\in\mathbb{R}^q$ for $q=q_1+q_2$ is the measured data, $\f:\zsp\times\msp\to\mathbb{R}^q$ the parameter-to-observable map  (i.e., the joint forward problem), and $\e\in\mathbb{R}^q$ the measurement errors in the data.
  
Joint inversion problems naturally occur in geophysics~\cite{GallardoMeju04,AbubakarGaoHabashyEtAl12,KowalskyFinsterlePetersonEtAl05,HaberGazit13}, biomedical imaging~\cite{SongLiYangEtAl19,EhrhardtThielemansPizarroEtAl14,RaschBrinkmannBurger17}, and non-destructive testing~\cite{inglese2004corrosion,NicholsonNiskanen21} among various other fields. Closely related problems also arise in co-kriging and multivariate spatial prediction, where a key task is to model and condition multiple correlated spatial fields \cite{wackernagel2003multivariate,chiles2012geostatistics}. The most common approaches to joint inversion are framed within the deterministic setting and are often based on (generalised) total variation with the goal of enforcing {\em structural similarity} between the parameters of interest, see~\cite{CrestelStadlerGhattas18} for an in-depth review. Though in some cases it may be possible to interpret these approaches in a statistical manner, how to choose a suitable regularisation parameter, encode uncertainty or spatial variability in the level of structural similarity, or quantify the resulting uncertainties remains unclear. It is also well documented that total variation-based approaches suffer from theoretical limitations, such as discretisation invariance issues~\cite{LassasSiltanen04}. Moreover, if there is no prior knowledge of any relationship between the inversion parameters, the default approach is to treat them as independent, see, e.g.,~\cite{recinos2023framework,petra2012inexact,NicholsonNiskanen21,hanninen2020application} among many.

In the current paper, we focus on the design of joint Gaussian prior models for high-dimensional Bayesian inverse problems. More specifically, we consider the case where the marginal prior models for $\z$ and $\w$ are known and we wish to encode correlation (or at least the possibility of correlation) between the two parameters. As such, the problem reduces to finding a suitable joint covariance matrix. The somewhat related problem of finding a valid correlation matrix has been, and continues to be, extensively studied~\cite{higham2002computing,boyd2005least,qi2006quadratically,qi2007application,borsdorf2010computing,higham2016bounds}. However, these works are typically aimed at finding the {\em nearest} correlation matrix based on sample-approximated correlation matrices when missing data has compromised the correlations and led to a non-positive (semi-)definite matrix.

Our underlying objectives for constructing the joint prior are fourfold: 
\begin{enumerate}
\item The constructed joint prior should not alter the two marginal distributions $\pi(\z)$ and $\pi(\w)$; 
\item Correlation between the parameters should be easy to incorporate (while still satisfying 1); 
\item The approach should easily allow for uncertainty in the correlation.
\item Computational costs of the prior\footnote{sampling from and/or evaluating} should not drastically increase using the joint prior. 
\end{enumerate}

\paragraph{Contribution}

The main contributions of this article are the provision of a covariance matrix which maintains the two marginal distributions (Objective 1) and easily allows for incorporation of a given possible spatially varying correlation (Objective 2). The proposed covariance formulation also allows for straightforward incorporation and quantification of uncertainty in the correlation (Objective 3). Finally, we provide a convenient approach to carry out the required computations involving the joint prior based on the marginal distributions so that the cost associated with evaluating (and sampling from) the joint prior are not significantly more than carrying out the related computations for the case of independent parameters (Objective 4). We demonstrate the applicability of the approach on an illustrative low-dimensional problem, and two higher-dimensional problems, one of which is governed by a partial differential equation (PDE). The illustrative examples serve to demonstrate both the dangers of neglecting uncertainty, namely uncertainty in the correlation between the two parameters, along with the benefits of modelling the uncertainty in the correlation.

\paragraph{Paper overview}
The remainder of the paper is organised as follows. In Section \ref{sec: Bjoint} we review the Bayesian approach to joint inversion and propose our method of constructing a jointly normal model, as well as demonstrating its use on generating correlated samples. Several technical considerations are discussed and demonstrated on joint inversion for the  maximum population and half-velocity constant of the Monod model in Section~\ref{sec: Ignorance}. Section \ref{sec: NumEx} is used to demonstrate the applicability of our proposed approach on two higher dimensional examples motivated by geophysical settings. Finally, in the last section,
Section \ref{sec: Conc},
 we close with some conclusions.

\section{The Bayesian approach to joint inversion}\label{sec: Bjoint}
We consider problems in the finite-dimensional setting with $\z\in\mathbb{R}^{n_1}$, $\w\in\mathbb{R}^{n_2}$ and set $n=n_1+n_2$. However, the results given here should hold in the infinite-dimensional (Hilbert space) setting~\cite{Stuart10}  with special care taken to ensure the correct spaces are used~\cite{BuiGhattasMartinEtAl13,PetraMartinStadlerEtAl14}. We denote by $\SPD$ the space of $n\times n$ symmetric positive definite matrices. A matrix $A\in\mathbb{R}^{n\times q}$ is said to be a strict contraction if $\norm[2]{A}<1$, i.e., all singular values of $A$ are strictly less than 1. We denote the space of all $n\times q$ strict contractions by $\CON^{n\times q}$. 

Posed within the Bayesian framework~\cite{KaipioSomersalo05,CalvettiSomersalo07,Stuart10} the joint inversion problem becomes a problem of joint inference. The Bayesian approach treats all unknowns as random variables and naturally allows for the incorporation, and subsequent quantification, of uncertainty. The solution to the inference problem is the (joint) posterior density $\pi(\z,\w\vert\data)$ which by Bayes' theorem can be written as
\begin{align}
\pi_{\rm post}(\z,\w\vert\data)\propto\pi(\data\vert\z,\w)\pi(\z,\w),
\end{align}
where $\pi(\data\vert\z,\w)$ denotes the likelihood function and $\pi(\z,\w)$ the joint prior which encodes our beliefs about the parameters prior to considering the data. However, it is well understood that the relevance of any resulting estimates is heavily dependent on the ability to accurately model and account for the various sources of uncertainty~\cite{KaipioSomersalo05,CalvettiSomersalo07}. Failure to accurately account for the statistics of noise and/or errors typically results in a mis-specified likelihood, and consequently, misleading posterior estimates~\cite{KaipioKolehmainen13,CalvettiDunlopSomersaloEtAl18,NicholsonPetraKaipio18,BabaniyiNicholsonVillaEtAl21}. However, the same can be said for the prior: {\em failure to model the prior uncertainty in the parameters can result in misleading posterior estimates}. The design of prior densities for a single (distributed) parameter for Bayesian inference is a fairly mature field~\cite{RoininenHuttunenLasanen14,KaipioKolehmainenVauhkonenEtAl99,CalvettiSomersalo18,LindgrenRueLindstrom11,ChadaRoininenSuuronen21}. On the other hand, for joint inference of two or more distributed parameters, the literature is comparatively modest.

\subsection{A jointly normal model}\label{sec: norm} Here we present an efficient means to construct a joint Gaussian distribution. The proposed approach can be seen as a generalisation of the so-called {\em co-simulation} approach \cite{Oliver03}. The approach outlined here allows for a spatially varying correlation between the two parameter fields of interest, correlation between parameters of different dimensions, and (most importantly) uncertainty in the correlation itself. 

As stated previously, our approach is based on the construction of a joint Gaussian, thus provision of a suitable joint covariance matrix is the key to the approach. Before presenting our main results, we recall a key theorem.

\begin{theorem}\label{thm:contract}
    Let $H_1\in\mathbb{S}^+_{n_1}$ and $H_2\in\mathbb{S}^+_{n_2}$. Then the matrix 
    \begin{align}
H=\begin{pmatrix}
H_1 & H_{12}\\
 H_{12}^T & H_2
 \end{pmatrix}
    \end{align}
    is positive definite {\em if and only if} there exists a strict contraction $K\in\mathbb{A}^{n_1\times n_2}$ such that $H_{12}=H_1^{\frac{1}{2}}K H_2^{\frac{1}{2}}$.
\end{theorem}
\begin{proof} The result and proof can be found in, for example, \cite[Theorem 7.7.7]{HornJohnson12}.
\end{proof}

We then propose the following construction of the joint covariance matrix for a given desired cross-correlation encoded in $C\in\CON^{n_1\times n_2}$;
\begin{align}\label{eq: CovOut}
\GGamma=
\begin{pmatrix}
\GGamma_\z & \LL_\z^{-1}\CC\LL_\w^{-T}\\
 \LL_\w^{-1}\CC^{T}\LL_\z^{-T} & \GGamma_\w
 \end{pmatrix}\in\mathbb{R}^{n_1+n_2},
 \end{align}
which, by the following corollary is a proper covariance matrix, i.e. $\GGamma\in\mathbb{S}^+_{n}$. That is, we now present our main results.
\begin{corollary}[Valid joint covariance with preserved marginals]\label{Thm2} Let $\GGamma_{\z}=(\LL_{\z}^T\LL_{\z})^{-1}\in\mathbb{S}^+_{n_1}$ and $\GGamma_{\w}=(\LL_{\w}^T\LL_{\w})^{-1}\in\mathbb{S}^+_{n_2}$, and suppose $\CC\in\mathbb{A}^{n_1\times n_2}$. Then, taking $n=n_1+n_2$,  we have
\begin{align}\label{eq: CovProof}
\GGamma=\begin{pmatrix}
\GGamma_\z & \GGamma_{\z\w}\\
 \GGamma_{\w\z} & \GGamma_\w
 \end{pmatrix}=\begin{pmatrix}
\GGamma_\z & \LL_\z^{-1}\CC\LL_\w^{-T}\\
 \LL_\w^{-1}\CC^{T}\LL_\z^{-T} & \GGamma_\w
 \end{pmatrix}\in\mathbb{S}^+_n,
 \end{align}
 that is, $\GGamma$ is a proper covariance matrix. 
\end{corollary}

\begin{proof} The proof follows immediately by setting $H_1=\GGamma_\z$, $H_2=\GGamma_\w$, and $K=\CC$ in Theorem~\ref{thm:contract}.
\end{proof}

\begin{rmk}
The off-diagonal blocks of the proposed joint covariance matrix can equivalently
be written in terms of prior covariance square roots, i.e.,  as $\GGamma_\z^{1/2}\CC\GGamma_\w^{1/2}$ and $\GGamma_\w^{1/2}\CC^T\GGamma_\z^{1/2}$ 
when the principal square root factorisation is used. However, Bayesian inverse problems are often formulated in terms of (square root factors of) prior precision operators and we thus use the factorisations given in~(\ref{eq: CovProof}).
\end{rmk}

Corollary~\ref{Thm2} states $\GGamma$ is a proper covariance matrix, it does not say anything about how well the desired correlation is encoded. In fact, the choice of $\GGamma$ above is in general {\em not} optimal in terms of enforcing the (Pearson's) correlation point-wise, i.e., there may exist other choices of $\GGamma_{\z\w}$ such that the correlation matrix $\mathbb{R}^{n_1\times n_2}\ni \Phi:=\diag(\GGamma_\z)^{-\frac{1}{2}}\GGamma_{\z\w} \diag(\GGamma_\w)^{-\frac{1}{2}}$ is closer to $\CC$ in a particular sense.  However, when the principal square root factorisation is used, our choice is optimal in the sense of canonical correlation. Recall, canonical correlation can be interpreted as a multivariate analogue of Pearson's correlation. Specifically, in canonical correlation analysis (CCA) the goal is to find $v_1\in\mathbb{R}^{n_1}$ and $v_2\in\mathbb{R}^{n_2}$ to maximise
\[\rho=\frac{v_1^T\GGamma_{\z\w} v_2}{\sqrt{v_1^T\GGamma_{\z} v_1}\sqrt{v_2^T\GGamma_{\w} v_2}},\]
and then repeat under orthogonality to produce the sequence $1\geq \rho_1\geq\rho_2,\dots,\rho_{\min\{n_1,n_2\}}\geq 0$, which are termed the canonical correlations and are the singular values of the (canonical correlation) matrix $\GGamma_{\z}^{-\frac{1}{2}}\GGamma_{\z\w}\GGamma_{\w}^{-\frac{1}{2}}$ with $\GGamma_{\z}^{-\frac{1}{2}}$ and $\GGamma_{\w}^{-\frac{1}{2}}$ the inverse principal square roots (for more on canonical correlation see, for example, \cite[Chapter 12]{anderson2003introduction} or \cite{hotelling1936relations}). We then have the following theorem on the optimality of $\GGamma$.

\begin{theorem}[Optimality of $\GGamma$ in canonical correlation]\label{ThmOpt}
    Let $\GGamma_{\z}\in\mathbb{S}^+_{n_1}$ and $\GGamma_{\w}\in\mathbb{S}^+_{n_2}$, and suppose $\CC\in\mathbb{A}^{n_1\times n_2}$. Then, letting $n=n_1+n_2$, the unique solution to the optimisation problem:
    \begin{align}\label{eq:Copt}
    \min_{\Gamma_{\z\w}\in\mathbb{R}^{n_1\times n_2}}\;\norm[F]{\Gamma_\z^{-\frac{1}{2}}\Gamma_{\z\w}\Gamma_\w^{-\frac{1}{2}}-\CC}^2\quad \text{s.t.}\quad \GGamma=\begin{pmatrix}
\GGamma_\z & \GGamma_{\z\w}\\
 \GGamma_{\w\z} & \GGamma_\w
 \end{pmatrix}\in\mathbb{S}^+_{n},\end{align}
 is $\Gamma_{\z\w}=\Gamma_\z^{\frac{1}{2}}\CC\GGamma_\w^{\frac{1}{2}}$, i.e., we set $\LL_\z^{-1}$ and $\LL_\w^{-1}$ as the principal square roots of $\Gamma_{\z}$ and $\Gamma_{\w}$, respectively.
 
\end{theorem}
\begin{proof} Recall, by Theorem~\ref{thm:contract} we have $\GGamma\in\mathbb{S}^+_{n}$ if and only if $\GGamma_{\z\w}=\LL_\z^{-1}K\LL_\w^{-T}$ for some contraction $K\in\mathbb{A}^{n_1\times n_2}$. Substituting this in, we can rewrite the optimisation problem as:
\[\min_{\Gamma_{\z\w}\in\mathbb{R}^{n_1\times n_2}}\;\norm[F]{\Gamma_\z^{-\frac{1}{2}}(\LL_\z^{-1}K\LL_\w^{-T})\Gamma_\w^{-\frac{1}{2}}-\CC}^2\quad \text{s.t.}\quad \GGamma=\begin{pmatrix}
\GGamma_\z & \GGamma_{\z\w}\\
 \GGamma_{\w\z} & \GGamma_\w
 \end{pmatrix}\in\mathbb{S}^+_{n}.\]
Taking $\LL_{\z}=\Gamma_\z^{-\frac{1}{2}}$ and $\LL_{\w}=\Gamma_\w^{-\frac{1}{2}}$, the inverse principal square roots, we can rewrite the optimisation problem in terms of the contraction $K$: 
 \begin{align}
    \min_{K\in\mathbb{A}^{n_1\times n_2}}\;\norm[F]{K-\CC}^2\quad \text{s.t.}\quad \GGamma=\begin{pmatrix}
\GGamma_\z & \GGamma_{\z\w}\\
 \GGamma_{\w\z} & \GGamma_\w
 \end{pmatrix}\in\mathbb{S}^+_{n},\quad \GGamma_{\z\w}=\LL_\z^{-1}K\LL_\w^{-T},\end{align}
 which has a minimiser of $K=C$. Furthermore, $K=C$ is the unique minimiser as the squared Frobenius norm is strictly convex.
\end{proof}

\begin{rmk}\label{rem:facts}
The optimality statement in Theorem ~\ref{ThmOpt} is tied to the standard definition of
the canonical correlation matrix, which uses the inverse principal square roots
$\Gamma_p^{-1/2}$ and $\Gamma_m^{-1/2}$. If a different choice of whitening
filters is used, for example Cholesky factors, then an analogous statement can
be obtained only after redefining the objective in terms of those whitening
filters. More precisely, for fixed whitening filters $L_p$ and $L_m$ satisfying
$\Gamma_p=(L_p^T L_p)^{-1}$ and $\Gamma_m=(L_m^T L_m)^{-1}$, the construction
$\Gamma_{pm}=L_p^{-1}CL_m^{-T}$ is optimal for matching $C$ in the transformed
coordinates $L_p\Gamma_{pm}L_m^T$.
Thus different factorisations lead to different interpretations of the matrix
$C$, even though they all produce valid joint covariance matrices. In Section~\ref{sec:factor} we provide additional motivation for choosing the principal square root.
\end{rmk}

\subsection{Generating samples}\label{sec: samps}
Here we show how to efficiently draw samples from the joint distribution of $\pi(\z,\w)$. To this end, 
we present the following theorem showing how samples can be drawn from $\pi(\z,\w)$ and how the joint whitening filter can be constructed efficiently.
\begin{theorem}\label{ThmWhite} 
For $\z_\ast\in\mathbb{R}^{n_1}$ and $\w_\ast\in\mathbb{R}^{n_2}$, let $\z\sim\mathcal{N}(\z_\ast, \GGamma_{\z})$ and  $\w\sim\mathcal{N}(\w_\ast, \GGamma_{\w})$ where $\GGamma_{\z}=(\LL_{\z}^T\LL_{\z})^{-1}\in\mathbb{S}^+_{n_1}$ and $\GGamma_{\w}=(\LL_{\w}^T\LL_{\w})^{-1}\in\mathbb{S}^+_{n_2}$, and suppose $\CC\in\mathbb{A}^{n_1\times n_2}$. Then, if $s=(\z,\w)$ has a (joint) density of the form $\pi(s)=\mathcal{N}(s_\ast,\GGamma)$ with $s\in\mathbb{R}^n$ and $\GGamma\in\mathbb{S}^+_n$ as in Theorem~\ref{Thm2} and $\eta\sim\mathcal{N}(0,I)$, samples from the joint density can be computed by 
\begin{align}\label{eq: taking draw}
\tilde{s}
=
\begin{pmatrix}\LL_\z^{-1} & 0\\ \LL_\w^{-1}\CC^T & \LL_\w^{-1}\DD
\end{pmatrix}\eta
+
s_\ast,
\end{align} 
where the matrix $\DD\in\mathbb{R}^{n_2\times n_2}$ satisfies $\DD\DD^T=I-\CC^T\CC$ (i.e., $\DD$ is the defect operator of $\CC$). Moreover,
\begin{align}\label{eq: bigWhite}
\LL=
\begin{pmatrix}\LL_\z & 0\\ -\DD^{-1}\CC^T\LL_\z & \DD^{-1}\LL_\w
\end{pmatrix}
\end{align} 
satisfies $\GGamma=(\LL^T\LL)^{-1}$.
\end{theorem}
\begin{proof}
First, to prove (\ref{eq: taking draw}) it is clear that $\mathbb{E}[\tilde{s}]=s_\ast$. For the covariance, assuming without loss of generality that $s_\ast=0$, we have
\begin{align}
\mathbb{E}\left[\tilde{s}\tilde{s}^T
 \right]&=\mathbb{E}\left[\begin{pmatrix}\LL_\z^{-1} & 0\\ \LL_\w^{-1}\CC^T & \LL_\w^{-1}\DD
\end{pmatrix}
\eta \eta^T
\begin{pmatrix}\LL_\z^{-T} & \CC\LL_\w^{-T}\\  0 & \DD^T\LL_\w^{-T}
\end{pmatrix}\right]\nonumber\\
&=\begin{pmatrix}\LL_\z^{-1} & 0\\ \LL_\w^{-1}\CC^T & \LL_\w^{-1}\DD
\end{pmatrix}
\begin{pmatrix}\LL_\z^{-T} & \CC\LL_\w^{-T}\\  0 & \DD^T\LL_\w^{-T}
\end{pmatrix}\nonumber\\
&=\begin{pmatrix}\GGamma_\z & \LL_\z^{-1}\CC\LL_\w^{-T}\\ \LL_\w^{-1}\CC^T\LL_\z^{-T} & \LL_\w^{-1}(\CC^T\CC+\DD\DD^T) \LL_\w^{-T}
\end{pmatrix}=\GGamma,\nonumber
\end{align}
since $\mathbb{E}[\eta\eta^T]=I$ and $\CC^T\CC+\DD\DD^T=I$.

On the other hand, to prove (\ref{eq: bigWhite}), since $\CC\in\mathbb{A}^{n_1\times n_2}$, we have $I-\CC^T\CC\in\mathbb{S}^+_{n_2}$, implying $\DD$ is invertible, and thus
\begin{align}
\left(\LL^T\LL\right)\GGamma&=
\begin{pmatrix}\LL_\z & -\LL_\z\CC\DD^{-T}\\ 0&  \LL_\w\DD^{-T}
\end{pmatrix}
\begin{pmatrix}\LL_\z & 0\\ -\DD^{-1}\CC^T\LL_\z & \DD^{-1}\LL_\w
\end{pmatrix}
\begin{pmatrix}
\GGamma_\z & \LL_\z^{-1}\CC\LL_\w^{-T}\\
 \LL_\w^{-1}\CC^{T}\LL_\z^{-T} & \GGamma_\w
 \end{pmatrix}\nonumber\\ 
 &=\begin{pmatrix}\LL_\z & -\LL_\z\CC\DD^{-T}\\ 0&  \LL_\w\DD^{-T}
\end{pmatrix}
\begin{pmatrix}\LL_\z^{-1} & \CC\LL_\w^{-1}\\ 0 &  \DD^T\LL_\w^{-1}
\end{pmatrix}=I,\nonumber
\end{align}
concluding the proof.
\end{proof}

\subsection{Sampling examples}\label{sec:samp} Here we provide two examples of generating samples from the joint density of two distributed parameters given the associated marginal densities and a desired correlation structure. 

\paragraph{Example 1}
For the first example we assume both parameters are distributed over the rectangle $\Omega=[0,2]\times[0,1]$. We discretise the rectangle using the same regular lattice for each parameter, leading to $\z,\w\in\mathbb{R}^{5000}$. The marginal priors are set as $\pi(\z)=\mathcal{N}(0,\GGamma_\z)$ and $\pi(\w)=\mathcal{N}(0,\GGamma_\w)$. The marginal covariance matrix $\GGamma_\z$ is defined using a PDE-based covariance matrix \cite{RoininenHuttunenLasanen14,DaonStadler18,LindgrenRueLindstrom11} coming from a finite element discretisation of an elliptic PDE using linear Lagrange basis functions. More specifically, we set  
\begin{align}\label{eq: pdeprior}
\GGamma_\z=(a_1K+a_2M+a_3B)^{-2}, 
\end{align} 
with $K$, $M$, and $B$ the stiffness, mass, and boundary mass matrices, respectively, given by
\begin{align}
K_{ij}=\int_\Omega \Theta\nabla\phi_i\cdot\nabla\phi_j\;{\rm d}x,\quad M_{ij}=\int_\Omega \phi_i\phi_j\;{\rm d}x,\quad \quad B_{ij}=\int_{\partial\Omega} \phi_i\phi_j\;{\rm d}x,
\end{align}
where $\left\{\phi(x)\right\}_{i=1}^N$ denotes the linear Lagrange basis functions, and $i,j\in\{1,2,\dots,n_1\}$ where $n_1=5000$. The parameters $a_1,a_2>0$ control the desired correlation length and variance, while $\Theta\in\mathbb{S}^+_2$ can be used to set anisotropy (i.e., different correlation lengths in different directions) and $a_3\propto\sqrt{a_1a_2}$ is used to mitigate so-called boundary effects, see for example \cite{RoininenHuttunenLasanen14,DaonStadler18,khristenko2019analysis}. Here we set $a_1=4\times10^{-2}$, $a_2=1$, $a_3=0.125$, and $\Theta=I$.

On the other hand, the marginal covariance matrix $\GGamma_\w$ is taken as the exponential squared covariance matrix \cite{RasmussenWilliams05}, i.e., 
\begin{align}\label{eq: exp2prior}
\GGamma_\w(r)=\exp\left\{-\frac{1}{2}\left(\frac{r}{\ell}\right)^2\right\}, \quad r=\norm{x},
\end{align} 
with $\ell=0.2$ the correlation length. Finally, for the first example we consider two different correlation structures: (a) we attempt to encode homogeneous strong positive correlation throughout the domain by setting $C=0.999I$, and (b) we attempt to encode strong positive correlation on the left half of the domain strong negative correlation right half of the domain, i.e., we set $C=c(x)I$, where $c(x)=0.999$ for $x\leq1$ and $c(x)=-0.999$ for $x>1$. It is worth noting that (in the function space setting) the exponential squared covariance operator (\ref{eq: exp2prior}) generates smooth (i.e., infinitely differentiable) samples, while the PDE-based covariance operator (\ref{eq: pdeprior}) generates functions in $H^1(\Omega)$, that is to say, the samples of $\w$ should be significantly smoother than the samples of $\z$.

Three samples from the joint densities for cases (a) and (b) are shown in Figure \ref{fig: joint2dsamp}. The approach works as intended: the samples of the two parameters for Example 1 (a) appear to be strongly positively correlated throughout the domain, while the samples of the two parameters for Example 1 (b) appear to strongly positively correlated in the left half of the domain and strongly negatively correlated in the right half of the domain. We also see the difference in smoothness of the two parameters due to the difference in the associated marginal prior covariance matrices (see (\ref{eq: pdeprior}) and (\ref{eq: exp2prior})).

\begin{figure}[t!]
\centering

$\z \text{ (Example 1 (a) and 1 (b))}$ \hrulefill

\vspace{.1cm}

\includegraphics[height=0.15\textwidth]{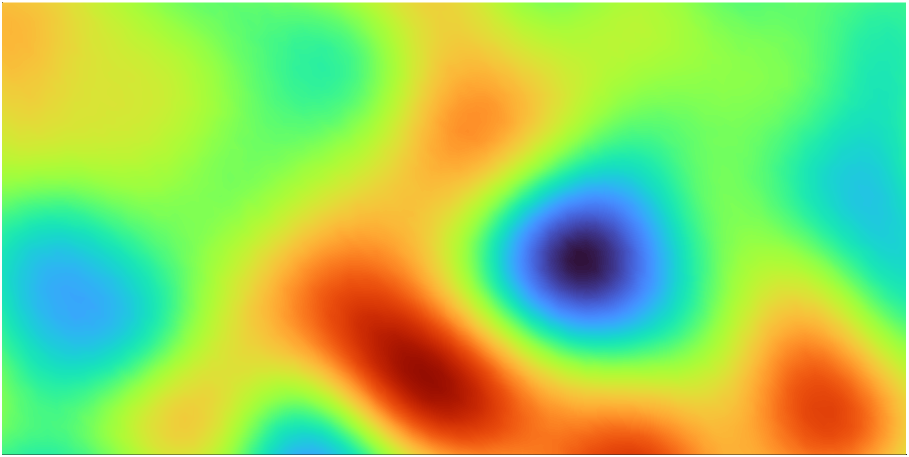}
\hfill
\includegraphics[height=0.15\textwidth]{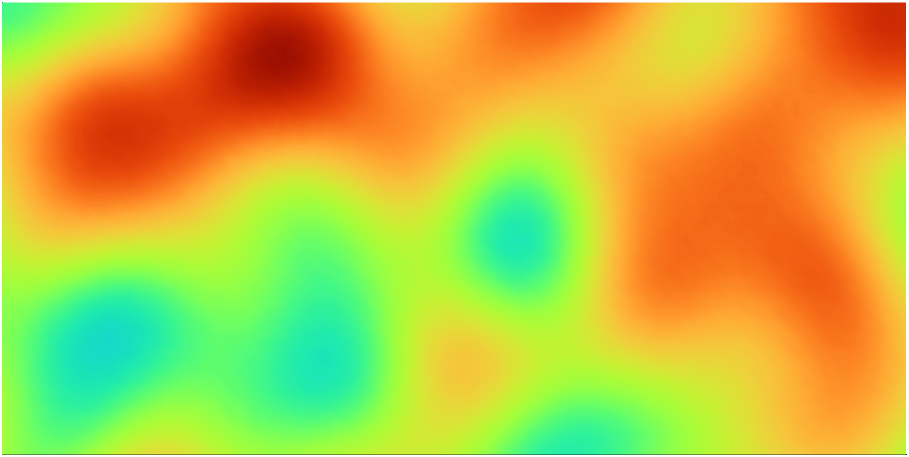}
\hfill
\includegraphics[height=0.15\textwidth]{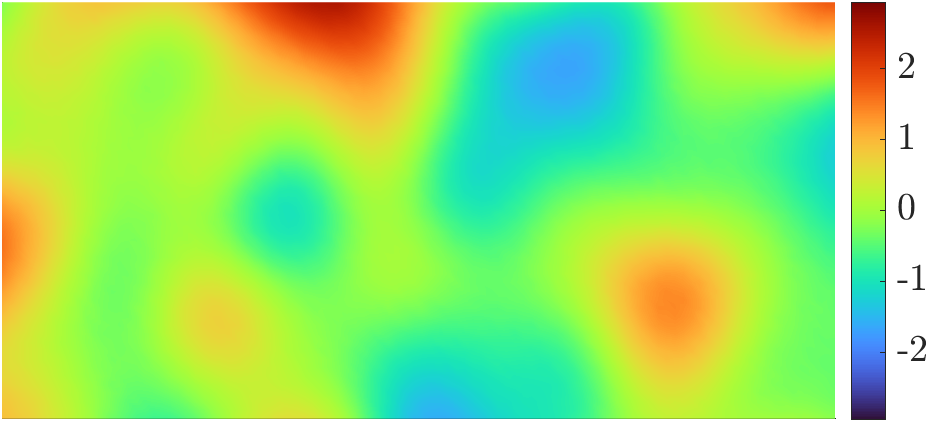}\\

$\w \text{ (Example 1 (a))}$ \hrulefill
\vspace{.1cm}

\includegraphics[height=0.15\textwidth]{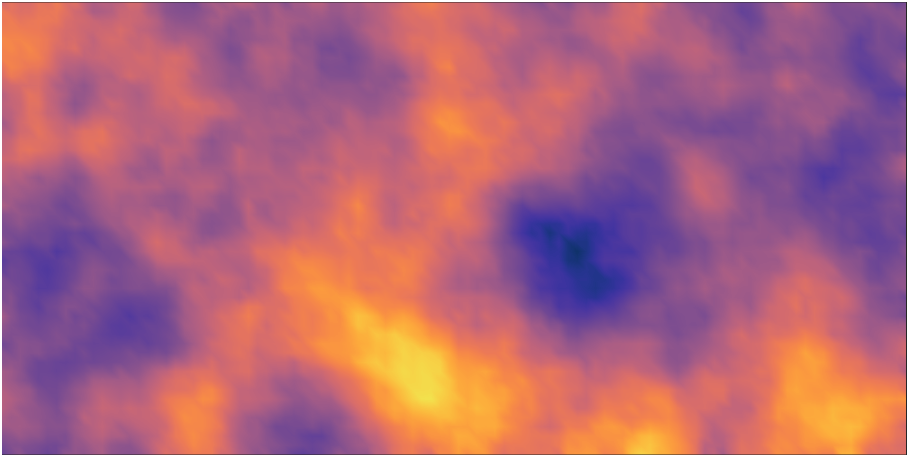}
\hfill
\includegraphics[height=0.15\textwidth]{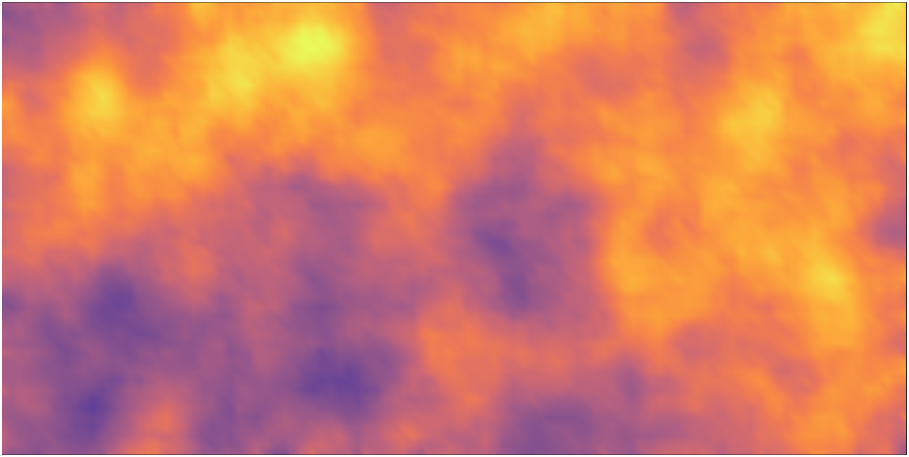}
\hfill
\includegraphics[height=0.15\textwidth]{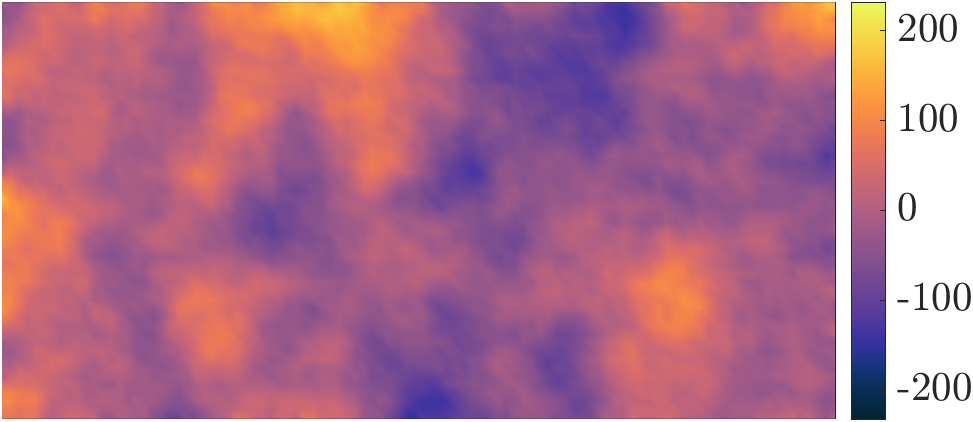}\\

$\w \text{ (Example 1 (b))}$ \hrulefill
\vspace{.1cm}

\includegraphics[height=0.15\textwidth]{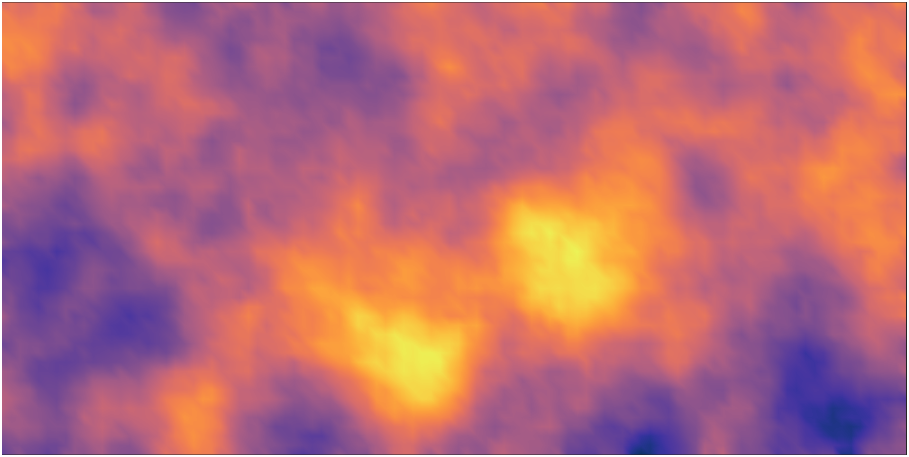}
\hfill
\includegraphics[height=0.15\textwidth]{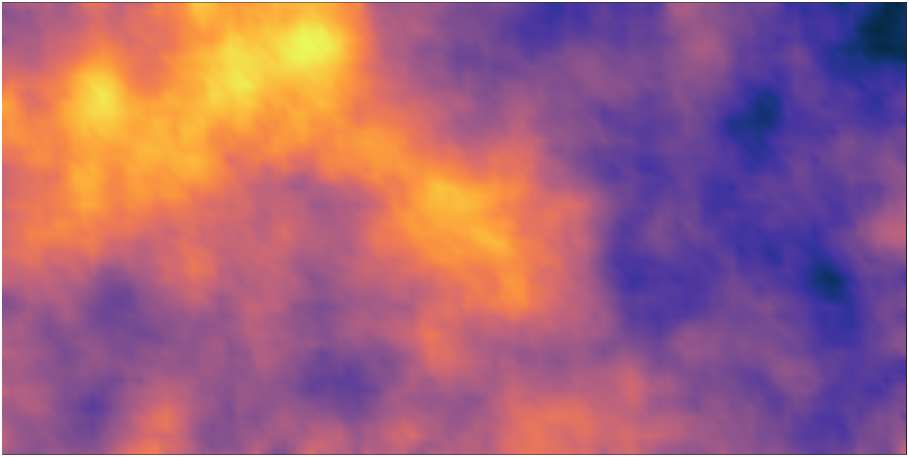}
\hfill
\includegraphics[height=0.15\textwidth]{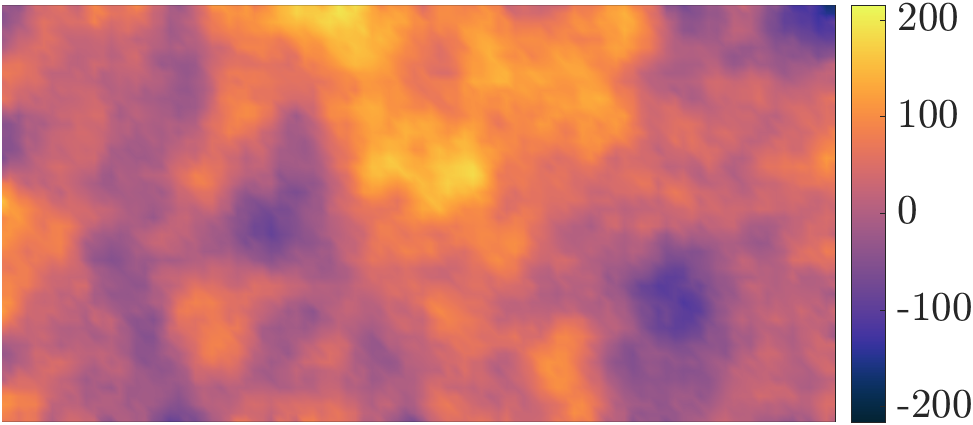}

 \caption{Samples from the joint priors in Example 1. The first row shows three samples of $\z$, the second row shows three corresponding samples of $\w$, where the desired correlation structure is encoded using $c(x)=0.999$ (Example 1 (a)), and the third row shows three corresponding samples of $\w$, where the desired correlation is encoded using $c(x)=0.999$ for $x\leq1$ and $c(x)=-0.999$ for $x>1$ (Example 1 (b)). As expected, in the first case (Example 1 (a)), the two parameters appear to be strongly positively correlated throughout the domain, while in the second case (Example 1 (b)), the samples of the two parameters appear to be strongly positively correlated in the left half of the domain and strongly negatively correlated in the right half of the domain.}\label{fig: joint2dsamp}
\end{figure}

\paragraph{Example 2}
In the second example we consider a case in which one parameter is distributed over the rectangle $\Omega=[0,2]\times[0,1]$ while the other is distributed over the bottom boundary $B_b=[0,2]\times\{0\}$. We discretise the rectangle  with 10000 nodes, of which 100 lie on $B_b$. As such $\z\in\mathbb{R}^{10000}$ and $\w\in\mathbb{R}^{100}$. In this example we again take $\Gamma_\z$ to be of the form (\ref{eq: pdeprior}), though now we take $a_1=a_2=1$, $a_3=0.125$, and impose an anisotropic smoothness structure by setting $\Theta=\diag(1,0.025)\in\mathbb{R}^{2 \times 2}$. We set $\GGamma_\w$ again as the exponential squared covariance matrix (see (\ref{eq: exp2prior})), but use a correlation length of $\ell=0.1$.
The desired correlation between the restriction of $\z$ to the boundary $B_b$ and $\w$ is set at 0.999 homogeneously. This is done by taking $\CC\in\mathbb{R}^{10000\times 100}$ with $\CC_{ij}=0$ unless $i$ and $j$ correspond to the same boundary node, in which case $\CC_{ij}=0.999$. 

In Figure \ref{fig: jointMdsamp} we show three samples of $\z$ and $\w$ (along with the restriction of $\z$ to the $B_b$, i.e., $\left.\z\right|_{B_b}$) drawn from the joint density $\pi(\z,\w)$. As in the previous example, we see that the approach works as intended, with $\left.\z\right|_{B_b}$ and $\w$ being positively correlated across $B_b$. Moreover, the anisotropic smoothness of $\z$ and higher order smoothness of $\w$ are evident.

\begin{figure}[t!]
\centering
$\z$ \hrulefill

\vspace{.1cm}

\includegraphics[height=0.155\textwidth]{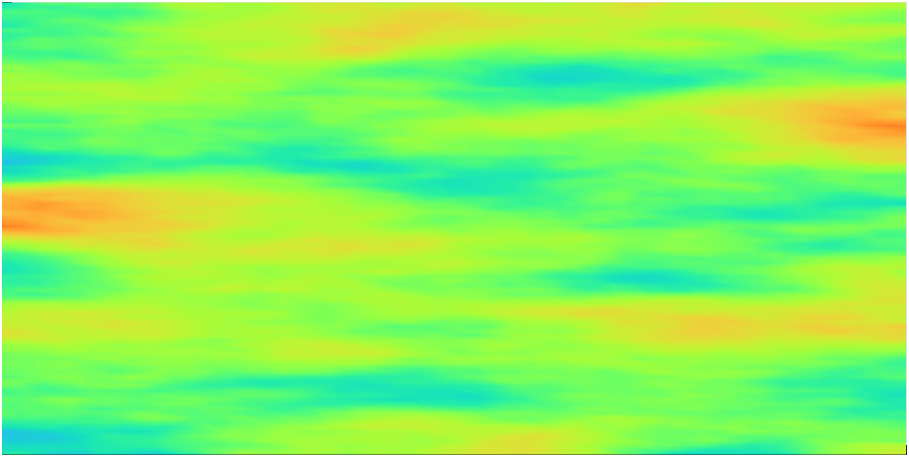}
\hfill
\includegraphics[height=0.155\textwidth]{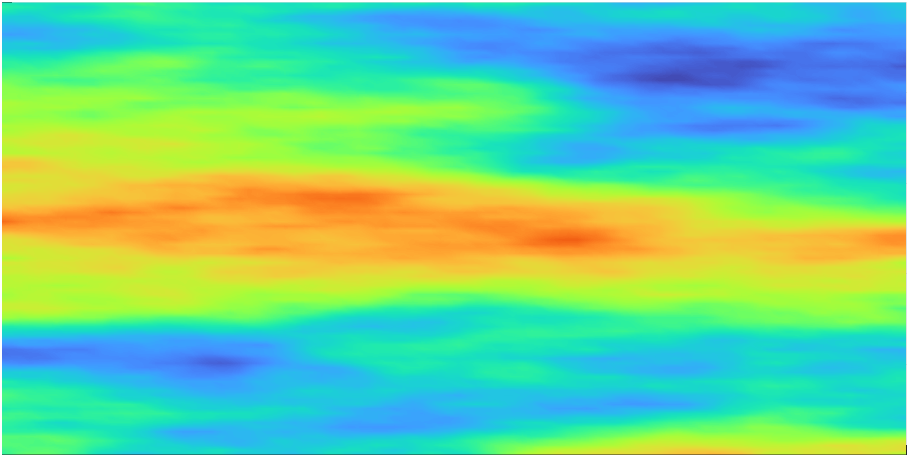}
\hfill
\includegraphics[height=0.155\textwidth]{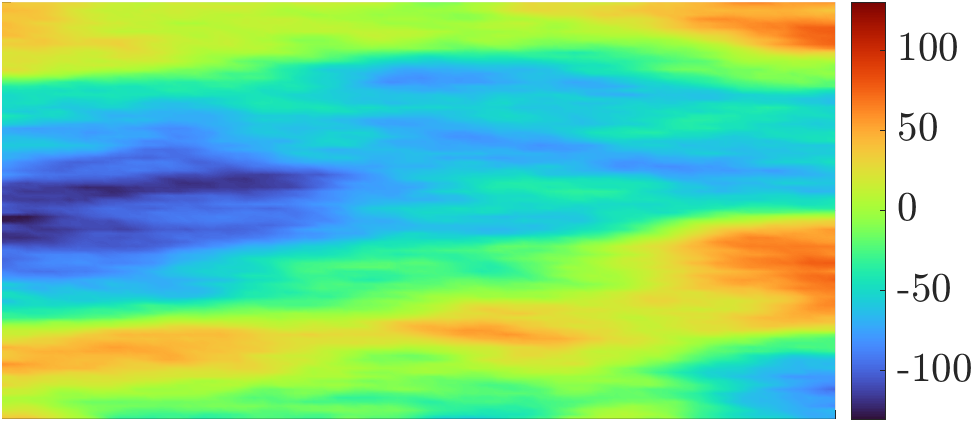}\\

$\w$ and $\left.\z\right|_{B_b}$ \hrulefill
\vspace{.1cm}

\includegraphics[height=0.155\textwidth]{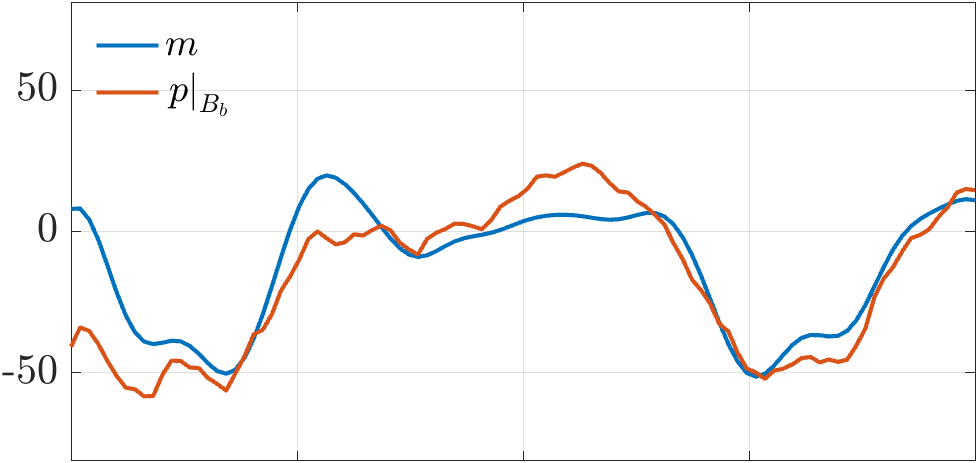}
\hfill
\includegraphics[height=0.155\textwidth]{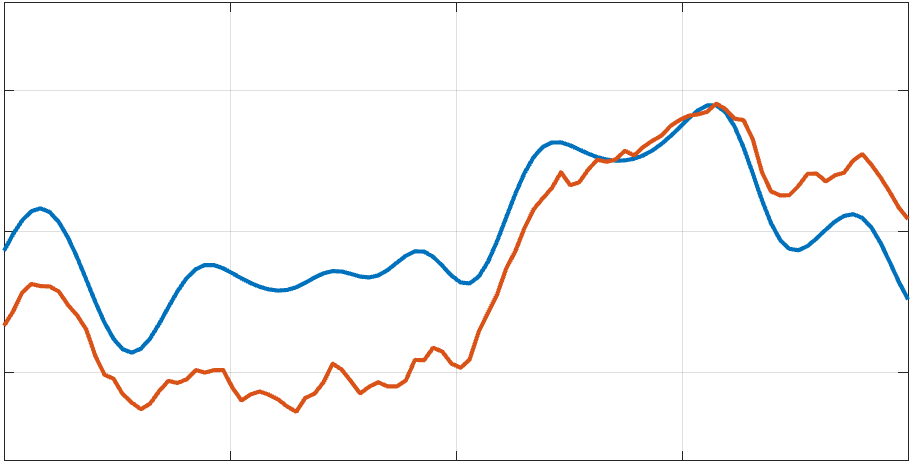}
\hfill
\includegraphics[height=0.155\textwidth]{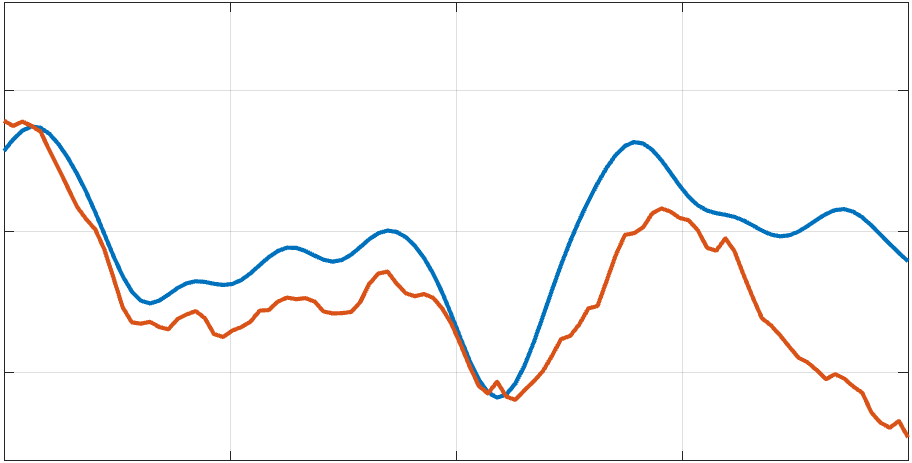}
 \caption{Three samples from the joint priors in Example 2. In the top row we show three samples of $\z$, while on the bottom row we show the three corresponding samples of $\w$ in red as well as the restriction of $\z$ to the boundary, i.e., $\left.\z\right|_{B_{b}}$, in blue. The desired correlation structure is encoded using $c(x)=0.999$ for $x\in B_{b}$. }\label{fig: jointMdsamp}
\end{figure}

\subsection{Choice of factorisation}\label{sec:factor}
The proposed joint covariance matrix $\GGamma$ given in (\ref{eq: CovOut}) is a proper covariance matrix independent of the choice of factorisation used, that is, $\GGamma\in\mathbb{S}^+_n$ for any $\LL_\z^{-1}$ and $\LL_\w^{-1}$ that satisfy $\GGamma_{\z}=(\LL_{\z}^T\LL_{\z})^{-1}$ and $\GGamma_{\w}=(\LL_{\w}^T\LL_{\w})^{-1}$.  However, as pointed out in Remark~\ref{rem:facts}, for different choices of whitening filters, the resulting cross-correlation structure is different. To see this, consider two new factors, say $\hat{\LL}_\z$ and $\hat{\LL}_\w$ satisfying $(\hat{\LL}_{\z}^T\hat{\LL}_{\z})^{-1}=(\LL_{\z}^T\LL_{\z})^{-1}=\GGamma_{\z}$ and $(\hat{\LL}_{\w}^T\hat{\LL}_{\w})^{-1}=(\LL_{\w}^T\LL_{\w})^{-1}=\GGamma_{\w}$, but where $\LL_{\z}\neq\hat{\LL}_{\z}$, and $\LL_\w\neq\hat{\LL}_\w$. Then in general,
\begin{align}
\LL_{\z}^{-1}\CC\LL_{\w}^{-T}\neq \hat{\LL}_{\z}^{-1}\CC\hat{\LL}_{\w}^{-T},
\end{align} 
which leads to differences in the joint statistics, in particular, the actual resulting cross-correlation. 

There is an infinite number of choices for each of the whitening filters. The most frequently used whitening filters in PDE-constrained inference are the (upper or lower) Cholesky factor, and the principal square root, both of which are unique~\cite{golub2013matrix}.  However, as we demonstrate below, use of the Cholesky factor typically induces some possibly unwanted structure into the cross-correlation matrix $\Phi$, due to its triangular (causal) form. We therefore would generally advocate for using the principal square root (if computationally feasible). 

To shed some light on the effects of using the different whitening filter we consider a one dimensional example where the two parameters $\z$ and $\w$ are defined over a regular discretisation of $\Omega=[0,1]$. Suppose further that the marginal densities are identical, i.e., $\z\sim\mathcal{N}(0,\GGamma_\z)$ and $\w\sim\mathcal{N}(0,\GGamma_\w)$ with $\GGamma_\z=\GGamma_\w$ given by the exponential squared covariance matrix (see (\ref{eq: exp2prior})) with correlation length $\ell=0.1$.
We attempt to encode strong positive correlation in the left half of the domain and strong negative correlation in the right half of the domain by taking $c(x)=0.999$ for $x\leq\frac{1}{2}$ and $c(x)=-0.999$ for $x>\frac{1}{2}$. 

The resulting cross-correlation and several samples from the joint distributions are shown in Figure~\ref{fig: factorChoice} when taking the principal square root factorisation and when taking the Cholesky factorisation. As alluded to above, use of the Cholesky factor induces some (spatial) skewness towards the right in the cross-correlation matrix, while the  principal square root factor does not. In this example, this leads to samples from the joint distribution being positively correlated over most of the domain. On the other hand, samples from the joint distribution based on the principal square root are equally positively and negatively correlated over the domain.

It is also worth pointing out at this point, for many Bayesian inverse problems the prior covariance matrices used  are smoothing operators~\cite{RoininenHuttunenLasanen14,Stuart10,KaipioKolehmainenVauhkonenEtAl99,BuiGhattasMartinEtAl13}. Thus the left- and right-conditioning of $\CC$ by $\LL_\z^{-1}$ and $\LL_\w^{-1}$, respectively, tends to smooth the intended cross-correlation embedded in $\CC$, as seen in Figure \ref{fig: factorChoice}.


\begin{figure}[t!]
\centering

\begin{tikzpicture}

\node[inner sep=0pt] (a) at (0,.04)
{\includegraphics[width=0.3225\textwidth]{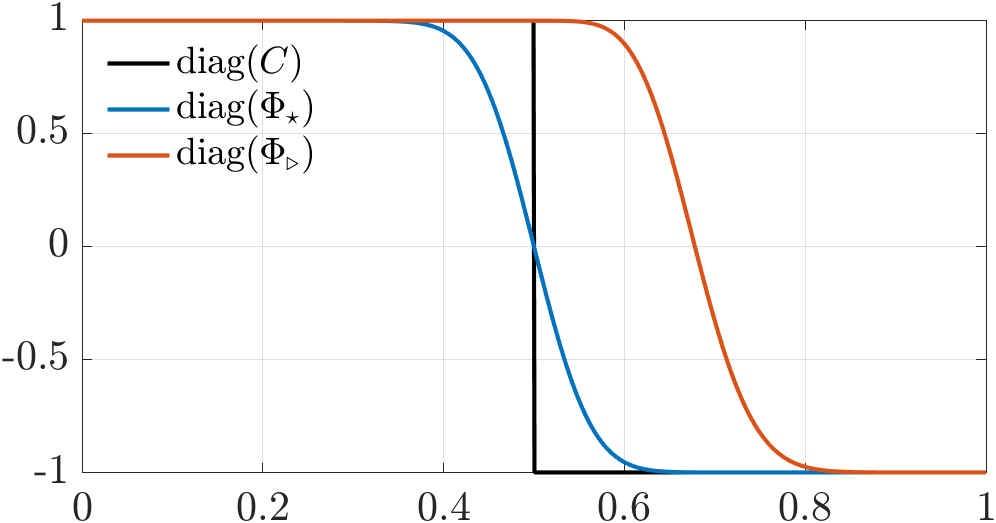}};

\node[inner sep=0pt] (a) at (5.5,0)
{\includegraphics[width=0.3\textwidth]{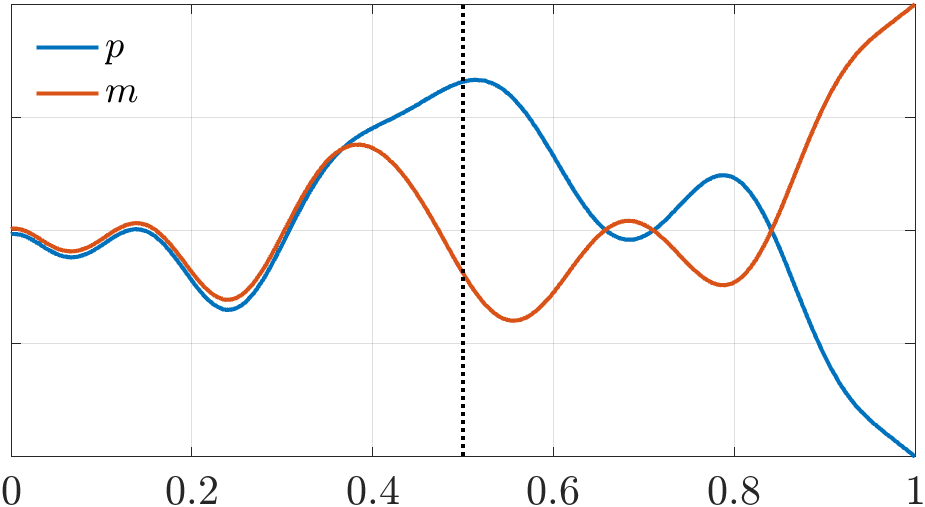}};

\node[inner sep=0pt] (a) at (11,0)
{\includegraphics[width=0.3\textwidth]{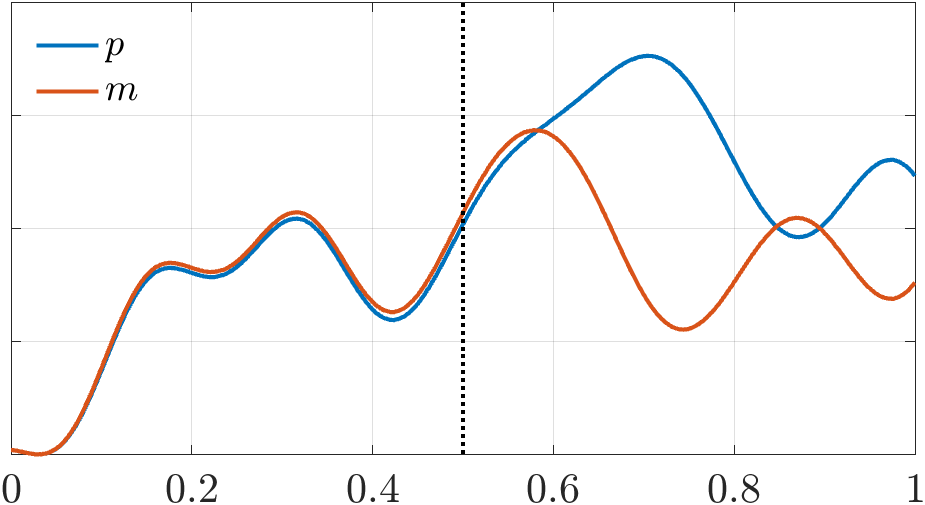}};

\node at (0,1.75) {Correlation};
\node at (5.5,1.75) {Principal square root};
\node at (11,1.75) {Cholesky factor};
\end{tikzpicture}    
 \caption{Demonstration of the effects of different factorisations of $\Gamma_\z$ and $\Gamma_\w$. On the left we show the desired correlation structure encoded in $\diag(C)$ (in black), along with the actual resulting diagonal of the correlation matrices $\Phi_\star$ (in blue) and $\Phi_\triangleright$ (in red) found  using the principal square root and the Cholesky factor, respectively. In the centre and on the right we show samples from the joint prior of $\z$ (in blue) and $\w$ (in red) using the principal square root factor and Cholesky factor, respectively.}
\label{fig: factorChoice}\end{figure}
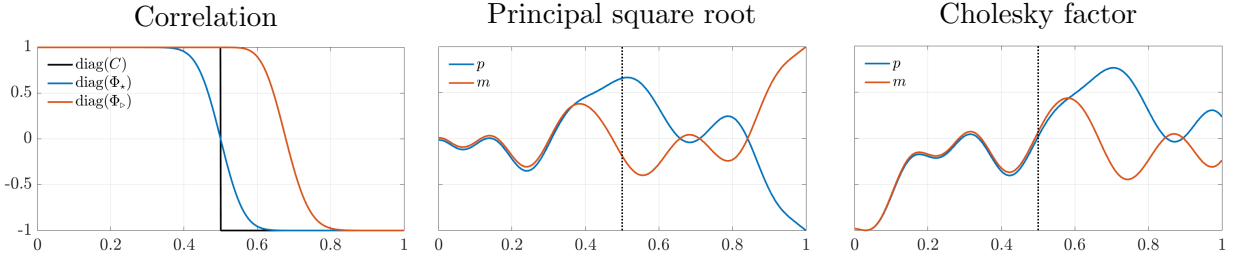

\section{Ignorance isn't always bliss}\label{sec: Ignorance}
If there is no prior knowledge of any relationship  (correlation or otherwise) between parameters, then the default approach is to treat them as independent, see e.g.,~\cite{recinos2023framework,petra2012inexact,NicholsonNiskanen21,hanninen2020application} among many. In the case of a joint Gaussian prior this amounts to setting the cross-correlation to 0, i.e., taking the prior covariance matrix as block diagonal. However, in the Bayesian framework, setting the correlation to 0 implies the correlation is {\em known} a priori to be 0. This can at best result in overly optimistic uncertainty estimates, and at worst can result in significant overconfidence in heavily biased estimates.

To help remedy this situation one should, at least attempt to, model the uncertainty in the correlation. This can be done by taking $\CC$ as a random variable with an associated prior density $\pi(\CC)$. If no prior beliefs on the cross-correlation are held, a natural starting point is the uniform prior over contractions, i.e., $\pi(\CC)=\mathcal{U}(\mathbb{A}^{n_1\times n_2})$. Independent of the choice of prior taken on $\CC$, the full posterior can be written as
\begin{align}\label{eq: fullpost}
\pi(\z,\w,\CC\vert\data)\propto\pi(\data\vert\z,\w,\CC)\pi(\z,\w,\CC)=\pi(\data\vert\z,\w)\pi(\z,\w\vert\CC)\pi(\CC),
\end{align}
since the likelihood is conditionally independent of $\CC$.

Incorporating cross-correlation into the prior through the proposed methodology means the joint prior of $\pi(\z,\w,\CC)$ can be decomposed in a straightforward manner. Specifically, taking $s=(\z,\w)\in\mathbb{R}^n$ we have 
\begin{align}
\pi(s,\CC)&=\pi(s\vert\CC)\pi(\CC)\\
&\propto\frac{1}{\vert\GGamma\vert^{\frac{1}{2}}}\exp\left\{-\frac{1}{2}\norm[\GGamma^{-1}]{s-s_\ast}^2\right\}\pi(\CC)\\
&=\exp\left\{-\frac{1}{2}\left(\norm[\GGamma^{-1}]{s-s_\ast}^2+\ln{\vert\GGamma\vert}\right)\right\}\pi(\CC)\\
&\propto\exp\left\{-\frac{1}{2}\left(\norm[\GGamma^{-1}]{s-s_\ast}^2+\ln{\vert I-\CC\CC^T\vert}\right)\right\}\pi(\CC),
\end{align}
where the last line follows from the fact that 
\begin{align}
\vert\GGamma\vert=\vert \GGamma_\z-(\LL_\z^{-1}\CC\LL_\w^{-T})\GGamma_\w^{-1}(\LL_\w^{-1}\CC^T\LL_\z^{-T})  \vert \cdot\vert\GGamma_\w\vert=\vert\LL_\z^{-1}\vert \cdot \vert I-\CC\CC^T\vert \cdot\vert\LL_\z^{-T}\vert \cdot\vert\GGamma_\w\vert\propto \vert I-\CC\CC^T\vert.\nonumber
\end{align}
It is worth noting that for cases in which $n_1>n_2$, for computational ease, 
Sylvester’s determinant theorem can be used, i.e.,  $\vert I-\CC\CC^T\vert=\vert I-\CC^T\CC\vert$. Furthermore, if $\CC$ is diagonal we have $\vert I-\CC\CC^T\vert=\prod_{i=1}^{\min\{n_1,n_2\}}(1-C_{ii}^2)$.

\subsection{Computational considerations}\label{sec: CompCon}
We will be particularly interested in accurately characterising the full posterior using Markov chain Monte Carlo (MCMC) methods (Section \ref{sec: ID} provides motivation for this). To define the likelihood used throughout the numerical experiments, we consider the additive error model, i.e., a relationship between the data and parameters of the form in~(\ref{eq: genjoint}). We assume the observational errors are centred Gaussian and independent of the unknown parameters, i.e.,
$\e\sim\mathcal{N}(0,\Gamma_e)$ with $(\z,\w)\perp \e$.
The likelihood is then~\cite{KaipioSomersalo05}
\begin{align}
 \pi(\data|\z,\w)\propto\exp\left\{-\tfrac{1}{2}\norm[\Gamma_e^{-1}]{\data-\f(\z,\w)}^2\right\}.   
\end{align}
 
 For PDE-based inverse problems with higher-dimensional unknown parameters (the problems of interest here) the computational costs of characterising the posterior are dominated by repeated forward solves, while statistical efficiency is dominated by how well the sampler mixes. To address these challenges, we propose (i) a Metropolis-within-Gibbs (MwG) approach to the sampling to leverage conditional distributions available from the joint prior, and (ii)  a truncated Karhunen-Lo\`eve (KL) expansion of the parameters to reduce the dimensions.

\paragraph{An adaptive Metropolis-within-Gibbs sampler}
The MwG sampler is well-suited to Bayesian inverse problems involving hyper-parameters and has been widely applied to such problems, see for example, \cite{dunlop2017hierarchical,robert1999monte,agapiou2014analysis,roininen2019hyperpriors,bardsley2019metropolis}. Setting $s=(\z,\w)\in\mathbb{R}^n$, the approach is based on generating samples from the joint posterior $\pi(s,\CC|\data)$ by alternately updating $s$ conditional on $\CC$, and $\CC$ conditional on $s$. More precisely, starting from the current state
\[s^{(k)}\sim\pi(s|C^{(k)},\data)\propto\exp\{-\tfrac{1}{2}\alpha(s,C^{(k)})\},\]
where
\[\alpha(s,C)=\norm[\Gamma_e^{-1}]{\data-\f(s)}^2+\norm[\GGamma^{-1}(\CC)]{s-s_\ast}^2,\] 
and then
\[C^{(k+1)}\sim\pi(C|s^{(k)},\data)\propto\exp\{-\tfrac{1}{2}\beta(s^{(k)},C)\},\] 
where 
\[\beta(s,C)=\norm[\GGamma^{-1}(\CC)]{s-s_\ast}^2+\ln{\vert I-\CC\CC^T\vert}-2\ln{\pi(\CC)}.\]

In practice, when direct sampling from these conditional distributions is not available, the updates are carried out using Metropolis–Hastings steps. In the present work we restrict attention to Gaussian proposal distributions $q(\cdot|\cdot)$, although a range of more sophisticated proposal mechanisms and sampling strategies could also be employed. The resulting adaptive MwG scheme is summarised in Algorithm~\ref{alg: MwG}.

\begin{algorithm}[h!]
\caption{Adaptive Metropolis-within-Gibbs}\label{alg: MwG}
\begin{algorithmic}[1]
\State Choose an initial state $(s^{(0)},C^{(0)}) \in \mathbb{R}^n \times \mathcal{A}^{n_1\times n_2}$, and set $k=0$
\For{$k=0,1,2,\dots,M-1$}
    \State Propose $s' \sim q_s^{(k)}(\,\cdot\,|s^{(k)})$
    \State Compute
    \[
    a_s=\min\left\{1,\,
    \frac{\pi(s'|C^{(k)},d)\,q_s^{(k)}(s^{(k)}|s')}
         {\pi(s^{(k)}|C^{(k)},d)\,q_s^{(k)}(s'|s^{(k)})}
    \right\}
    \]
    \State Set $s^{(k+1)} = s'$ with probability $a_s$; otherwise set $s^{(k+1)}=s^{(k)}$
    
    \State Propose $C' \sim q_C^{(k)}(\,\cdot\,|C^{(k)})$
    \State Compute
    \[
    a_C=\min\left\{1,\,
    \frac{\pi(C'|s^{(k+1)},d)\,q_C^{(k)}(C^{(k)}|C')}
         {\pi(C^{(k)}|s^{(k+1)},d)\,q_C^{(k)}(C'|C^{(k)})}
    \right\}
    \]
    \State Set $C^{(k+1)} = C'$ with probability $a_C$; otherwise set $C^{(k+1)}=C^{(k)}$
    
    \State Update the proposal parameters defining $q_s^{(k)}$ and/or $q_C^{(k)}$ if adaptation is used
\EndFor
\end{algorithmic}
\end{algorithm}

\begin{rmk}\label{rm: MwG}
    The Metropolis-within-Gibbs approach is particularly well suited for cases in which the forward model is linear, i.e., $\f(s)=Gs$ for some $G\in\mathbb{R}^{q\times n}$, since in such cases $\pi(s|\CC,\data)$ is Gaussian. Consequently, updating $s$ can be done directly by sampling from $\pi(s|d,\CC)$, avoiding the need for an additional Metropolis step, i.e., the approach simplifies to exact Gibbs updates for $s$, with only the updates for $\CC$ requiring Metropolis-Hastings.
\end{rmk}
\paragraph{Dimension reduction based on truncated Karhunen-Lo\`eve expansion}

A standard approach used to reduce the parameter dimension in Bayesian inverse problems is to use a truncated KL expansion of the parameters~\cite{Jolliffe05,LipponenSeppanenKaipio13,peherstorfer2018survey}. First, we consider the eigenvalue decomposition of each of the prior covariance matrices;
\begin{align}\label{eq: covs}
\GGamma_\z=\sum_{i=1}^{n_1}\lambda_iv_iv_i^T=V\Lambda V^T,\quad \GGamma_\w=\sum_{j=1}^{n_2}\sigma_ju_ju_j^T=U\Sigma U^T.
\end{align} 
The truncated KL expansions of each of the parameters are then given by
\begin{align}\label{eq:TKL}
\tilde\z=\z_\ast+\sum_{i=1}^{k_\z}\sqrt{\lambda_i}\hat{\z}_iv_i,\quad \tilde\w=\w_\ast+\sum_{j=1}^{k_\w}\sqrt{\sigma_j}\hat{\w}_ju_j,
\end{align}
where we take $k_\z\ll n_1$ and $k_\w\ll n_2$. The goal of the joint inference problem is then to estimate
$\hat{\z}\in\mathbb{R}^{k_\z}$ and $\hat{\w}\in\mathbb{R}^{k_\w}$, which have associated prior distributions $\hat{\z}\sim\mathcal{N}(0,I)$ and $\hat{\w}\sim\mathcal{N}(0,I)$. Furthermore, for a given desired cross-correlation encoded in $C\in\mathbb{A}^{n\times m}$, the joint prior covariance for $\hat{\z}$ and $\hat{\w}$ is (using (\ref{eq: covs}))
\begin{align}
\hat{\Gamma}&=\begin{pmatrix}\Lambda^{-\frac{1}{2}} V^T & 0\\0 & \Sigma^{-\frac{1}{2}} U^T \end{pmatrix}\Gamma\begin{pmatrix}V\Lambda^{-\frac{1}{2}} & 0\\0 & U\Sigma^{-\frac{1}{2}} \end{pmatrix}=\begin{pmatrix}I & V^T\CC U\\U^T\CC^TV & I \end{pmatrix},\nonumber
\end{align} 
while the joint prior mean is $(\hat{\z}_\ast,\hat\w_\ast)=(0,0)$.

\subsection{Parameterisation, identifiability and estimability of $\CC$}\label{sec: ID}
In this section we briefly outline some important features related to the parameterisation, identifiability, and  estimability of the matrix $\CC$. 

\paragraph{Parameterisation and identifiability} First, we note that there are several standard choices of prior distributions used when estimating covariance and correlation matrices, e.g., Wishart priors and/or Lewandowski-Kurowicka-Joe (LKJ) priors~\cite{gelman1995bayesian}. However,  in the current study the marginal prior covariance matrices should not be altered, thus a natural (non-informative) choice for the prior on $\CC$  is $\pi(\CC)=\mathcal{U}(\mathbb{A}^{n_1\times n_2})$, i.e., the uniform distribution over $n_1\times n_2$ strict contractions. Then, by Theorem \ref{Thm2}, for any choice of $\CC$ from $\pi(\CC)$, the joint covariance matrix $\GGamma$ is a proper covariance matrix. Thus, in theory, given $\pi(\CC)$ and measured data we could attempt to estimate the $n_1\times n_2$ elements of $\CC$. However, without strong regularisation or additional structure the problem of estimating $\CC$ is significantly under-determined. In the current paper we attempt to remedy this situation by (a) postulating only correlation structures based on diagonal matrices $\CC$ and (b) we enforce strong spatial structure within $\CC$. The choice of restricting to diagonal matrices $\CC$ can be seen as prescribing (or estimating) only the canonical correlations themselves (recall the spectrum of $\Gamma_\z^{-\frac{1}{2}}\GGamma_{\z\w}\Gamma_\w^{-\frac{1}{2}}=\CC$ is the canonical correlations) and leads to a uniform prior of the form
\[\pi(\CC)=\Pi_{i=1}^{\min\{n_1,n_2\}}\pi(\CC_{ii})=\Pi_{i=1}^{\min\{n_1,n_2\}}\mathcal{U}(-1,1).\]
In the examples considered in the current work the parameters of interest $\z$ and $\w$ represent spatially distributed geophysical quantities. Nonspatial cross-correlations of $\z$ and $\w$ are usually estimated from samples obtained from various locations, for example, from boreholes, and treated as spatially homogeneous, suggesting $\text{corr}(p(x),m(x)) \equiv c$ over $x\in\Omega_k$ for a specific subsurface type. On the other hand, the (marginal) spatial correlation structures are supposed to be specified first. Accordingly, in the numerical examples considered here we restrict attention to diagonal contractions. To ensure $|\CC_{ii}|<1$, we reparameterise each of the correlation coefficients as $\CC_{ii}=\tanh(\gamma_i)$, with $\gamma_i\in\mathbb{R}$ having the prior induced by the transformation, i.e., $\pi(\gamma_i)=\frac{1}{2}\text{sech}^2(\gamma_i)$.
\paragraph{Estimability} 

The maximum a posteriori (MAP) estimate, defined as the point\footnote{See~\cite{helin2015maximum,dashti2013map} for an infinite-dimensional interpretation} that maximises the posterior, is often used as a computationally efficient estimator of the parameters. However, like any other point estimate, the MAP estimate can poorly represent the posterior (see \cite[Section 3.1.1]{KaipioSomersalo05}). Indeed, for the approach proposed here, when estimating both the parameters and the correlation (encoded via) $\CC$, the joint MAP estimate can be arbitrarily misleading, due to the unbounded form of the joint prior. We therefore use MCMC methods in this work and use the conditional mean (CM) estimate as the representative point estimate. This issue is illustrated in the Monod example of the next section, while more insights are given in Appendix \ref{sec: appMAP}. 

\subsection{A preliminary inference example}\label{sec:MONOD}
Here, we demonstrate the applicability of the proposed approach for inference on a simple
algebraic model from microbiology and microbial
ecology, the Monod model \cite{Koch98}. The Monod model is used to model bacterial population $\mu$ as a function of substrate concentration $S$ in the form
\begin{align}
\mu=\frac{\z S}{\w+S},
\end{align}
where $\z$ denotes the maximum population, and $\w$ the half-velocity constant. The Monod model has been used in the literature previously to demonstrate proposed methods in Bayesian inference, see \cite{BardsleySolonenHaarioEtAl14}. The parameters of interest $(\z,\w)$ are to be estimated based on several measurements of the population at different substrate concentrations. Specifically, we take measurements at $S=[28, 55, 83, 110, 138, 225, 375]$ which are corrupted by centred additive {\em iid} Gaussian noise.

We consider two slightly different approaches to the problem. First, to illustrate the dangers of mis-modelling the correlation, we consider the estimation problem using different values for the cross-correlation. Secondly, we consider the problem when the cross-correlation is also taken to be a random variable. In all cases the prior means are taken as $\z_\ast=0.4$ and $\w_\ast=40$, while the marginal prior standard deviations are set at $\sigma_\z=0.1$ and $\sigma_\w=10$. In the first of these cases we solve the problem using $c\in\mathcal{C}:=\{-0.99,-0.85,0,0.85,0.99\}$, and for noise levels of standard deviation $\delta_e=0.1$ and $\delta_e=0.03$. That is, the prior and noise models are given by
\[\pi(\z,\w)=\mathcal{N}\left(\begin{pmatrix}0.4\\40\end{pmatrix},\begin{pmatrix}0.1^2 & c\\c & 10^2 \end{pmatrix}\right),\quad \pi(e)=\mathcal{N}(0,\delta_e^2I),\]
and for the second approach ($c$ being unknown) we set $\pi(c)=\mathcal{U}(-1,1)$ and consider $\delta_e=0.03$ only.

Level set plots of each of the prior distributions and resulting posterior distributions for each $c\in\mathcal{C}$ and both noise levels are shown in Figure~\ref{fig: MonodFixed}. The truth is taken as $(\z_{\rm true},\w_{\rm true})=(0.7,65)$, so that we expect the joint inference to perform better (i.e., give posterior densities which support the truth better) when using the positively correlated priors. This is indeed the case, particularly in the case with larger noise. In both cases the posterior found when using a prior with $c=0.85$ gives the best results.

\begin{figure}[t!]
\centering
\begin{tikzpicture}
\node[inner sep=0pt,above] (a) at (0,0)
    {\includegraphics[width=0.3175\textwidth]{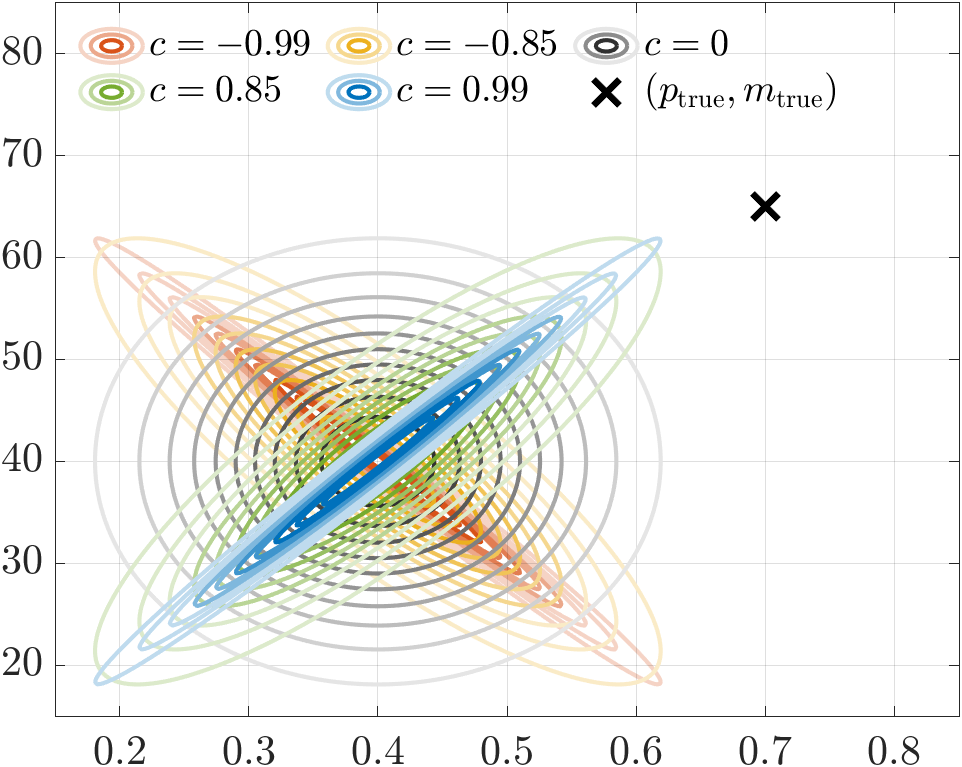}};
    \node at (-2.9,2) {$\w$};
    \node at (.125,-.2) {$\z$};
\node[inner sep=0pt,,above] (b) at (5.42,0)
    {\includegraphics[width=0.3\textwidth]{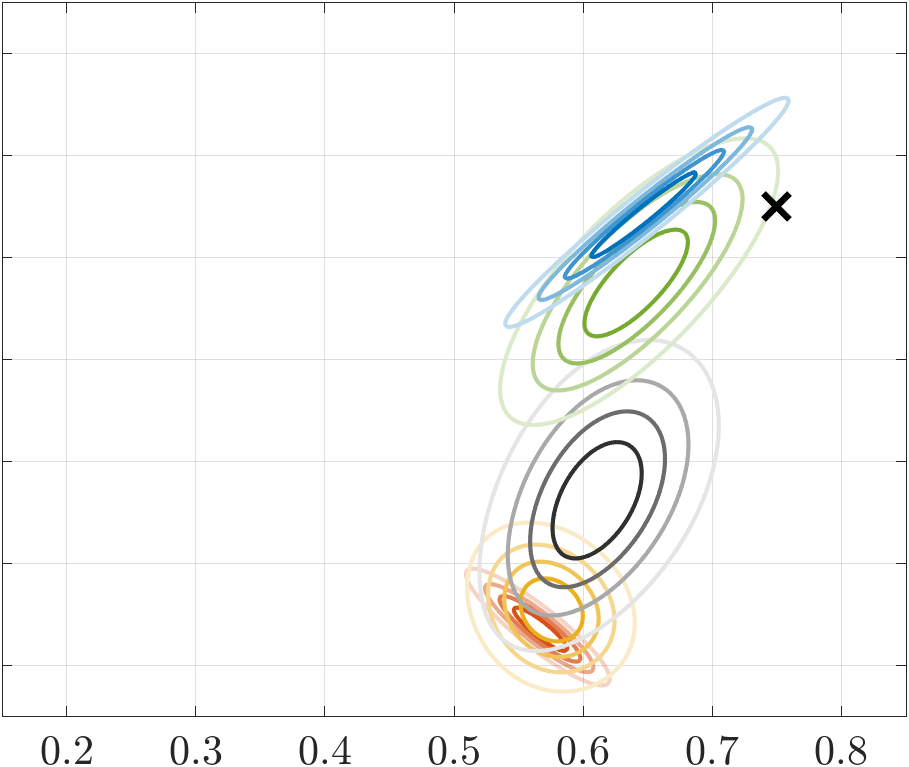}};

        \node at (5.4,-.2) {$\z$};
    \node[inner sep=0pt,above] (c) at (10.7,0)
    {\includegraphics[width=0.3\textwidth]{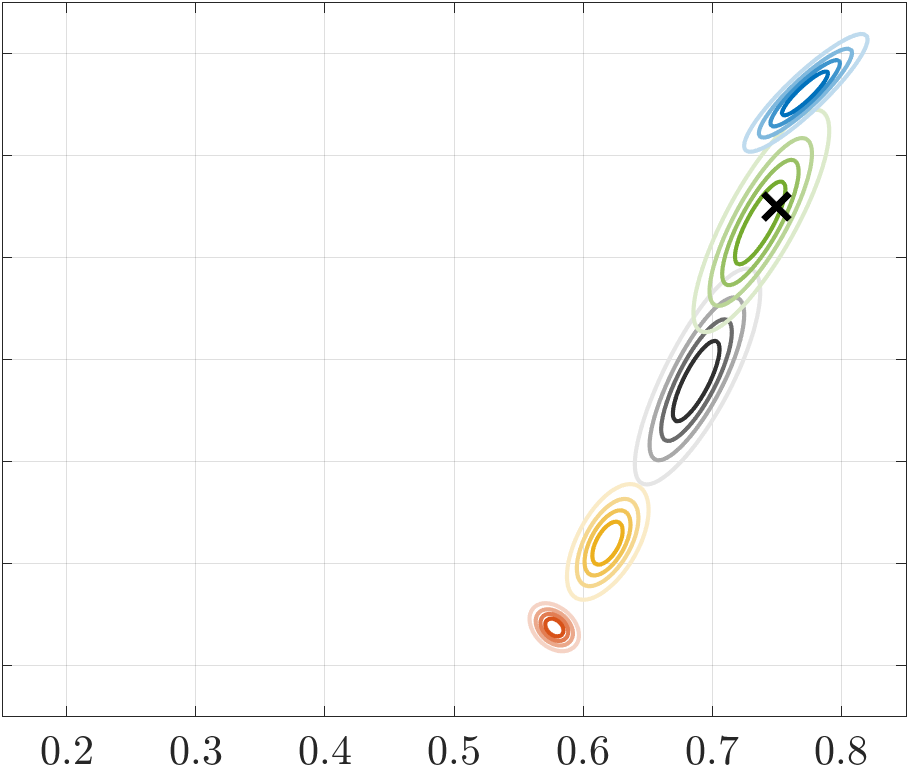}};

        \node at (10.7,-.2) {$\z$};

        \node at (0,4.75) {Prior};
\node at (5.5,4.75) {Posterior $(\delta_e=0.1)$};
\node at (11,4.75) {Posterior 
 $(\delta_e=0.03)$};
\end{tikzpicture}
\caption{Joint distributions for the Monod example with fixed values of $c$. On the left are the joint priors, in the middle are the joint posteriors for $\delta_e=0.1$, while on the right are the joint posteriors for $\delta_e=0.03$. In all cases the colors red, yellow, black, green, blue are used for the cases $c=-0.99$, $c=-0.85$, $c=0$, $c=0.85$, and $c=0.99$, respectively. Furthermore,  in all cases the black cross represents the truth.}\label{fig: MonodFixed}
\end{figure}
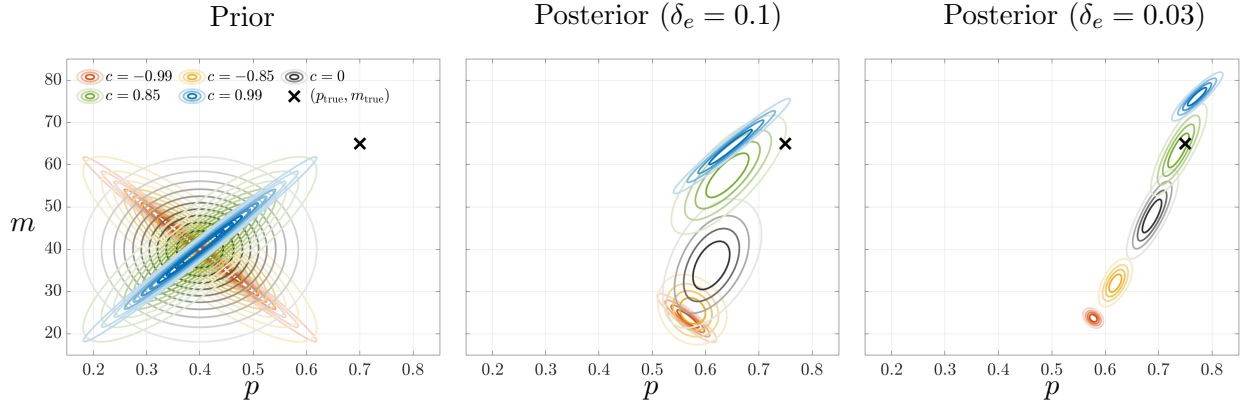

In the second approach to the Monod model we take $\z$, $\w$ and $c$ as unknown and compute the joint posterior $\pi(\z,\w,c|\mu)$. A corner plot showing the one and two parameter marginal posterior distributions is shown in Figure~\ref{fig:MONODfull}. The approach seems to perform fairly well; The marginal densities of the parameters $\z$ and $\w$ contain the true values well, and the posterior for $c$ indicates that a positive correlation is more likely than a negative one. The two dimensional marginals and the marginal distribution of $c$ illustrate the potential pitfalls of using an optimisation-based inference approach. When $c\to1$ the posterior becomes unbounded and poorly represents the truth, see Appendix~\ref{sec: appMAP} for more insight into this.
\begin{figure}[h!]
\centering
\begin{tikzpicture}
\node[inner sep=0pt] (a) at (0,0)
    {\includegraphics[width=0.3\textwidth]{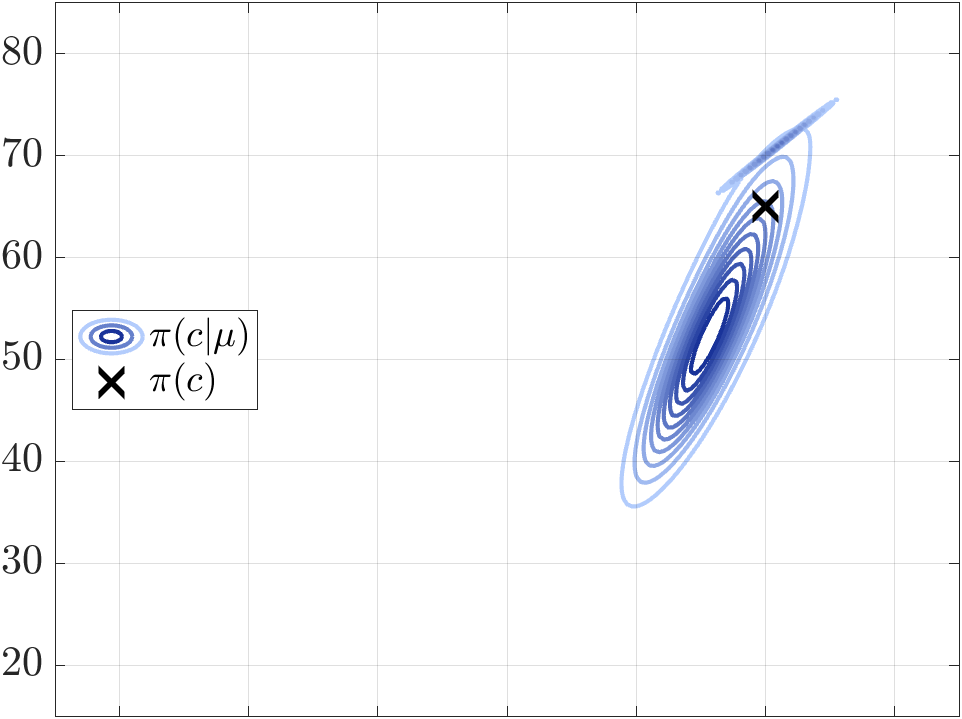}};

        \node at (-2.7,0) {$\w$};
    \node at (0,-6.5) {$\z$};
\node[inner sep=0pt] (b) at (-.07,-4.25)
    {\includegraphics[width=0.309\textwidth]{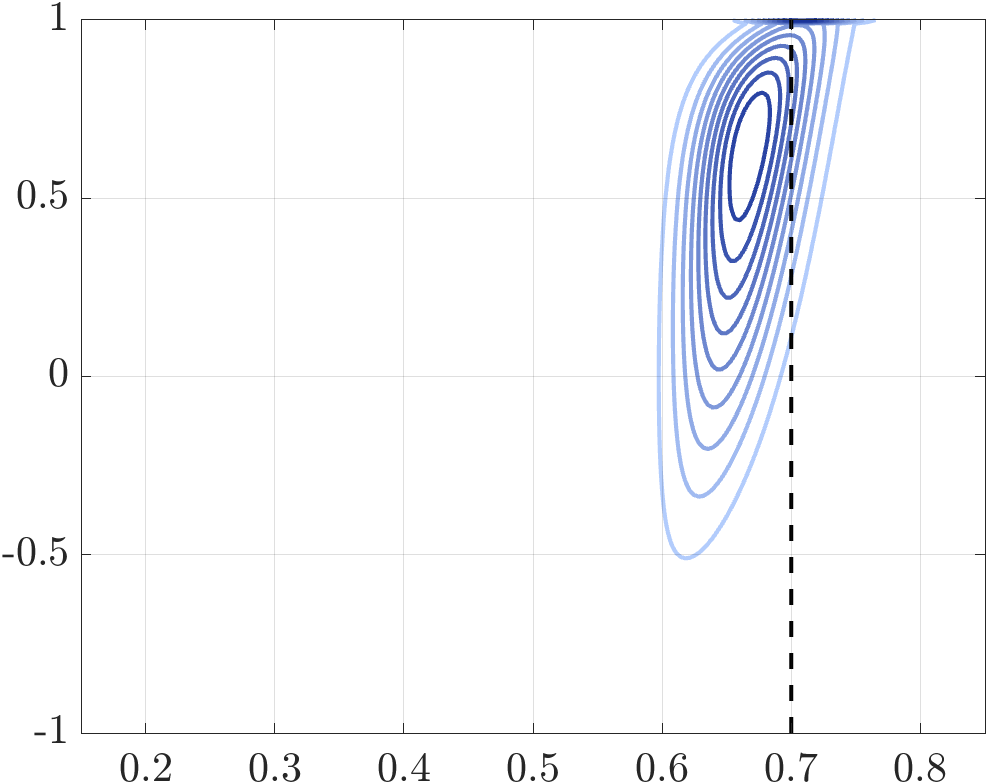}};
            \node at (3,0) {$c$};
    
    \node[inner sep=0pt] (c) at (5.62,-4.29)
    {\includegraphics[width=0.285\textwidth]{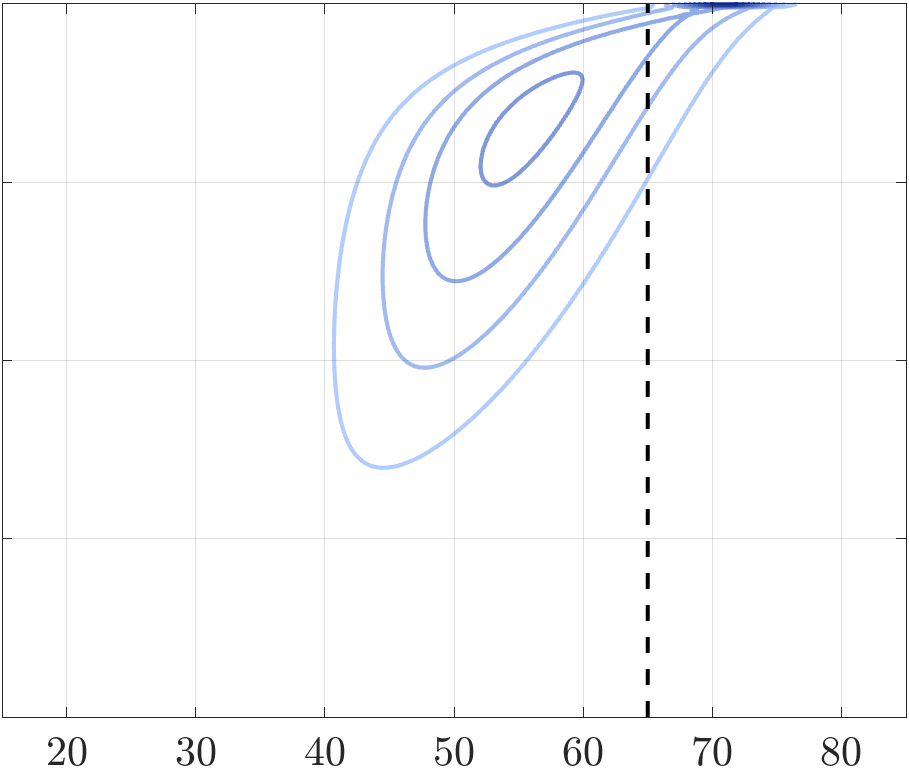}};
            \node at (-2.7,-4.29) {$c$};
    \node at (5.62,-6.5) {$\w$};
    
         \node[inner sep=0pt] (a) at (0,4.25)
    {\includegraphics[width=0.3\textwidth]{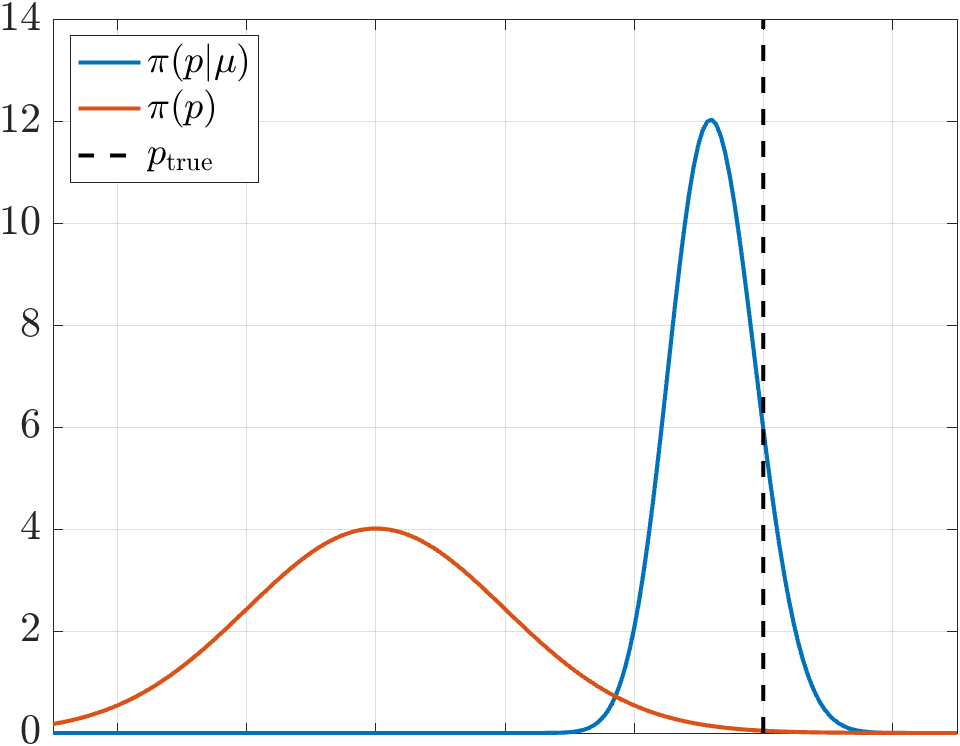}};

     \node[inner sep=0pt] (a) at (5.4,0)
    {\includegraphics[width=0.31\textwidth]{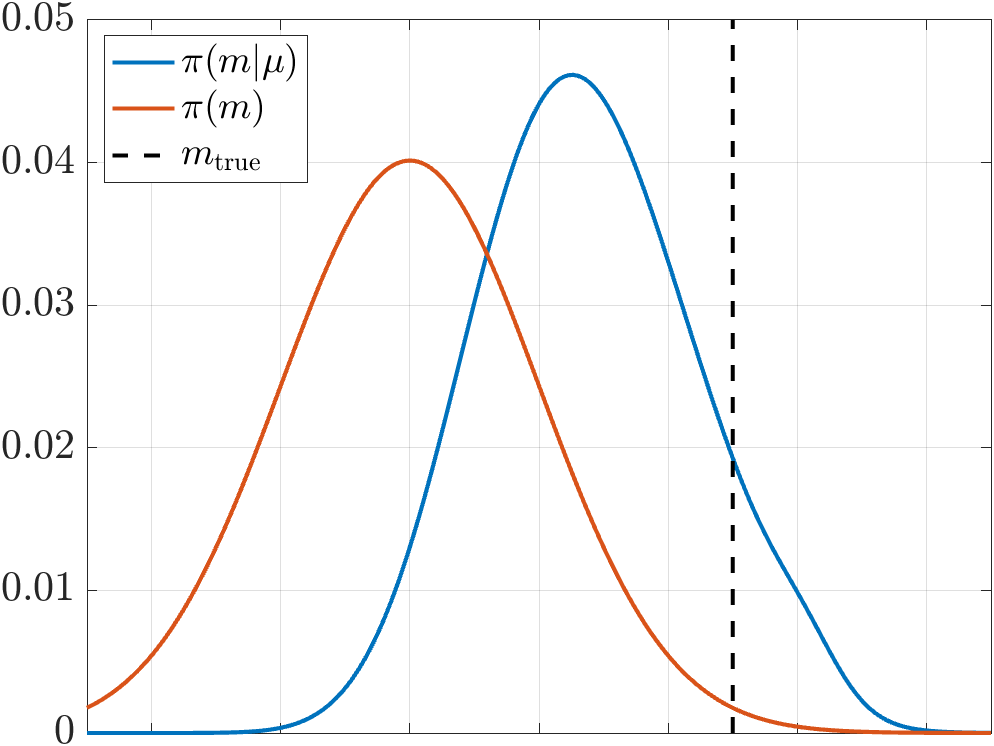}};
    
    \node[inner sep=0pt] (a) at (10.8,-4.25)
    {\includegraphics[width=0.306\textwidth]{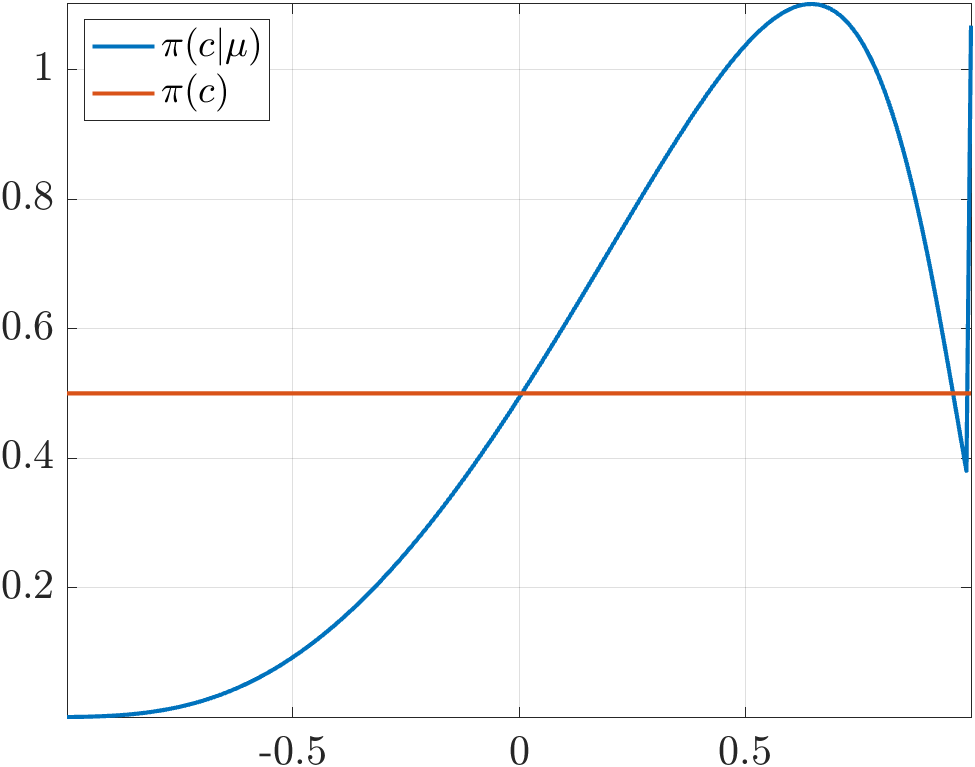}};

\node at (10.9,-6.5) {$c$};

\end{tikzpicture}
 \caption{Corner plot for the posterior of the Monod example where $c$ is also treated as an unknown and estimated. The true parameters are $(\z_{\rm true},\w_{\rm true})=(0.7,65)$. }\label{fig:MONODfull}
\end{figure}
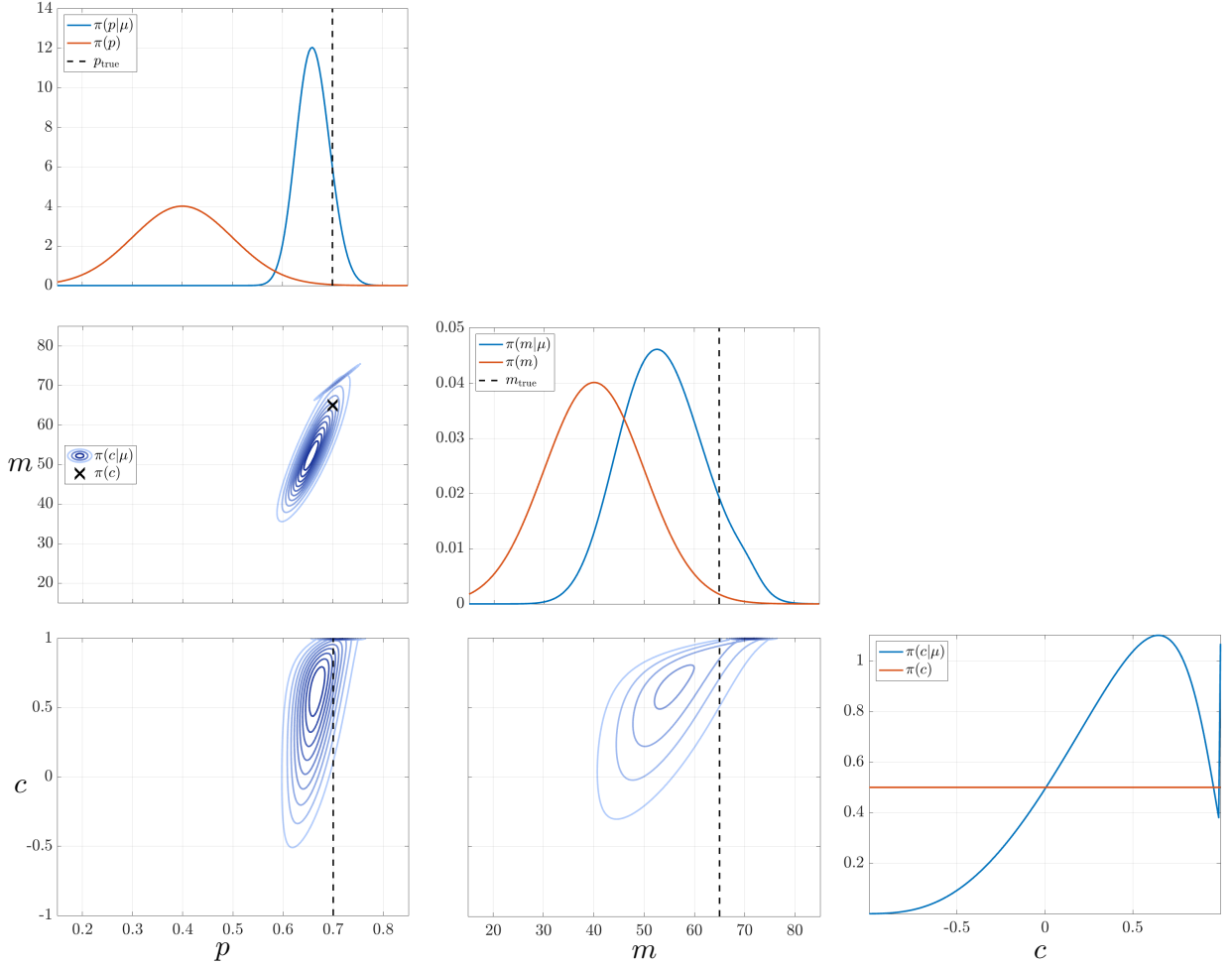

\section{Higher-dimensional computational experiments}\label{sec: NumEx}
In this section we consider two examples motivated by geophysical situations in which the (discretised) parameters  are significantly higher dimensional than the previous example. We first consider an example of co-Kriging, and secondly a subsurface flow problem where we estimate the permeability and recharge of an aquifer.

In both examples the computational domain is taken as $\Omega=[0,2]\times[0,1]$. Furthermore, in both examples we consider a finite element method (FEM) discretisation of the parameters. For simplicity, both parameters, $\z$ and $\w$, are represented using the same FEM basis functions, i.e., $\z=\sum_{i=1}^{n_1}\z_i\phi(x)$
and $\w=\sum_{i=1}^{n_2}\w_i\psi(x)$ (note we take $n_1=n_2$). The mesh consists of 1250 nodes, and 2352 triangular elements, thus (with a slight abuse of notation) the parameters of interest are the $\z,\w\in\mathbb{R}^{1250}$ coefficients of the FEM basis functions.

In line with the rest of the paper, the PDE-based covariance matrix (Equation~(\ref{eq: pdeprior})) and the exponential-squared covariance matrix (Equation~(\ref{eq: exp2prior})) are  used as the marginal prior covariance matrices. Specifically, for $\Gamma_\z$ we use an exponential-squared covariance matrix with a correlation of $\ell=0.3$, while for $\Gamma_\w$ we use the PDE-based covariance matrix with $a_1=1.5$, $a_2=30$, and $a_3=7.5$. Furthermore, in both examples we use our proposed approach to estimate $\z$, $\w$ as well as $\CC$. We compare the results against those found when we assume a priori that the parameters $\z$ and $\w$ are independent. In the first example we also show results found when using fixed values of the $\CC$.

\paragraph{Metrics}
To investigate the performance of the proposed joint approach, we compare the results against those obtained under the standard assumption that the two parameters are a priori independent. Specifically, we consider the relative errors in the CM estimates and the true parameter values,
\begin{align}
    \mathtt{E}(\z)=\norm{\z_{\rm true}-\z_{\ast|\data}}/\norm{\z_{\rm true}},\quad     \mathtt{E}(\w)=\norm{\w_{\rm true}-\w_{\ast|\data}}/\norm{\w_{\rm true}}.
\end{align}
Furthermore, to quantify the uncertainty reduction (from prior to posterior), we introduce a measure of relative posterior uncertainty
\begin{align}
  \mathtt{U}(\z)=\tr{\Gamma_{\z|\data}}/\tr{\Gamma_{\z}},\quad \mathtt{U}(\w)=\tr{\Gamma_{\w|\data}}/\tr{\Gamma_{\w}}, 
\end{align}
such that $\mathtt{U}(\cdot)\in[0,1]$, with smaller values indicating a larger reduction in uncertainty. We note that the aim is not necessarily to obtain smaller posterior variances, but to accurately model, incorporate, and quantify the uncertainty in the unknown parameters and the correlation. 

Finally, we also examine the difference in the point-wise posterior standard deviations
\begin{align}\label{eq:diffs}
  \mathbb{R}^{n_1} \ni\mathtt{D}_\z=\diag{(\Gamma_\z^{\rm I})}^{\frac{1}{2}}-\diag{(\Gamma_\z^{\rm J})}^{\frac{1}{2}},\quad
  \mathbb{R}^{n_2}\ni\mathtt{D}_\w=\diag{(\Gamma_\w^{\rm I})}^{\frac{1}{2}}-\diag{(\Gamma_\w^{\rm J})}^{\frac{1}{2}},
\end{align} 
where $\Gamma_\z^{\rm I}$ (resp. $\Gamma_\w^{\rm I}$) and $\Gamma_\z^{\rm J}$ (resp. $\Gamma_\w^{\rm J}$) are the posterior covariance matrices of $\z$ 
(resp. $\w$) found using independent inference and joint inference, respectively.

\subsection{Computational example 1: Co-Kriging}\label{sec ex1}
For the first example we consider a classic co-kriging inspired problem. Specifically, we let $\Omega$ denote the computational domain of interest and assume we have direct (noisy) measurements of $\z$ and of $\w$ at different sets of measurement locations. The full inference problem is then to compute the joint posterior $\pi(\z,\w,\CC|\data_1,\data_2)$ given the measurements $\data_1$ and $\data_2$ which are assumed to be of the form
\begin{equation}
\begin{aligned}
\data_1&=\mathcal{B}_1\z+\ee\\
\data_2&=\mathcal{B}_2\w+\eee,
\end{aligned}
\end{equation}
with $\mathcal{B}_1$ and $\mathcal{B}_2$ the point-wise restriction/observation operators which evaluate $\z$ and $\w$, respectively, at the observation locations, see Figure \ref{fig:P1setup}.

It is worth pointing out that as both forward models are linear, for a fixed correlation between $\w$ and $\z$, the posterior is Gaussian and can be computed analytically \cite{KaipioSomersalo05}. However, when the goal is to also estimate the correlation, the inference problem is no longer linear, which can significantly increase the computational overheads.

\subsection{Computational details for Example 1}
In this example we parameterise the unknown correlation using $C=cI$, where $c$ is a scalar with a uniform prior distribution, i.e., $c\sim\mathcal{U}(-1,1)$. That is, the correlation is assumed to be homogeneous throughout $\Omega$. The total number of unknowns is thus $1250+1250+1=2501$. On the other hand, the data consists of 60 point-wise measurements of $\z$ in the right half of the domain and 32 point-wise measurements of $\w$ in the top half of the domain, see Figure \ref{fig:P1setup}.  To generate the true parameters (and data) we set a true correlation of $c_{\rm true}=-0.9$ and generate $\z_{\rm true}$ and $\w_{\rm true}$ as samples from the joint conditional distribution $\pi(\z,\w|c=c_{\rm true})$. The resulting true parameters as well as the true potential induced by $\z_{\rm true}$ are shown in Figure~\ref{fig:P1setup}. Synthetic data is then generated by applying the observation operators and adding white noise of the form $\ee\sim\mathcal{N}(0,\delta^2_1I)$ and 
$\eee\sim\mathcal{N}(0,\delta^2_2I)$. For this example we set  $\delta_1=1/100\times(\max(\mathcal{B}_1\z_{\rm true})-\min(\mathcal{B}_1\z_{\rm true}))$ and $\delta_2=1/100\times(\max(\mathcal{B}_2\w_{\rm true})-\min(\mathcal{B}_2\w_{\rm true}))$. That is, both sets of synthetic data are corrupted by 1\% noise.

To further illustrate the importance of correctly accounting for the (uncertainty in the) correlation, for this example we also compute the posterior for two additional fixed choices of $c$: $c=c_{\rm true}=-0.9$ and $c=-c_{\rm true}=0.9$. For each of the three cases of fixed correlations, $c\in\{-0.9,0,0.9\}$, the CM estimates and conditional covariances of $\z$ and $\w$ are calculated using the well known formulae \cite[Eqs.~(3.14) and~(3.15)]{KaipioSomersalo05}. On the other hand, to compute the full joint posterior $\pi(\z,\w,c|d_1,d_2)$ we use the MwG approach outlined in Algorithm~\ref{alg: MwG}. However, as noted in Remark~\ref{rm: MwG}, for cases in which the forward model is linear, which is the case for Example 1, the MwG approach simplifies considerably, with only the correlation $\CC$ requiring a Metropolis step for updating, and $\z$ and $\w$ updated by directly sampling from $\pi(\z,\w|d,\CC)$. Since in Example 1 we have $\CC=cI$, the correlation update reduces to a one-dimensional Metropolis step for the scalar parameter $c$. As noted in Section~\ref{sec: CompCon}, we reparameterise the correlation as $c=\tanh{(\gamma)}$. We initialize the chain for $\gamma$ at $\gamma=c=0$ and update using a Gaussian proposal centred at the current state with unit variance, i.e., 
\[\gamma'=\gamma^{(k)}+\xi^{(k)},\quad\xi^{(k)}\sim\mathcal{N}(0,1).\]
The MCMC is run for a total of $10^5$ samples, with the first $10^3$ discarded as {\em burn-in}\footnote{The length of the burn-in may seem small, however, only the scalar parameter $c$ is updated via a Metropolis step, while $\z$ and $\w$ are sampled directly from their conditional posterior distribution.}. The results found using each of fixed correlations and for the joint approach are discussed in the next section, while a discussion on the computational costs and effectiveness of the MCMC is left to Section~\ref{sec: CompCost}.

\begin{figure}[t!]
\centering
\begin{tikzpicture}
\node[inner sep=0pt] (a) at (0,0)
{\includegraphics[height=0.19\textwidth]{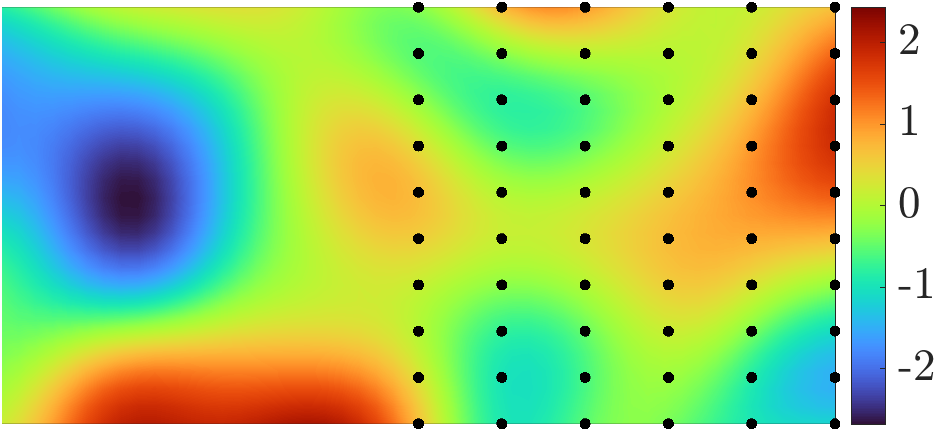}};

\node[inner sep=0pt] (a) at (8.5,0)
{\includegraphics[height=0.19\textwidth]{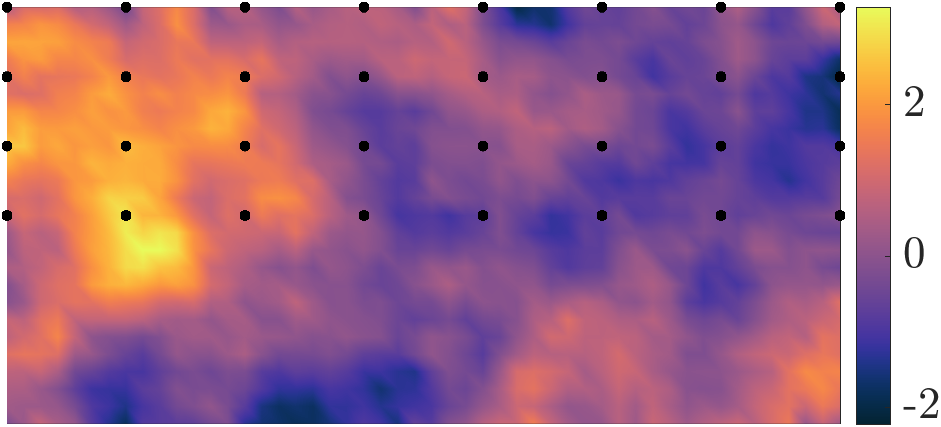}};

\node at (-.5,2) {$\z_{\rm true}$}; 
\node at (8,2) {$\w_{\rm true}$};

\end{tikzpicture}
\caption{The true parameters $\z_{\rm true}$ (left) and $\w_{\rm true}$ (right). In each plot the black dots indicate the locations in which the noisy (direct) point-wise measurements are taken. }\label{fig:P1setup}
\end{figure}

\subsection{Results for Example 1}\label{sec: res1}
For Example 1, the CM estimates of both $\z$ and $\w$ obtained using independent and joint inference are shown in Figure~\ref{fig:Ex1means}, along with the true parameters used to generate the data. As expected, we see that when using independent inversion, the CM estimates represent the truth fairly well near the respective measurement locations. However, away from the measurement locations, the estimates found using independent inversion are less insightful and converge spatially towards the prior mean. On the other hand, when using the proposed joint approach the CM estimates improve significantly compared to the independent estimates as data from both models is taken into account when estimating both parameters. This is, of course, an intended benefit of carrying out the inversion jointly. For this example, the error in the CM estimates using the independent approach is $\mathtt{E}(\z)=0.852$ and $\mathtt{E}(\w)=0.629$, whereas the proposed joint approach reduces these to $\mathtt{E}(\z)=0.513$ and $\mathtt{E}(\w)=0.589$.

\begin{figure}[t!]
\centering
$\z$ \hrulefill\vspace{5pt}

{\includegraphics[height=0.15\textwidth]{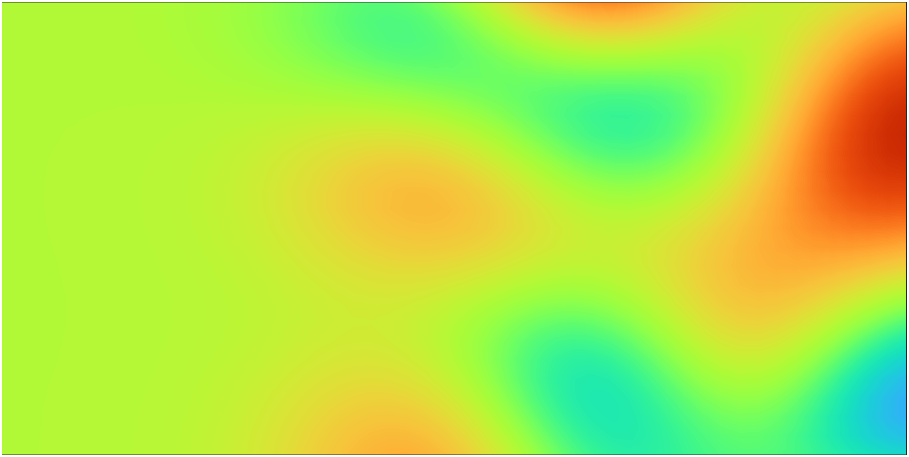}}
\hfill
{\includegraphics[height=0.15\textwidth]{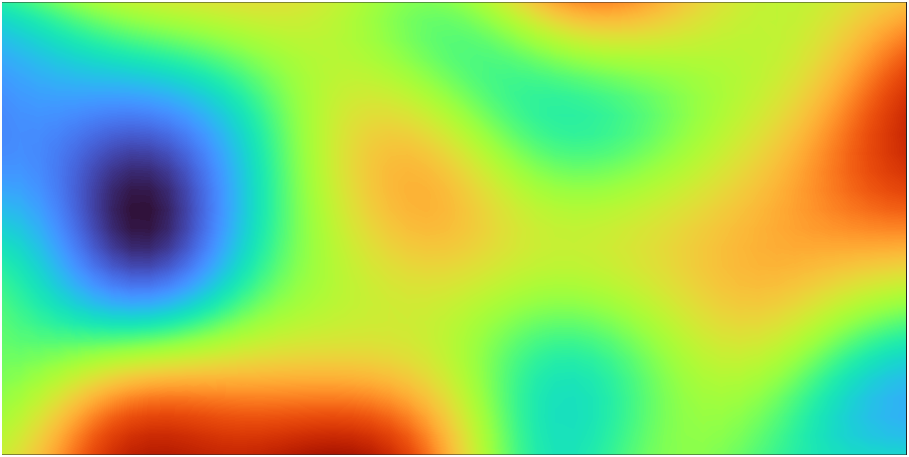}}
\hfill
{\includegraphics[height=0.15\textwidth]{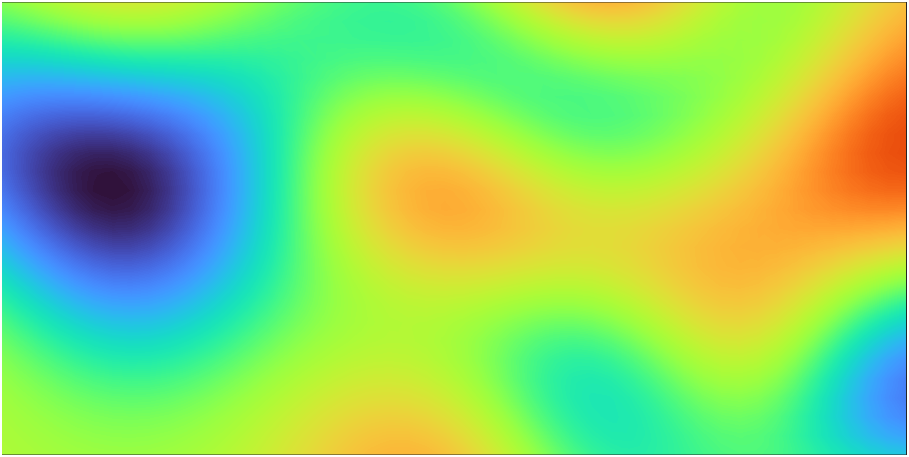}}
{\includegraphics[height=0.15\textwidth]{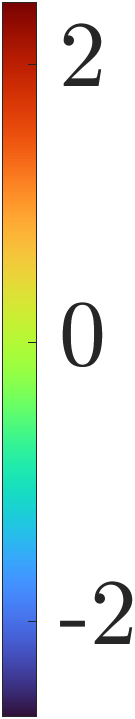}}\\
$\w$ \hrulefill\vspace{5pt}

{\includegraphics[height=0.15\textwidth]{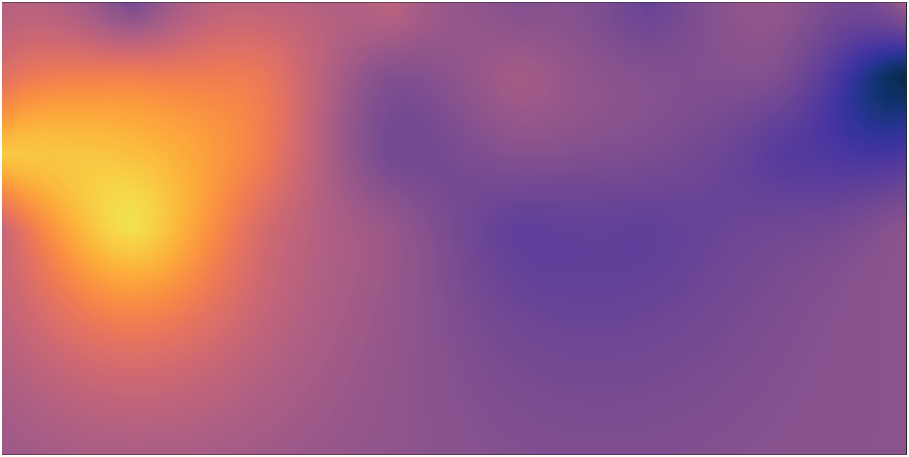}}
\hfill
{\includegraphics[height=0.15\textwidth]{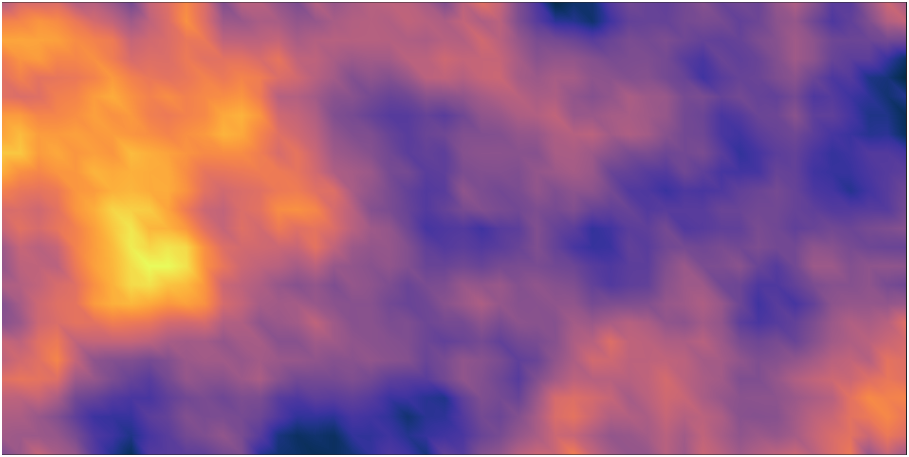}}
\hfill
{\includegraphics[height=0.15\textwidth]{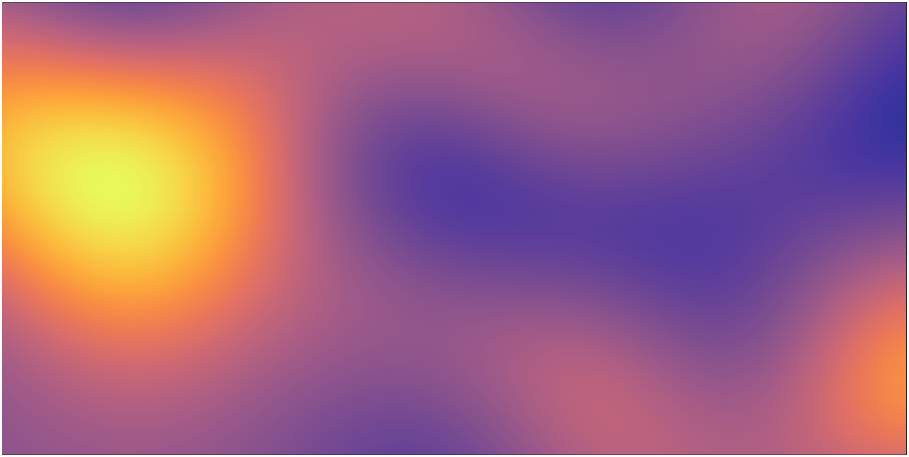}}
{\includegraphics[height=0.15\textwidth]{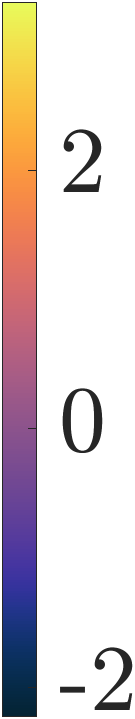}}
\caption{True parameters and CM estimates for Example 1. In the top row we show the CM found using independent inference (left), the true parameter $\z_{\rm true}$ (centre), and the CM found using joint inference (right). On the bottom row the corresponding plots for $\w$.}\label{fig:Ex1means}
\end{figure}

In Figure \ref{fig:Prob1UQ} we show the point-wise posterior standard deviations obtained using independent and joint inference, together with their difference (defined (\ref{eq:diffs})). As expected, the largest reduction in uncertainty for each of the parameters is near the measurement locations. However, as discussed above for the CM estimates, using the joint approach incorporates both sets of data into both estimates. Consequently, we see that when using the proposed joint approach, the posterior uncertainty in $\z$ is significantly reduced where $\w$ is measured, and that the posterior uncertainty in $\w$ is significantly reduced where $\z$ is measured. These differences are made apparent in the plots of the differences in the point-wise posterior standard deviation shown in Figure \ref{fig:Prob1UQ}. Quantitatively, the relative posterior uncertainties are $\mathtt{U}(\z)=0.401$ and $\mathtt{U}(\w)=0.411$ using independent inference, and $\mathtt{U}(\z)=0.313$ and $\mathtt{U}(\w)=0.291$ when using the joint approach. That is to say, in this example using the joint approach significantly reduces the  posterior uncertainty for both parameters.

\begin{figure}[t!]
\centering
$\z$ \hrulefill\vspace{5pt}

{\includegraphics[height=0.145\textwidth]{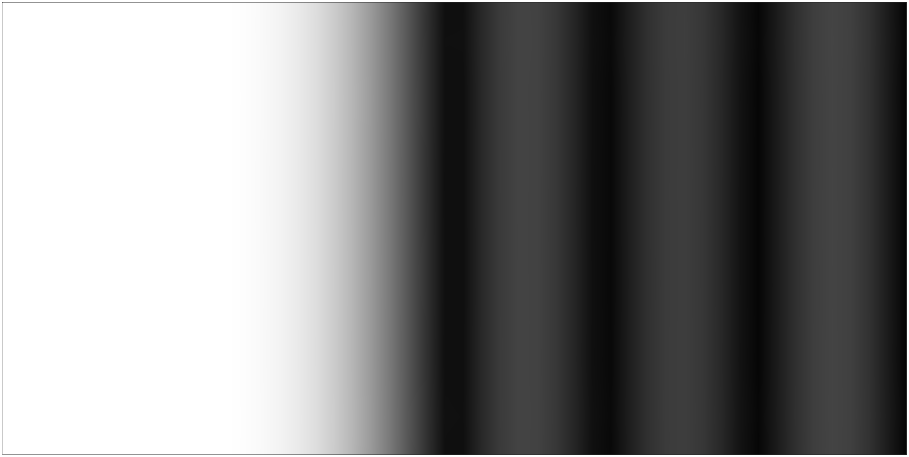}}
\hfill
{\includegraphics[height=0.145\textwidth]{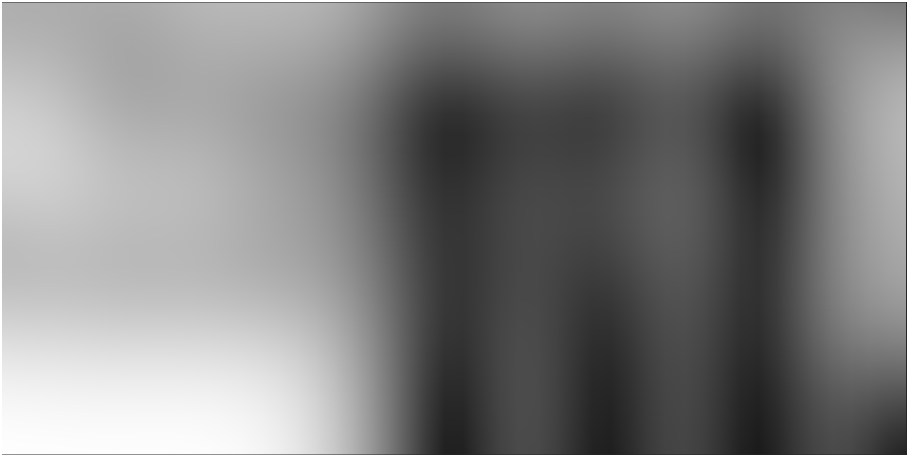}}
{\includegraphics[height=0.145\textwidth]{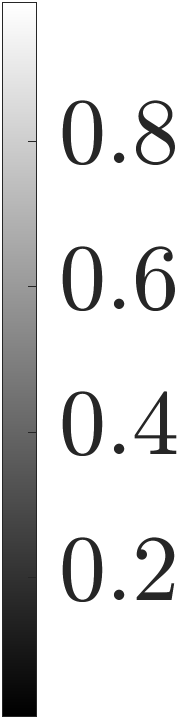}}
\hfill
{\includegraphics[height=0.145\textwidth]{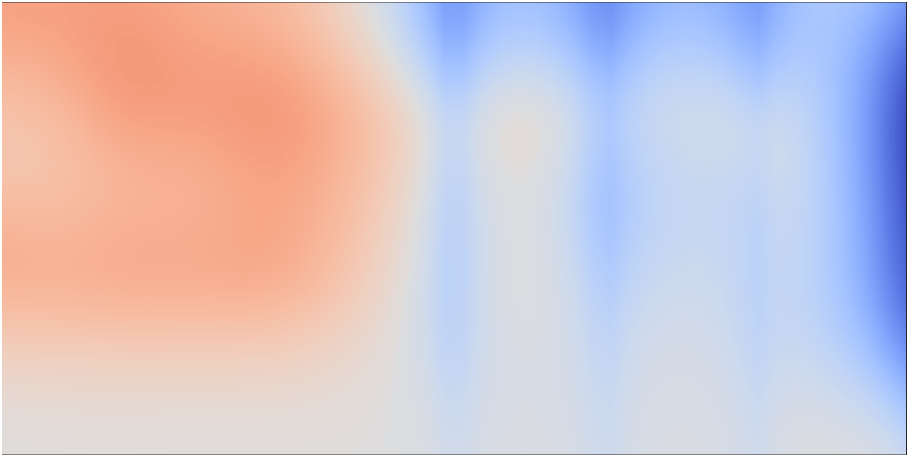}}
{\includegraphics[height=0.145\textwidth]{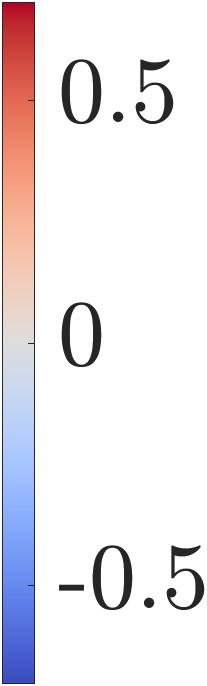}}\\
$\w$ \hrulefill\vspace{5pt}

{\includegraphics[height=0.145\textwidth]{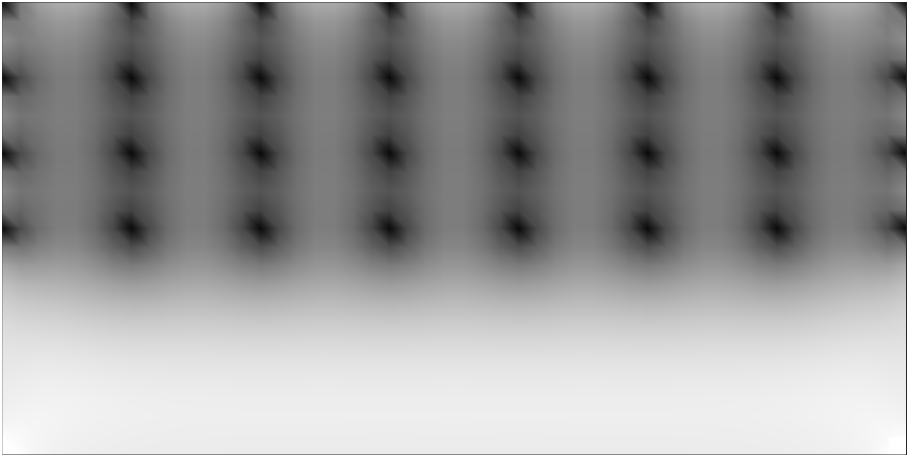}}
\hfill
{\includegraphics[height=0.145\textwidth]{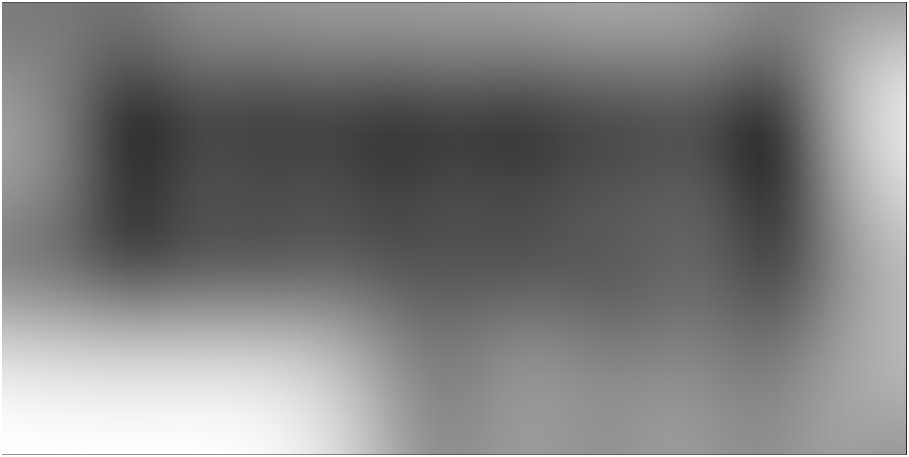}}
{\includegraphics[height=0.145\textwidth]{figures/p1COVCB.png}}
\hfill
{\includegraphics[height=0.145\textwidth]{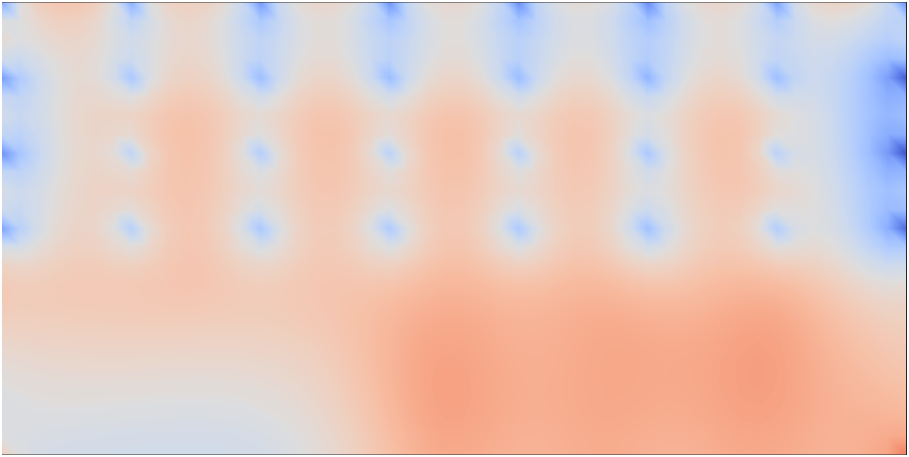}}
{\includegraphics[height=0.145\textwidth]{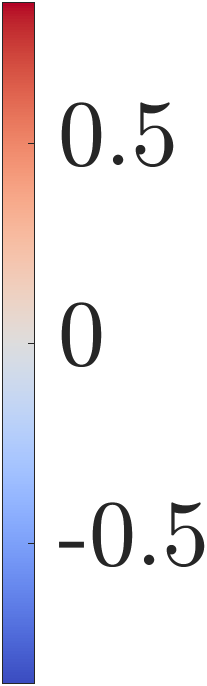}}

\caption{The reduction in uncertainty for Example 1. On the top row we show the point-wise marginal posterior standard deviation for $\z$ found using independent inference (left), the point-wise marginal posterior standard deviation for $\z$ found using joint inference (centre), and the difference of the two $\mathtt{D}_\z$ (see Equation \ref{eq:diffs}). On the bottom row we show the equivalent figures for $\w$.}\label{fig:Prob1UQ}
\end{figure}

A key benefit of the proposed joint approach is the ability to estimate the correlation also. In Figure \ref{fig:C1} we show the marginal histogram of the samples of $c$ (recall $\CC=cI$ for this example), the trace of the MCMC samples of $c$, and the autocorrelation plot. The posterior places most of its mass on negative values of $c$, with the true value, $c_{\rm true}=-0.9$ being well contained within the bulk of the posterior. This is consistent with the improved reconstructions obtained using the proposed joint inference approach, and indicates that the data are informative about the sign and magnitude of the correlation. The trace plot and sample autocorrelation suggest that the chain mixes reasonably well, although a discussion of computational costs is deferred to Section~\ref{sec: CompCost}.
\begin{figure}[t!]
\centering
\includegraphics[height=0.3\textwidth]{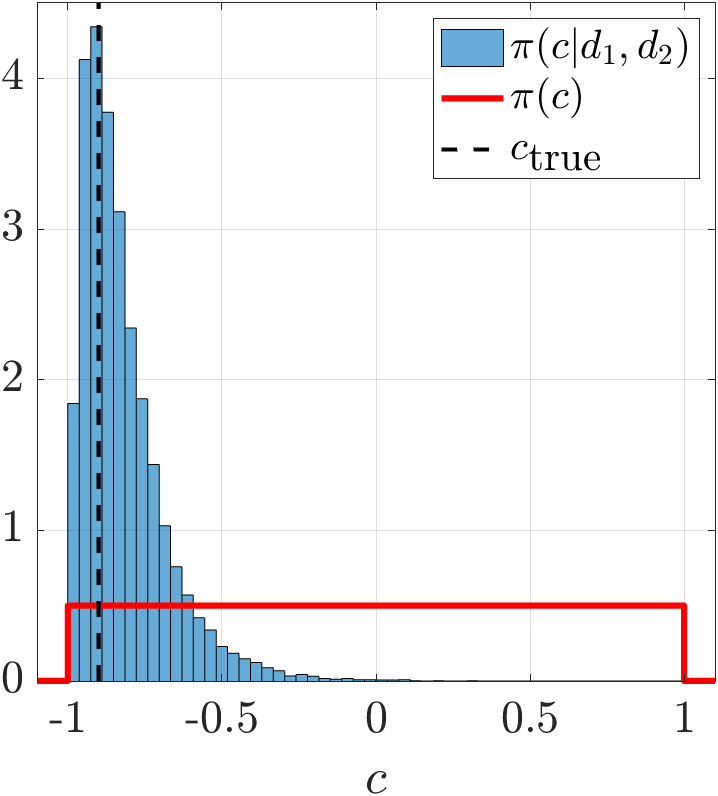}
\hfill
\includegraphics[height=0.3\textwidth]{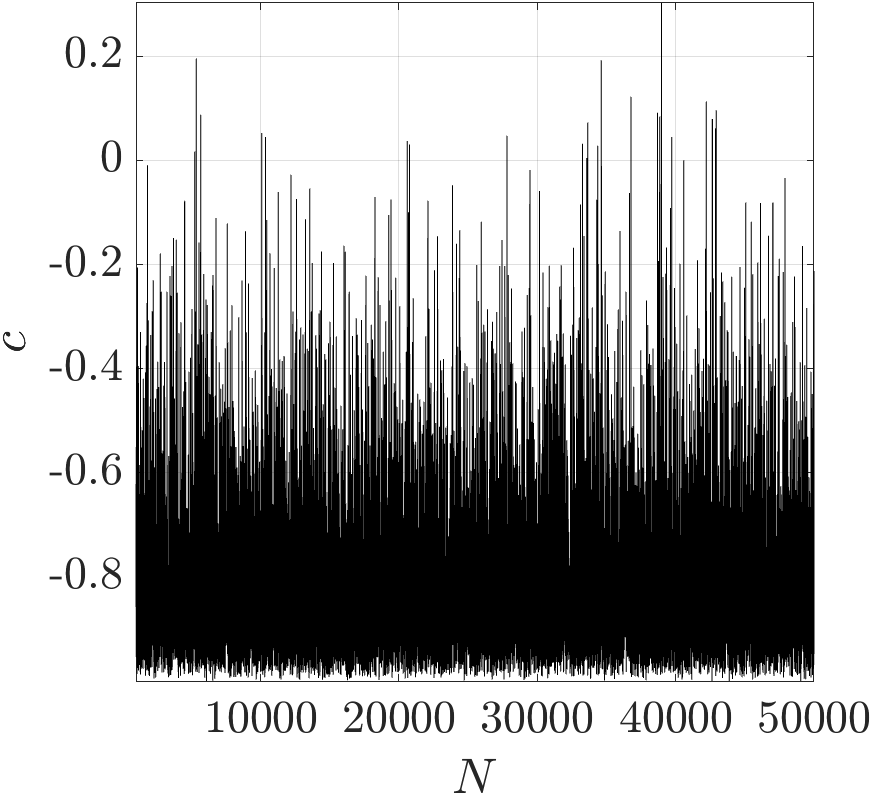}
\hfill
\includegraphics[height=0.3\textwidth]{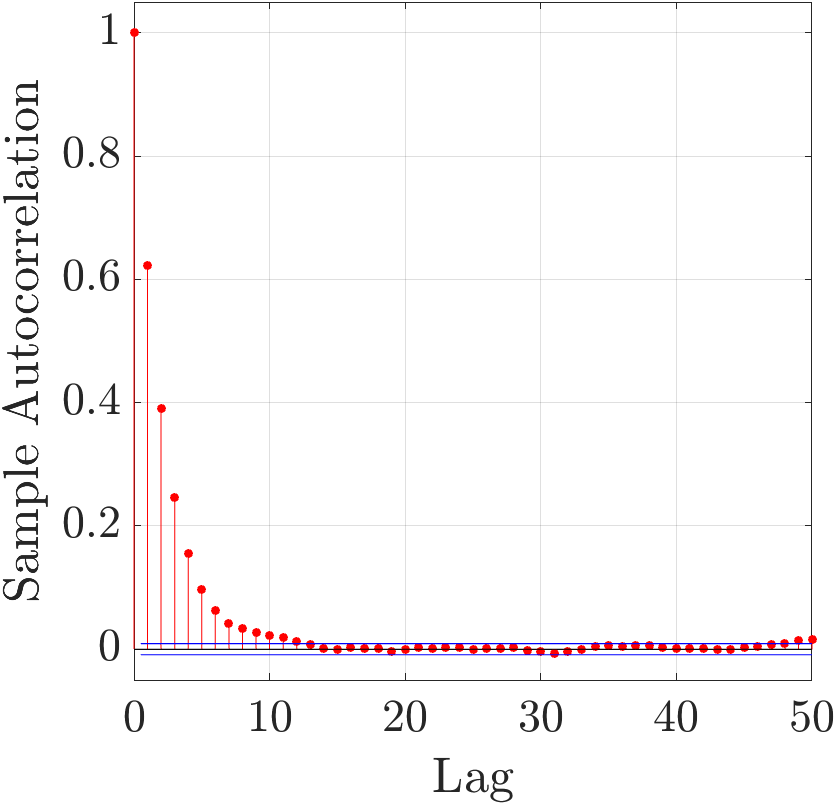}

\caption{Results for estimating the correlation $c$ in Example 1. Shown on the left is the histogram of samples from the marginal posterior $\pi(c|d)$, the prior $\pi(c)$ (red line), and the truth $c_{\rm true}$ (black dashed line), while in the centre we show the trace of the samples for $c$ and on the right the sample autocorrelation.}\label{fig:C1}
\end{figure}

To further illustrate the importance of correctly accounting for the correlation, we also compute the posterior for two fixed choices of $c$: $c=c_{\rm true}=-0.9$ and $c=-c_{\rm true}=0.9$. The resulting CM estimates and posterior uncertainty are shown in Figure~\ref{fig:Ex1fixedCs}. Fixing $c=c_{\rm true}$ results in CM estimates which agree well with the true parameters, while fixing $c=-c_{\rm true}$ gives misleading estimates. The largest reductions in marginal posterior uncertainty occur near the measurement locations associated with the parameters. However, as the correlation is strong, measurements of one parameter also reduce the uncertainty in the other. As shown in Appendix~\ref{app:sign}, the marginal posterior covariance is invariant under the sign change $c\mapsto-c$ in the present linear-Gaussian setting. Consequently, an incorrect sign of the correlation can yield a concentrated posterior in the wrong region of parameter space, so that the truth is poorly supported.
\begin{figure}[b!]
\centering
$\z$ \hrulefill\vspace{5pt}

\includegraphics[height=0.15\textwidth]{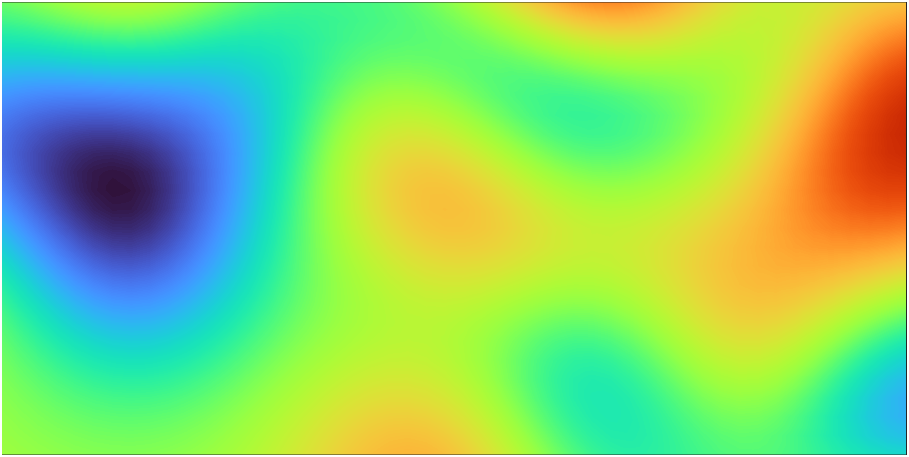}
\hfill
\includegraphics[height=0.15\textwidth]{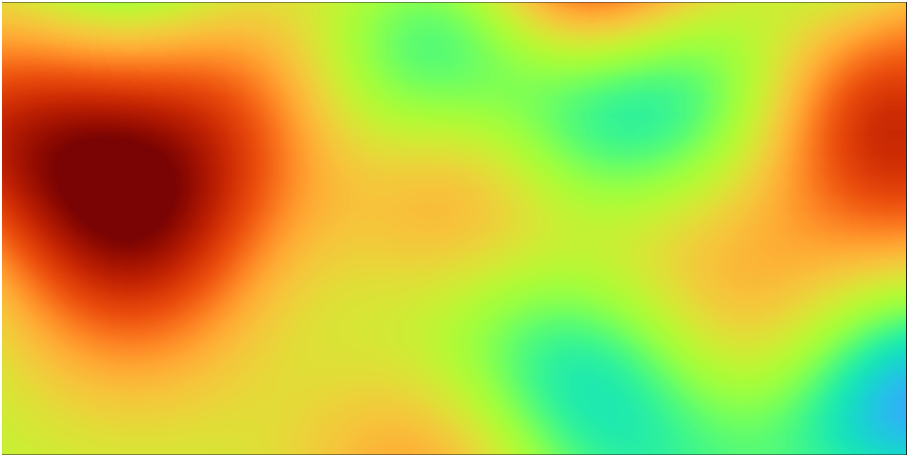}
\includegraphics[height=0.15\textwidth]{figures/Prob1P2CB.png}
\hfill
\includegraphics[height=0.15\textwidth]{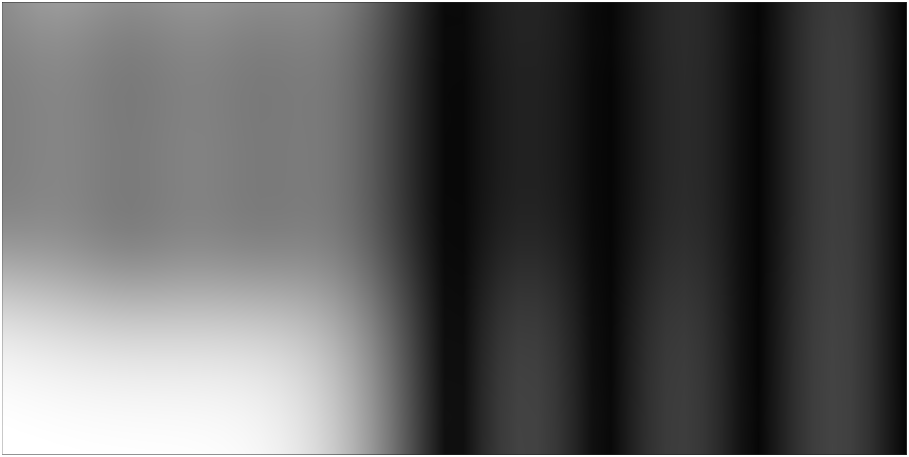}
\includegraphics[height=0.15\textwidth]{figures/p1COVCB.png}
\\
$\w$ \hrulefill\vspace{5pt}

\includegraphics[height=0.15\textwidth]{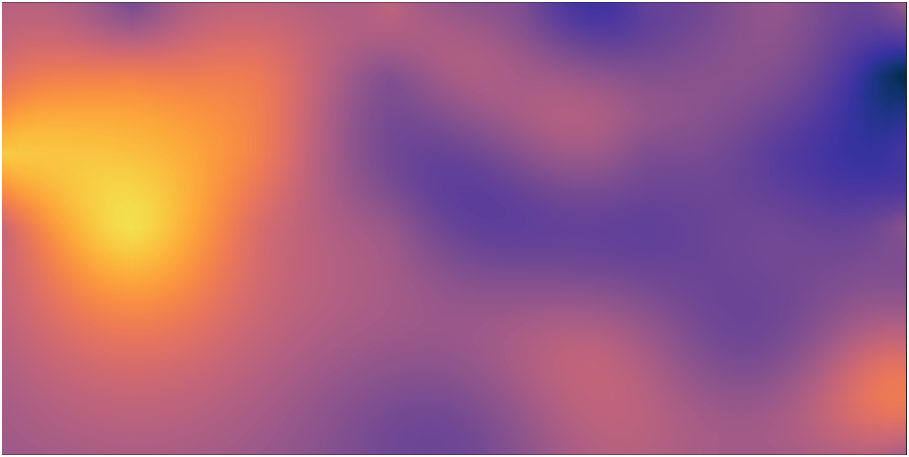}
\hfill
\includegraphics[height=0.15\textwidth]{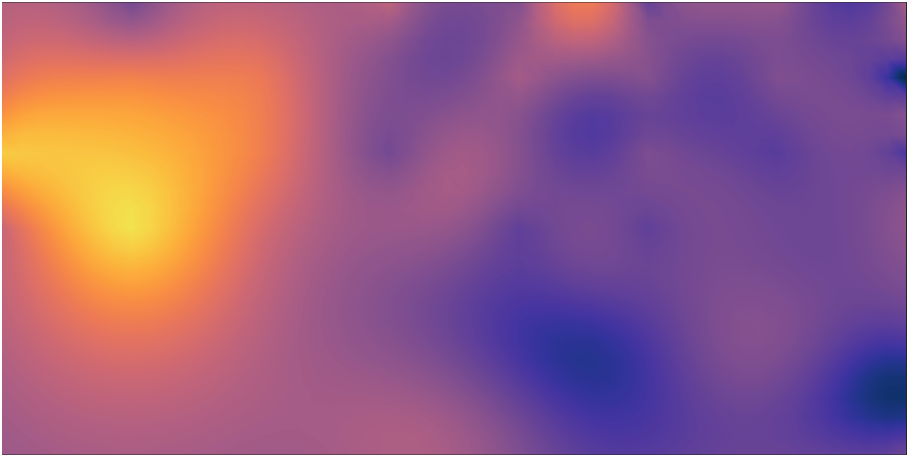}
\includegraphics[height=0.15\textwidth]{figures/Prob1P1CB.png}
\hfill
\includegraphics[height=0.15\textwidth]{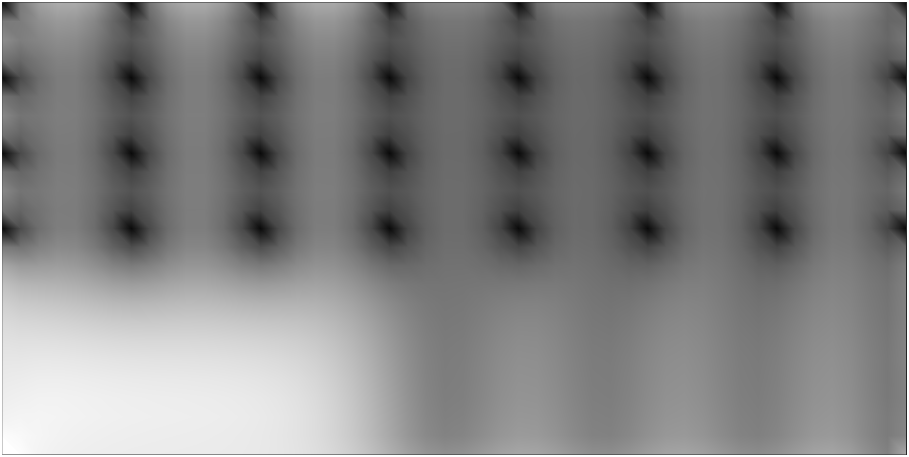}
\includegraphics[height=0.15\textwidth]{figures/p1COVCB.png}
\caption{Results for Example 1 using fixed values of $c$. On the top row we show the CM estimate of $\z$ using $c=c_{\rm true}=-0.9$ (left), the CM estimate of $\z$ using $c=-c_{\rm true}=0.9$ (centre), and the point-wise marginal posterior standard deviation of $\z$ using $c=c_{\rm true}$ or $c=-c_{\rm true}$ (right). Recall the point-wise marginal posterior standard deviation is invariant to the sign of $c$. In the bottom row we show the corresponding plots for $\w$.}\label{fig:Ex1fixedCs}
\end{figure}

\newpage
\subsection{Computational example 2: Inversion for aquifer permeability and recharge}

We now consider the more realistic subsurface flow problem of estimating the permeability and recharge of an aquifer based on measurements of hydraulic-head. This problem leads to a joint inverse problem governed by a single physical model. Specifically, taking the permeability $\kappa=\exp(\z)$ and the recharge $r=\exp(\w)$, we assume the hydraulic-head $u$ satisfies
\begin{equation}\label{eq: Darcy}
\begin{aligned}
-\nabla\cdot(\exp(\z)\nabla u)&=\exp(\w),\quad{\text{in }}\Omega\\
u&=0,\quad{\text{on }}\partial\Omega.
\end{aligned}
\end{equation}
As well as measurements of the hydraulic-head, we assume we also have direct point-wise measurements (e.g., core samples or borehole measurements) of the permeability. As such, the goal in the second problem is to characterise  the joint posterior $\pi(\z,\w,\CC|\data_1,\data_2)$ given the measurements of the form
\begin{equation}
\begin{aligned}
\data_1&=\mathcal{B}_1u(\z,\w)+\ee\\
\data_2&=\mathcal{B}_2\z+\eee,
\end{aligned}
\end{equation}
where $\mathcal{B}_1$ and $\mathcal{B}_2$ denote the point-wise observation operators. We assume the hydraulic-head is measured at shallow locations throughout the domain, while the direct measurements of $\z$ are taken along several well-paths, see Figure~\ref{fig:prob2truths}.

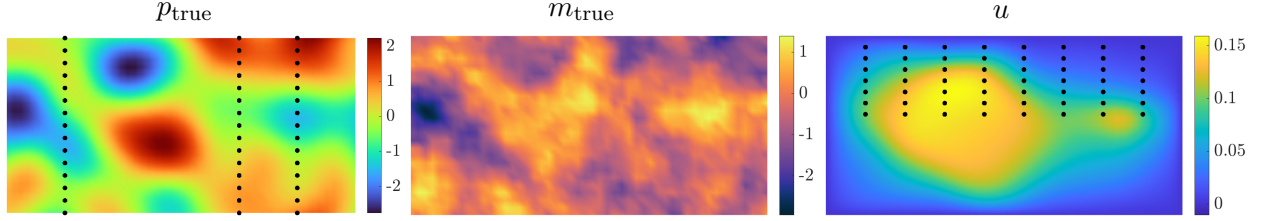
\begin{figure}[t!]
\centering
\begin{tikzpicture}
\node[inner sep=0pt] (a) at (0,0)
{\includegraphics[height=0.144\textwidth]{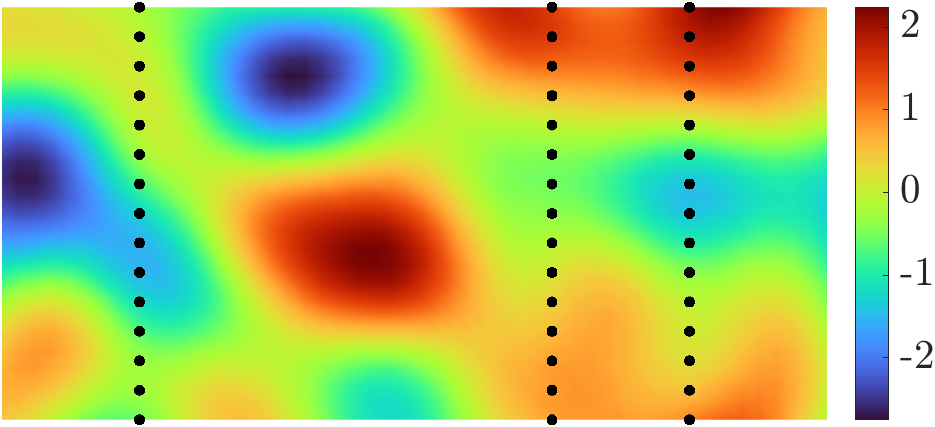}};

\node[inner sep=0pt] (a) at (5.4,0)
{\includegraphics[height=0.144\textwidth]{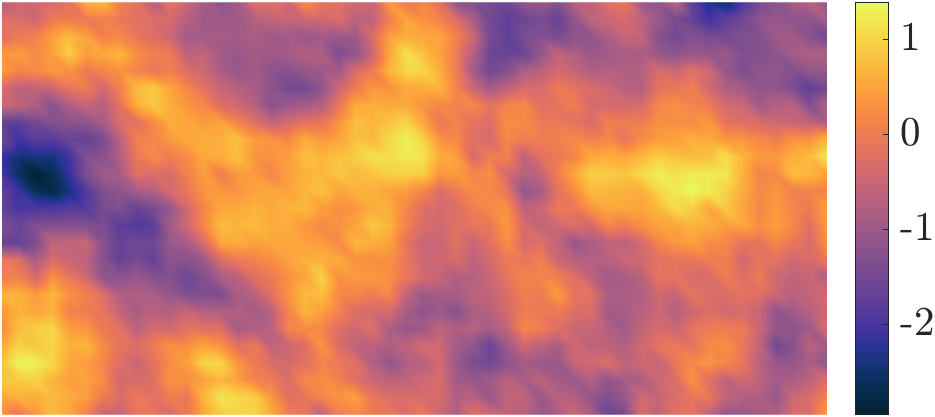}};

\node[inner sep=0pt] (a) at (11,0)
{\includegraphics[height=0.144\textwidth]{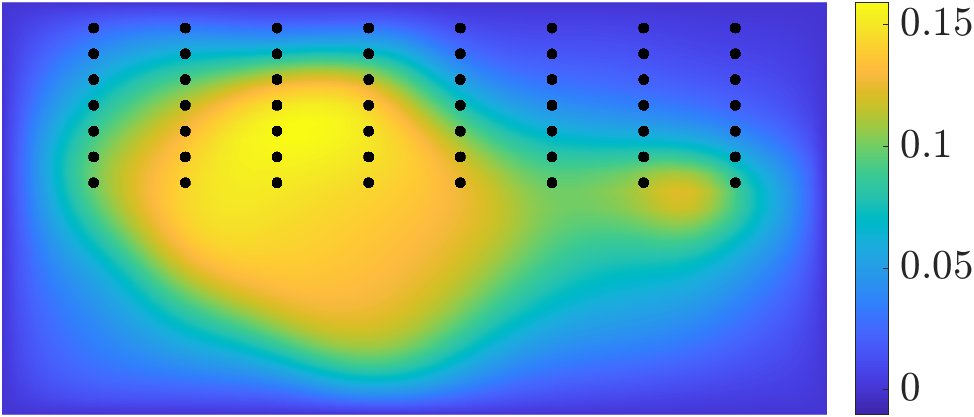}};
\node at (-.25,1.5) {$\z_{\rm true}$}; 
\node at (5.,1.5) {$\w_{\rm true}$}; 
\node at (10.55,1.5) {$u$};

\end{tikzpicture}
\caption{True parameters and data for Example 2. The true log-permeability $\z_{\rm true}$ (left), the true log-recharge $\w_{\rm true}$ (centre), as well as the resulting hydraulic-head $u$  (right). In each plot the black dots indicate the locations of the measurements.}\label{fig:prob2truths}
\end{figure}

In this example, we allow for the possibility of different correlations (between $\z$ and $\w$) in different subdomains. Specifically, we divide the domain $\Omega$ into two subdomains, $\Omega_1=\{(x,y)\in\Omega\,|\, x\leq1\}$ and $\Omega_2=\{(x,y)\in\Omega\,|\, x>1\}$, and allow for different correlations in each of these subdomains, i.e., we take 
\begin{align}
C=c(x)I,\quad c(x)=
\begin{cases}
c_1, & \text{for } x\leq1\\
c_2, & \text{for } x>1,
\end{cases}
\end{align}
where $|{c_1}|<1$ and $|c_2|<1$ are a priori unknown\footnote{Although not considered here, it is worth noting that both the number and geometry of the subdomains could also be treated as unknown, and estimated. This, however, would significantly increase the associated computational costs, and is beyond the scope of the current study.}

\subsection{Computational details for Example 2}
 We solve (\ref{eq: Darcy}) for $u$ using a Galerkin finite element method (FEM). More specifically, we represent the potential $u$ and (initially) the permeability $\z$ and recharge $\w$ using 1250 continuous piecewise linear basis function on a regular mesh of 1250 nodes and 2352 triangular elements. To reduce the computational costs associated with the MCMC sampling, we use a truncated KL expansion of the parameters, as outlined in Section \ref{sec: CompCon}. Specifically, we truncate the expansions for $\z$ and $\w$ after ${k_\z}=50$ and ${k_\w}=100$ (see Equation (\ref{eq:TKL})), respectively, which captures approximately $99\%$ and $94\%$  of the variation in each parameter, respectively. This results in a total number of parameters to be inferred of $50+100+2=152$. Similarly to the previous example, the true parameters are generated by setting a true correlation (in this case $c_{\rm true}=(0.8,-0.9)$) and then sampling $\z_{\rm true}$ and $\w_{\rm true}$ from the joint conditional distribution $\pi(\z,\w|c=c_{\rm true})$. Lastly, the data consists of 56 regularly spaced measurements of $u$ and 45 measurements of $\z$, see Figure \ref{fig:prob2truths}.  Additive noise  of the form $\ee\sim\mathcal{N}(0,\delta^2_1I)$ and 
$\eee\sim\mathcal{N}(0,\delta^2_2I)$ corrupts the measurements, where, for this example, we set  $\delta_1=5/100\times(\max(\mathcal{B}_1u(\z_{\rm true},\w_{\rm true}))-\min(\mathcal{B}_1u(\z_{\rm true})))$ and $\delta_2=2/100\times(\max(\mathcal{B}_2\z_{\rm true})-\min(\mathcal{B}_2\z_{\rm true}))$. That is, the hydraulic-head measurements are corrupted by 5\% noise while the direct measurements of $\z$ are corrupted by 2\% noise. The true parameters and resulting measurements are shown in Figure \ref{fig:prob2truths}.

To characterise the posterior distributions for Example 2 we use MCMC for both the independent and joint approaches, since the forward model is nonlinear. In the independent case we sample the KL coefficients only, while in the joint approach we sample the KL coefficients and $c_1$ and $c_2$, which as described in Section \ref{sec: CompCon}, we reparameterise as $c_1=\tanh(\gamma_1)$ and $c_2=\tanh(\gamma_2)$. The parameters $\gamma_1$ and $\gamma_2$ are updated using an uncorrelated Gaussian random-walk proposals, i.e., we propose 
\[\gamma_i'=\gamma_i^{(k)}+\xi_i,\quad \xi_i^{(k)}\sim\mathcal{N}(0,1),\quad i\in\{1,2\}.\]
As updating the correlations is cheap (compared to updating $\z$ and $\w$), in this example we take 100 steps of Metropolis-Hastings updates for the correlation parameters $\gamma_1$ and $\gamma_2$ before each update of $\z$ and $\w$. Since the dimensions of 
$\hat\z
$ and 
$\hat\w$ are significantly larger, we update them using an adaptive Metropolis step for improved computational efficiency~\cite{haario01}. Specifically letting $\hat s=(\hat\z,\hat\w)\in\mathbb{R}^{k_\z+k_\w}$ we use proposals of the form
\[\hat s'=\hat s^{(k)}+\tau^{(k)} A^{(k)} \zeta^{(k)},\quad \zeta^{(k)}\sim\mathcal{N}(0,I),\]
where $A^{(k)}$ is the Cholesky factor of the sample covariance matrix of the accepted samples, and $\tau^{(k)}$ is a scaling parameter. In line with~\cite{andrieu2008tutorial,haario01}, the scaling parameter is adapted during the sampling to push the acceptance ratio of the $(\z,\w)$ samples towards $0.23$ which has been shown to be optimal for high-dimensional applications of random-walk Metropolis in which the target and proposal distributions are Gaussian~\cite{gelman1997weak,andrieu2008tutorial}. 

For both the independent and joint approaches we generate $10^6$ samples and discard the first $2\times10^5$ as {\em burn-in}. Moreover, for both approaches we use a  `warm-start' for the MCMC. That is, we initialise the MCMC chains for $\hat\z$ and $\hat\w$ at the projection of joint MAP estimate found using $c_1=c_2=0$ of $\z$ and $\w$ onto to $\hat V$ and $\hat U$, respectively, where $\hat V=[v_1,v_2,\dots,v_{k_\z}]$ and $\hat U=[u_1,u_2,\dots,u_{k_\w}]$ (see Equations (\ref{eq: covs}) and (\ref{eq:TKL})). That is, we set 
\begin{align}
    (\hat\z^{(0)},\hat\w^{(0)})=&(\hat V^T(\z_{\rm MAP}-\z_\ast),\hat U^T(\w_{\rm MAP}-\w_\ast))\nonumber\\
    (\z_{\rm MAP},\w_{\rm MAP}):=&\arg\min_{(\z,\w)\in\mathbb{R}^n} \frac{1}{2}\norm[\Gamma_e^{-1}]{d-\f(\z,\w)}^2+\frac{1}{2}\norm[\Gamma_\z^{-1}]{\z-\z_\ast}^2+\frac{1}{2}\norm[\Gamma_\w^{-1}]{\w-\w_\ast}^2.\nonumber
\end{align}
On the other hand, we initialise the proposal covariance as a scaled version of the (Gauss-Newton) approximate posterior covariance matrix, i.e.,
\[A^{(0)}=L_{\rm post},\quad L_{\rm post}L_{\rm post}^T=\Gamma_{\rm post}:=(J(\z_{\rm MAP},\w_{\rm MAP})^T\Gamma_e^{-1}J(\z_{\rm MAP},\w_{\rm MAP})+\Gamma_0^{-1})^{-1},\]
where $\Gamma_0$ is the block diagonal joint prior covariance matrix (i.e., with $\CC=0$) and $J(\z_{\rm MAP},\w_{\rm MAP})$ is the Jacobian of $\f$ evaluated at the joint MAP estimate. The MAP estimates and approximate marginal posterior uncertainty are shown in Figure~\ref{fig:prob2LAP}. 
\begin{figure}[t!]
\centering
$\z$ \hrulefill\vspace{5pt}

\includegraphics[height=0.145\textwidth]{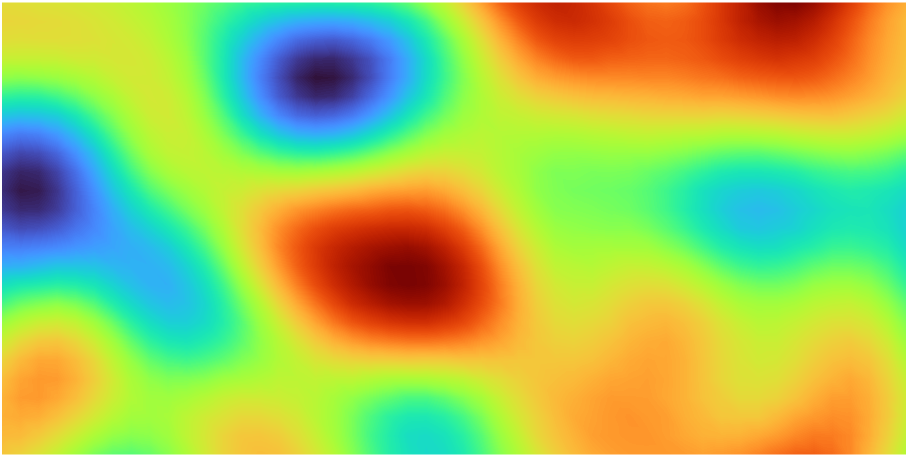}
\hfill
\includegraphics[height=0.145\textwidth]{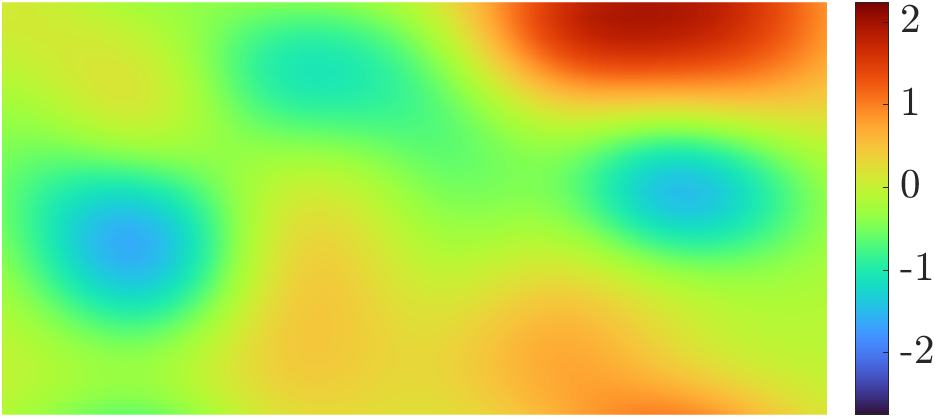}
\hfill
\includegraphics[height=0.145\textwidth]{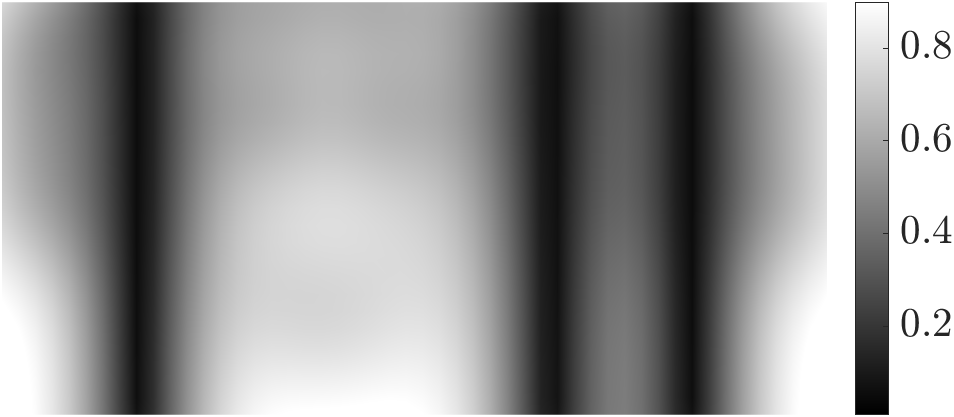}\\
$\w$ \hrulefill\vspace{5pt}

\includegraphics[height=0.145\textwidth]{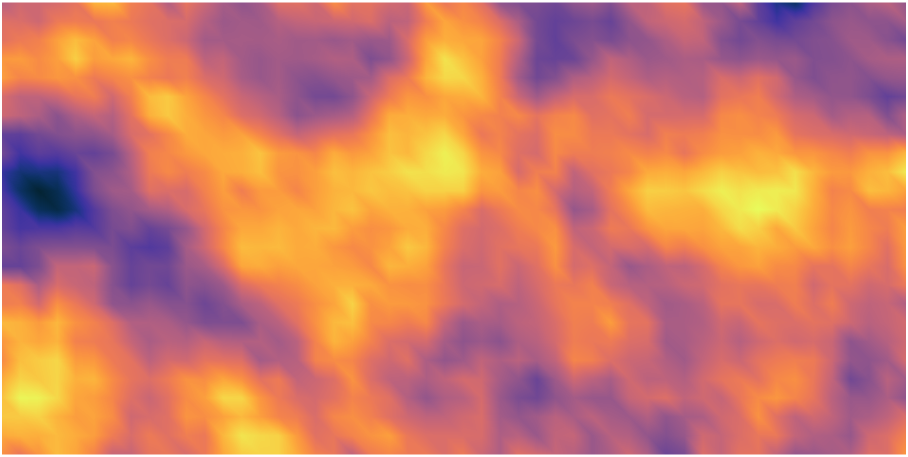}
\hfill
\includegraphics[height=0.145\textwidth]{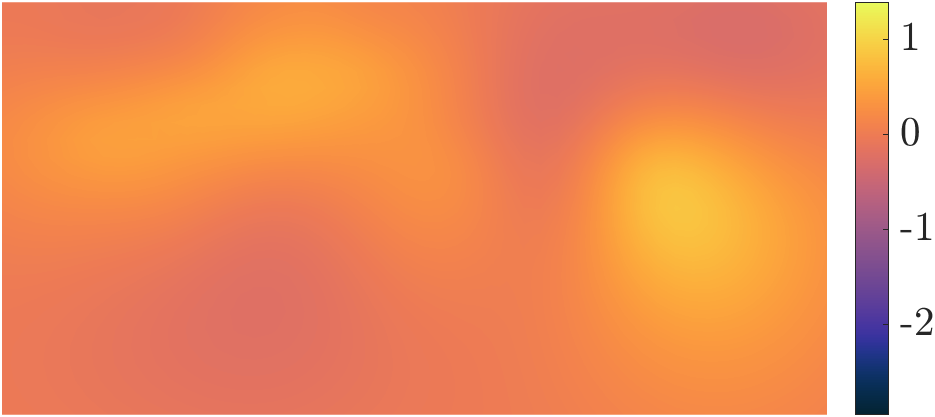}
\hfill
\includegraphics[height=0.145\textwidth]{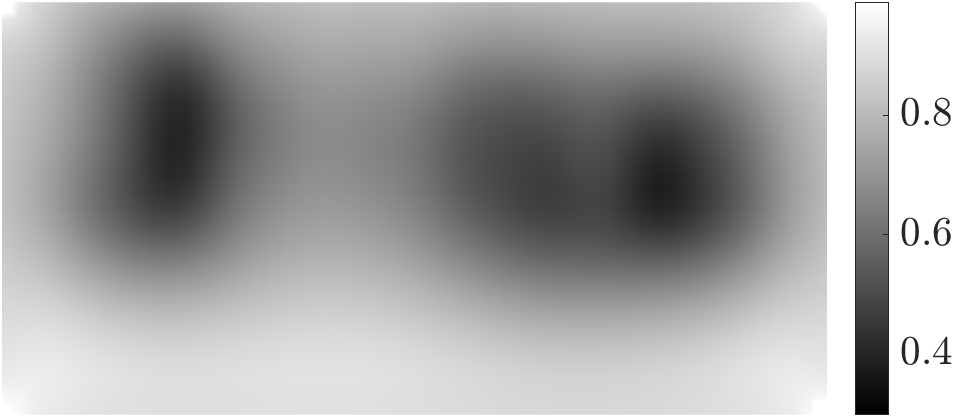}
\caption{Laplace approximation used to `warm-start' the MCMC sampling for Example 2. In the top row we show the true log-permeability $\z_{\rm true}$ (left), the MAP estimate (centre) and the approximate point-wise marginal posterior standard deviation (right).}\label{fig:prob2LAP}
\end{figure}

\subsection{Results for Example 2}\label{sec: res2}

In Figure \ref{fig:prob2CMs} we show the true log-permeability and log-recharge together with the conditional mean estimates obtained using independent and joint inference. As in Example 1, the independent approach recovers the parameters primarily in regions directly informed by the data, namely near the well paths for $\z$
and through the hydraulic-head observations for the coupled system. By contrast, the proposed joint approach allows information to be transferred between the two parameters through the inferred prior correlation, leading to improved reconstructions of both $\z$ and 
$\w$, particularly for the recharge field $\w$ at the well locations (where $\z$ was measured directly).

\begin{figure}[t!]
\centering
$\z$ \hrulefill\vspace{5pt}

\includegraphics[height=0.15\textwidth]{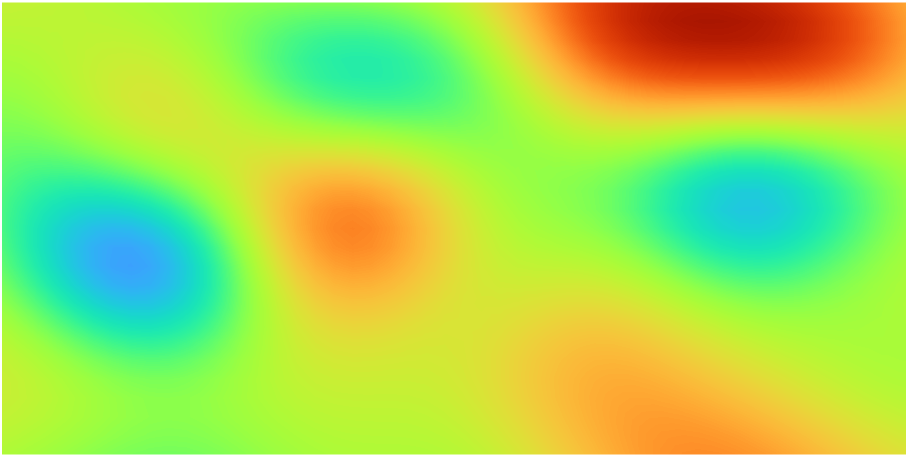}
\hfill
\includegraphics[height=0.15\textwidth]{figures/GeoSigmaNoCB.png}
\hfill
\includegraphics[height=0.15\textwidth]{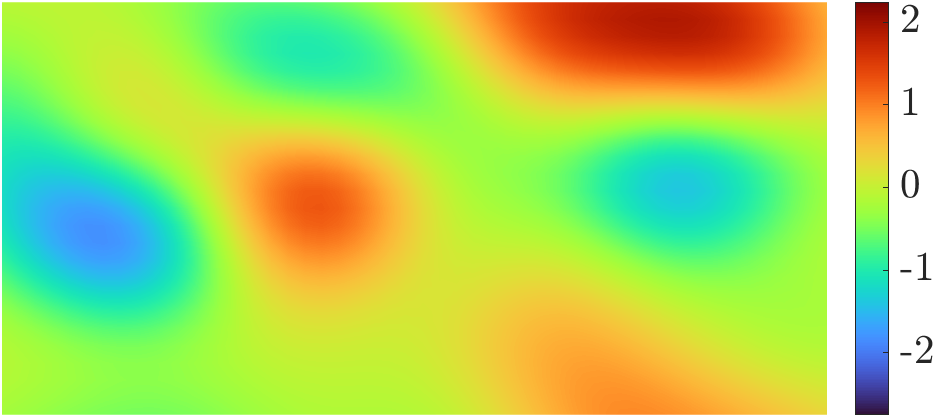}\\
$\w$ \hrulefill\vspace{5pt}

\includegraphics[height=0.15\textwidth]{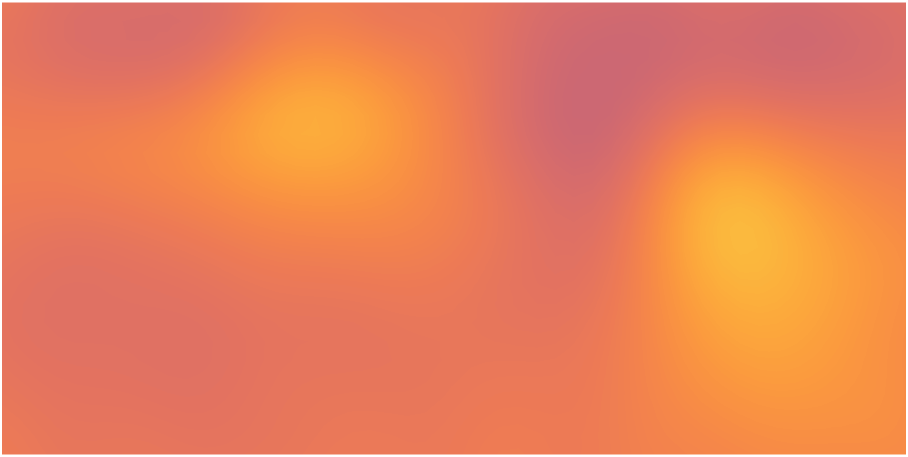}
\hfill
\includegraphics[height=0.15\textwidth]{figures/GeoKappaNoCB.png}
\hfill
\includegraphics[height=0.15\textwidth]{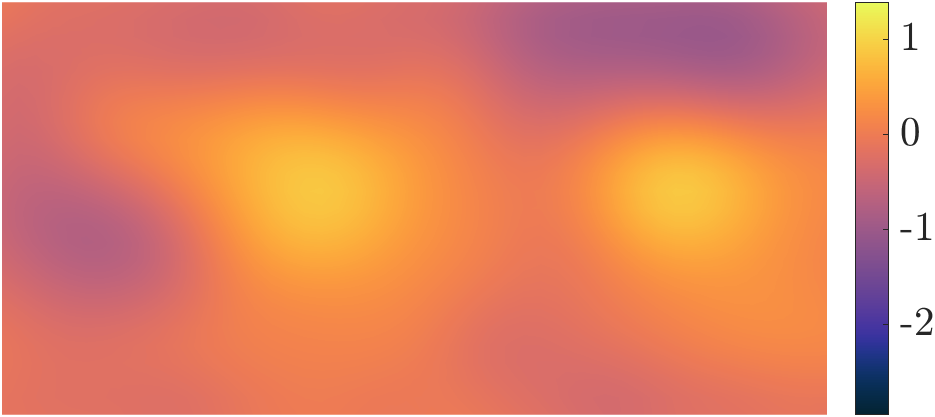}
\caption{True parameters and CM estimates for Example 2. In the top row we show the CM found using independent inference (left), the true log-permeability $\z_{\rm true}$ (centre), and the CM found using joint inference (right). On the bottom row we show the corresponsing plots for the log-recharge $\w$.}\label{fig:prob2CMs}
\end{figure}

\begin{figure}[t!]
\centering
$\z$ \hrulefill\vspace{5pt}

\includegraphics[height=0.146\textwidth]{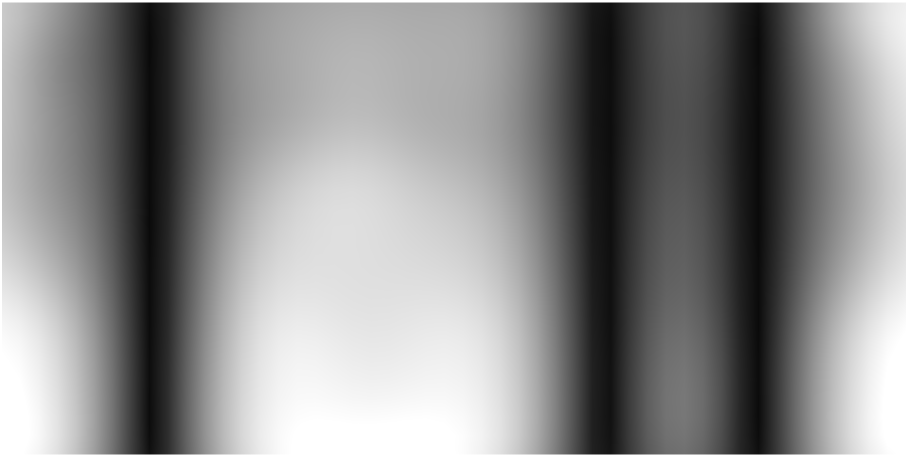} 
\hfill
\includegraphics[height=0.146\textwidth]{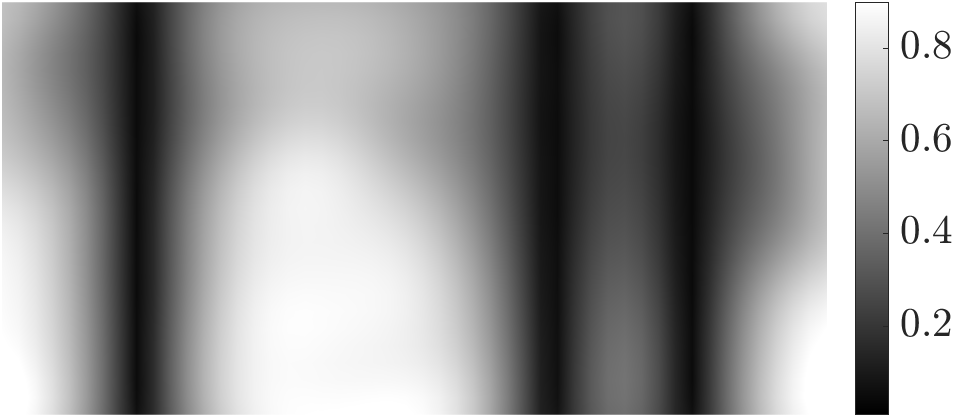}
\hfill\includegraphics[height=0.146\textwidth]{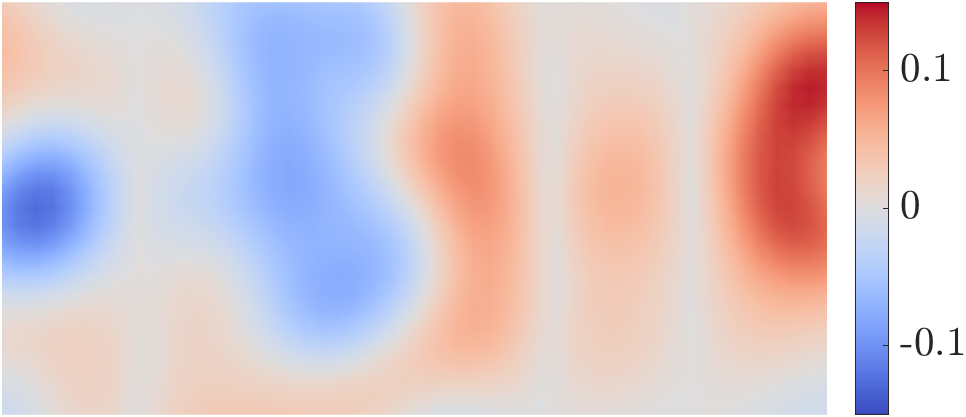}
        \\

      $\w$ \hrulefill\vspace{5pt}
      
\includegraphics[height=0.146\textwidth]{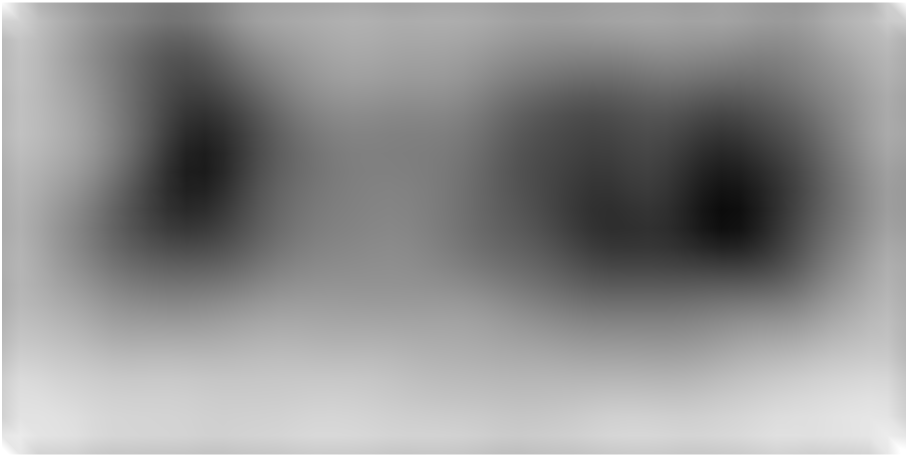}
\hfill
\includegraphics[height=0.146\textwidth]{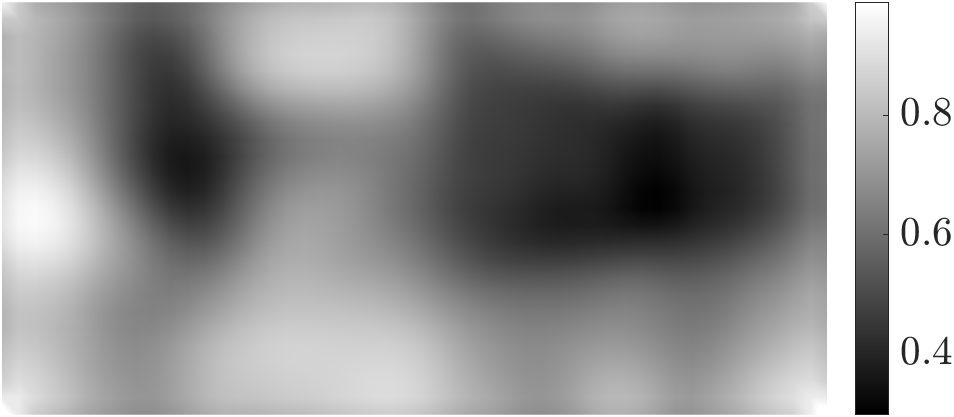}
\hfill
\includegraphics[height=0.146\textwidth]{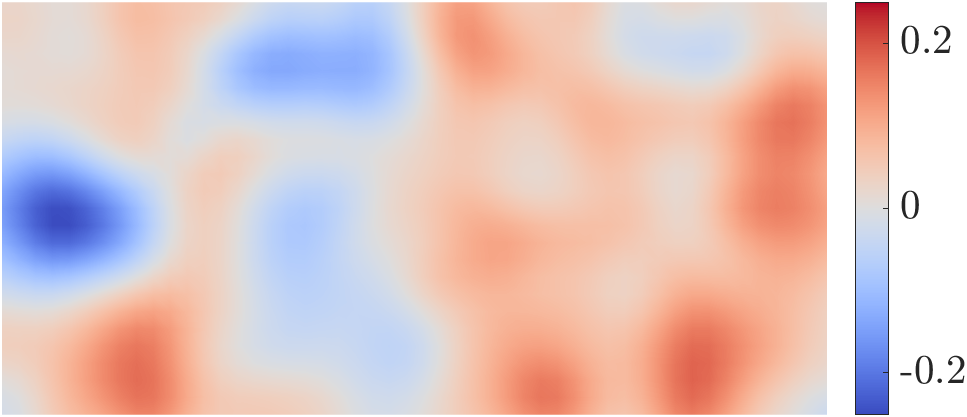}
\caption{The reduction in uncertainty for Example 2. On the top we show the point-wise marginal posterior standard deviation for $\z$ found using independent inference (left), the point-wise marginal posterior standard deviation for $\z$ found using joint inference (centre), and the difference of the two, $\mathtt{D}_\z$ (right). On the bottom row we show the corresponding plots for $\w$.}\label{fig:Prob2UQ}
\end{figure}

In Figure \ref{fig:Prob2UQ} we show the marginal posterior standard deviations obtained using independent and joint inference. In both cases the uncertainty is reduced most strongly near the data locations. However, under joint inference the reduction is not confined to those regions, the direct observations of $\z$ are also informative for (and reduce the uncertainty in) $\w$. In this example we have  $\mathtt{E}(\z)=0.597$ and $\mathtt{E}(\w)=0.975$,  when using  independent inference and $\mathtt{E}(\z)=0.567$ and $\mathtt{E}(\w)=0.790$ when using the joint approach. That is to say, in terms of reconstruction accuracy, the joint approach yields only a modest improvement for $\z$ but a much more substantial improvement for $\w$. This is consistent with the model setup, since 
$\z$ is informed by both the hydraulic-head data and the direct point-wise measurements, whereas 
$\w$ is only observed indirectly through the PDE model and thus benefits more from the inferred correlation between the two parameters.

In terms of uncertainty reduction, we obtain $\mathtt{U}(\z)=0.311$ and $\mathtt{U}(\w)=0.494$ using independent inversion, and 
$\mathtt{U}(\z)=0.314$ and
$\mathtt{U}(\w)=0.459$ when using the joint approach. 
Thus, the joint approach produces a noticeable reduction in posterior uncertainty for 
$\w$, while the uncertainty in 
$\z$ remains at a similar level. This is again consistent with the model setup, i.e.,  
$\w$ benefits from the additional information transferred through the joint model.

Finally, Figure~\ref{fig:C2s} shows the histograms of the inferred correlation parameters $c_1$ and $c_2$. The posterior clearly supports different correlations in the two subdomains $x\leq1$ and $x>1$, with mass concentrated on positive values for $c_1$ and negative values for $c_2$, in agreement with the true values ($c_{\rm true}=(0.8, -0.9)$). The posterior distributions are broader than in Example 1, likely a result of the nonlinearity and smoothing behaviour of the forward problem. The trace plots and  autocorrelation plot in Figure~\ref{fig:C2diag} indicate that both $c_1$ and $c_2$ are fairly well explored throughout the run, although  
$c_2$ appears to exhibits stronger autocorrelation. The next section provides further discussion on the efficiency and effectiveness of the approaches.

\begin{figure}[t!]
\centering
\includegraphics[height=0.3\textwidth]{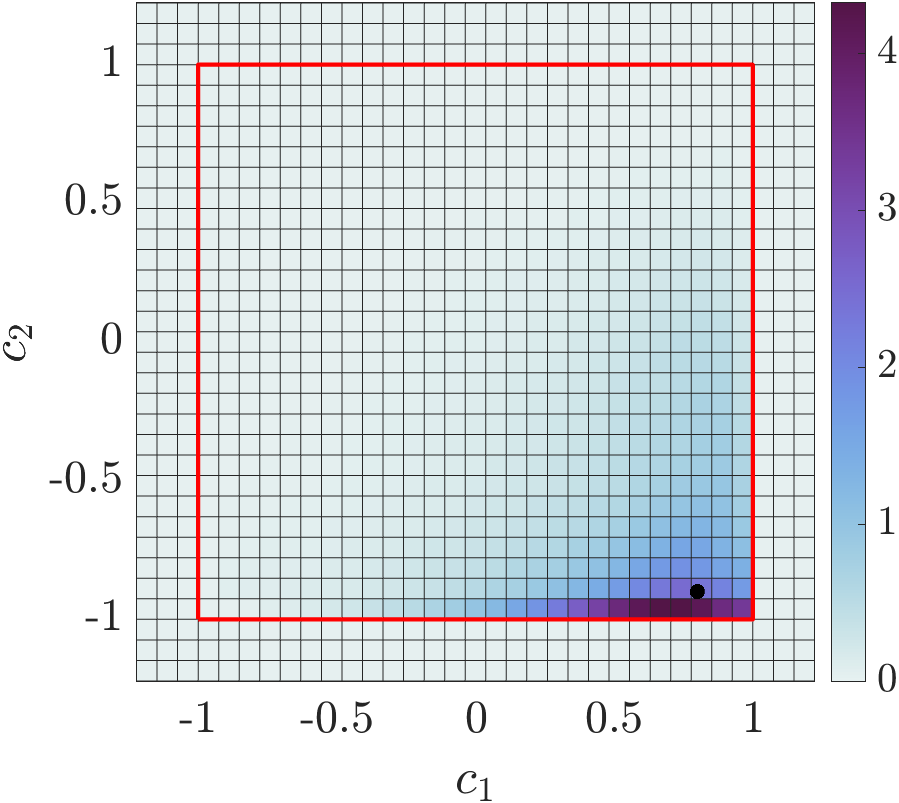}
\hfill
\includegraphics[height=0.3\textwidth]{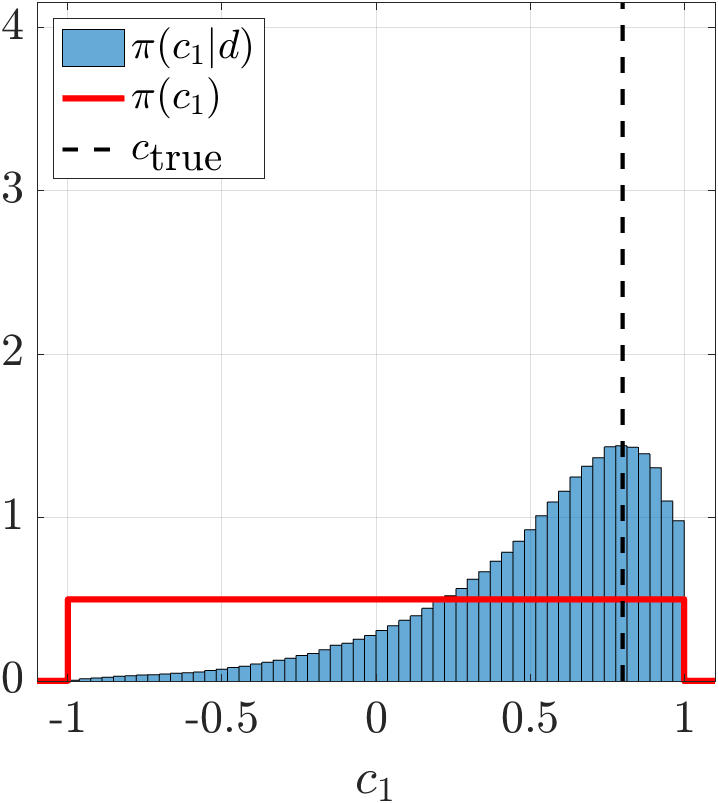}
\hfill
\includegraphics[height=0.3\textwidth]{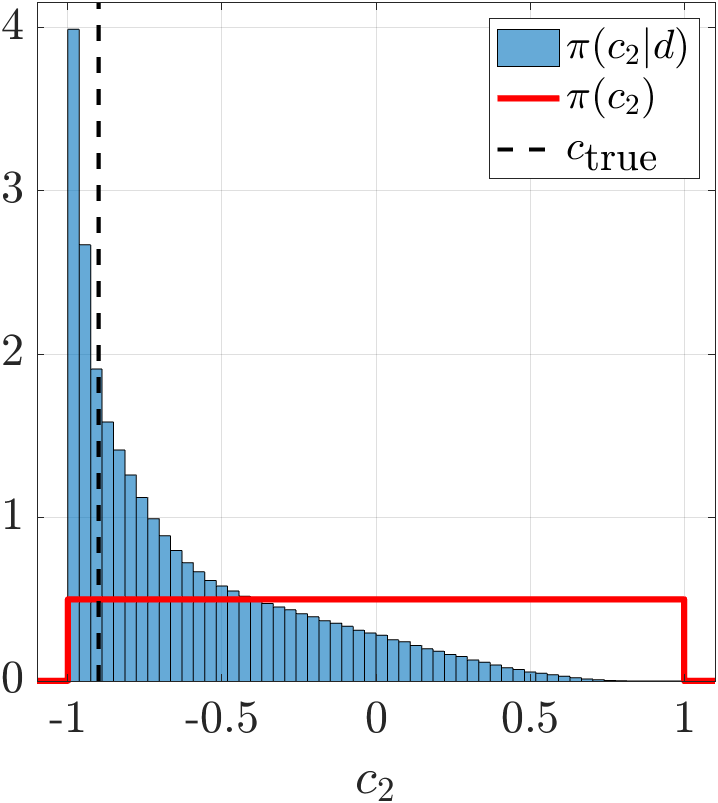}
\caption{Results for estimating the correlation in Example 2. We show the joint marginal posterior $\pi(c_1,c_2|d)$ (left) as well as the marginal posteriors of $\pi(c_1|d)$ (centre) and $\pi(c_2|d)$ (right). In each of the plots the red line is used to indicate the (support) of the relevant prior distribution, while the black dot in the joint marginal plot (right) and black dotted lines in the marginal plots shows the true values.}\label{fig:C2s}
\end{figure}
\begin{figure}[t!]
\centering
\includegraphics[height=0.3\textwidth]{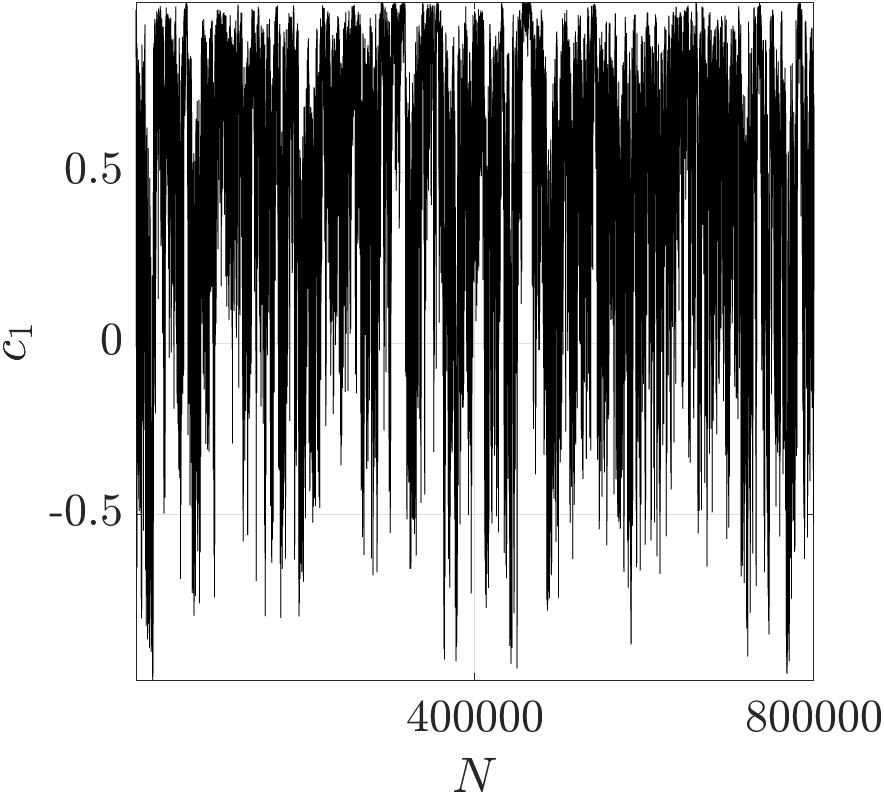}
\hfill
\includegraphics[height=0.3\textwidth]{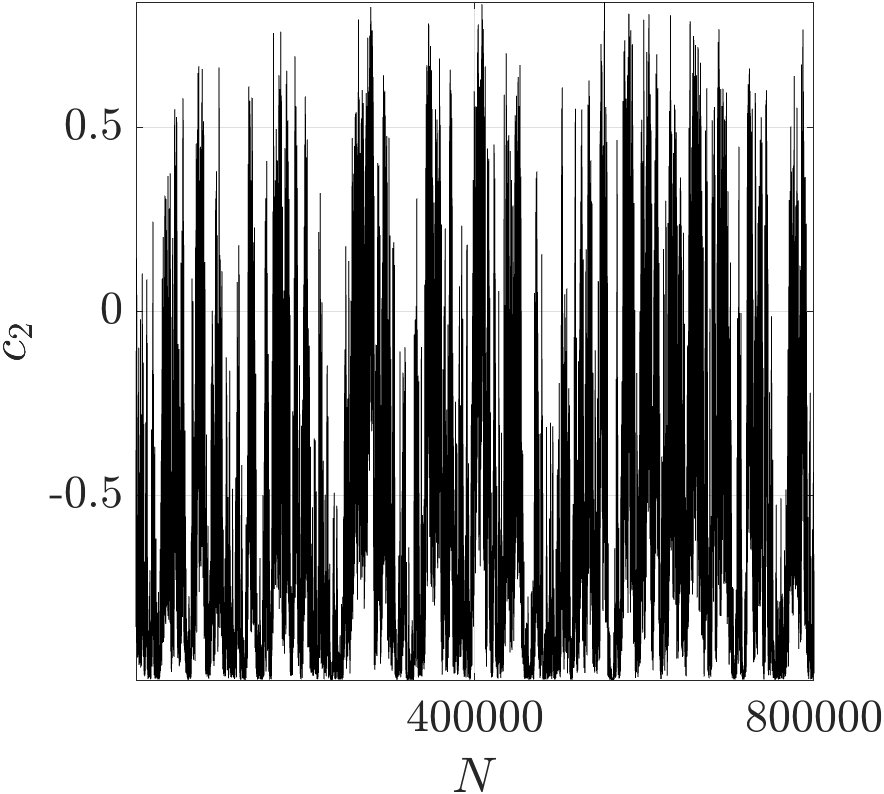}
\hfill
\includegraphics[height=0.3\textwidth]{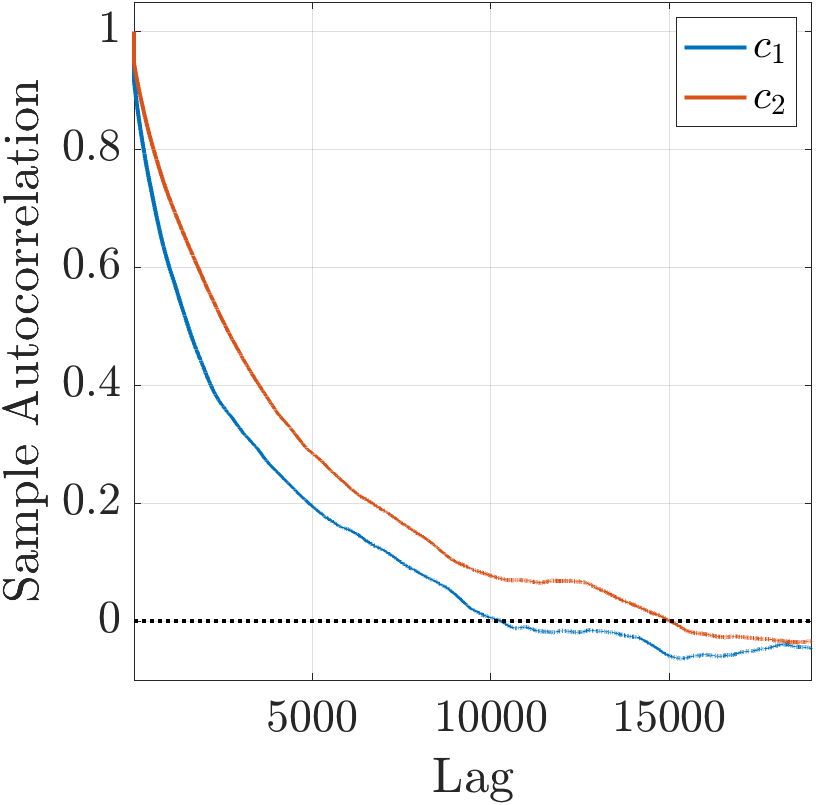}

\caption{MCMC diagnostics for the correlation parameters 
$c_1$ and $c_2$ in Example 2. The left and centre panels show the trace plots for 
$c_1$ and $c_2$, respectively, while on the right we show the corresponding sample autocorrelation functions.}\label{fig:C2diag}
\end{figure}

\begin{figure}[t!]
\centering
Independent \hrulefill\vspace{5pt}

\includegraphics[height=0.3\textwidth]{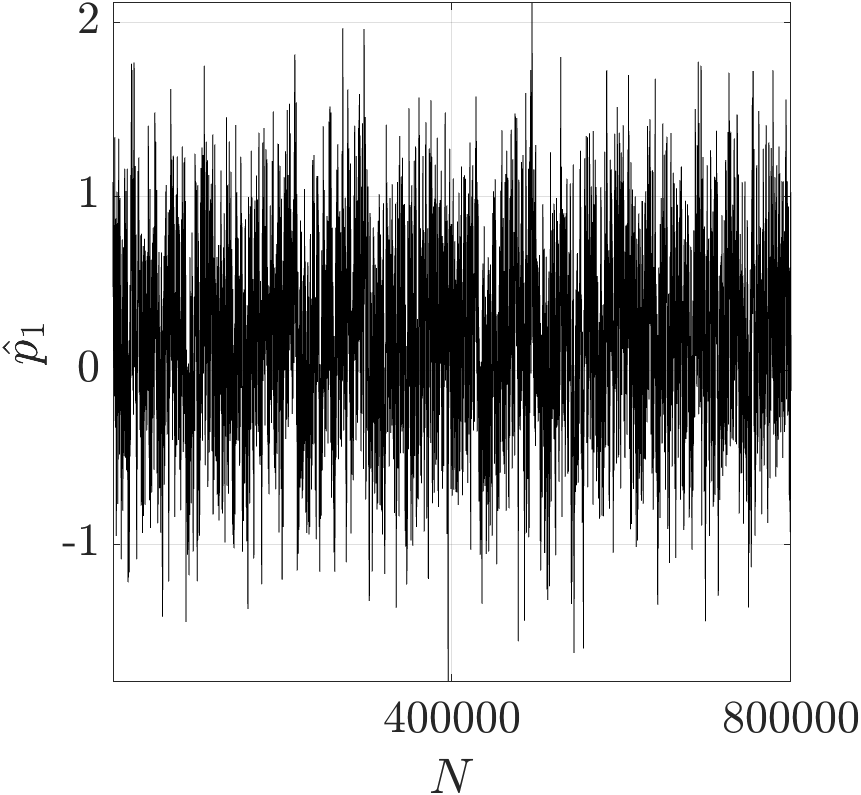}
\hfill
\includegraphics[height=0.3\textwidth]{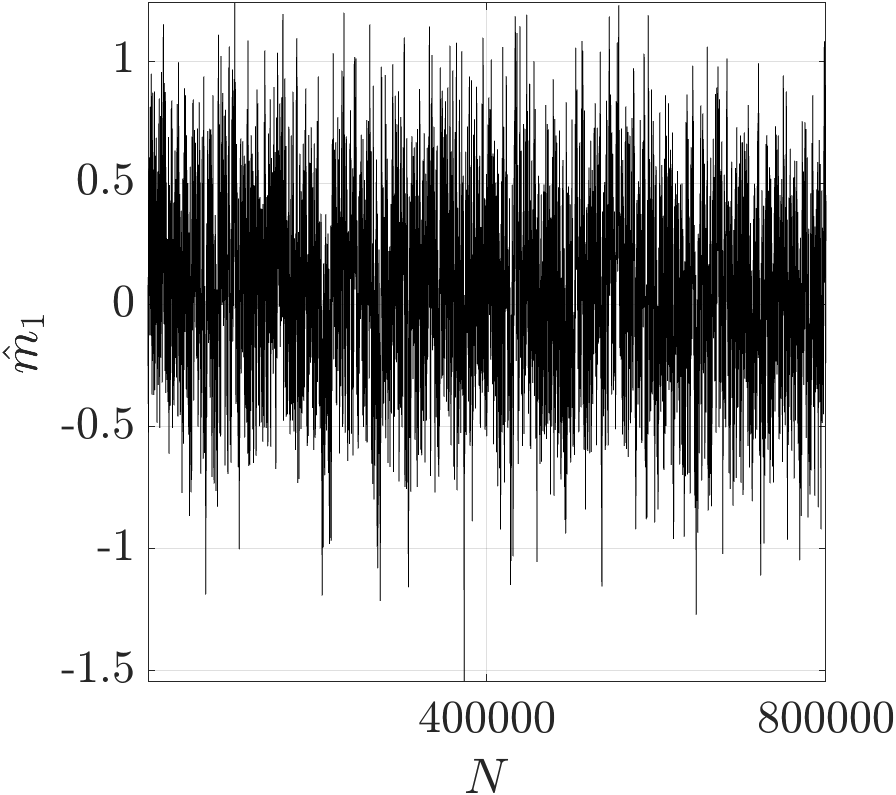}
\hfill
\includegraphics[height=0.3\textwidth]{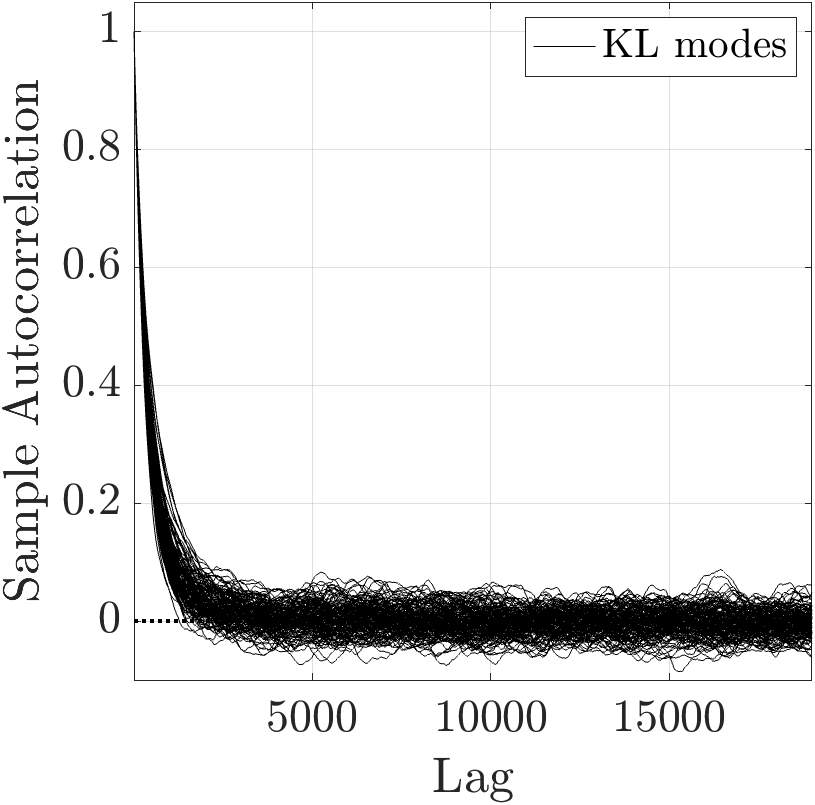}\\
Joint \hrulefill\vspace{5pt}

\includegraphics[height=0.3\textwidth]{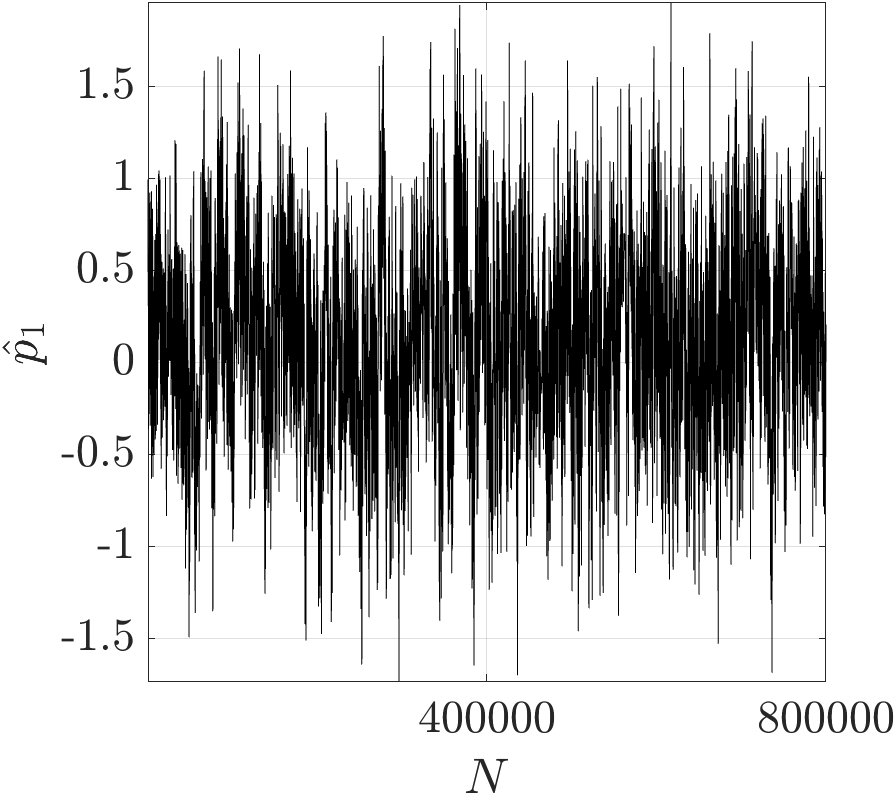}
\hfill
\includegraphics[height=0.3\textwidth]{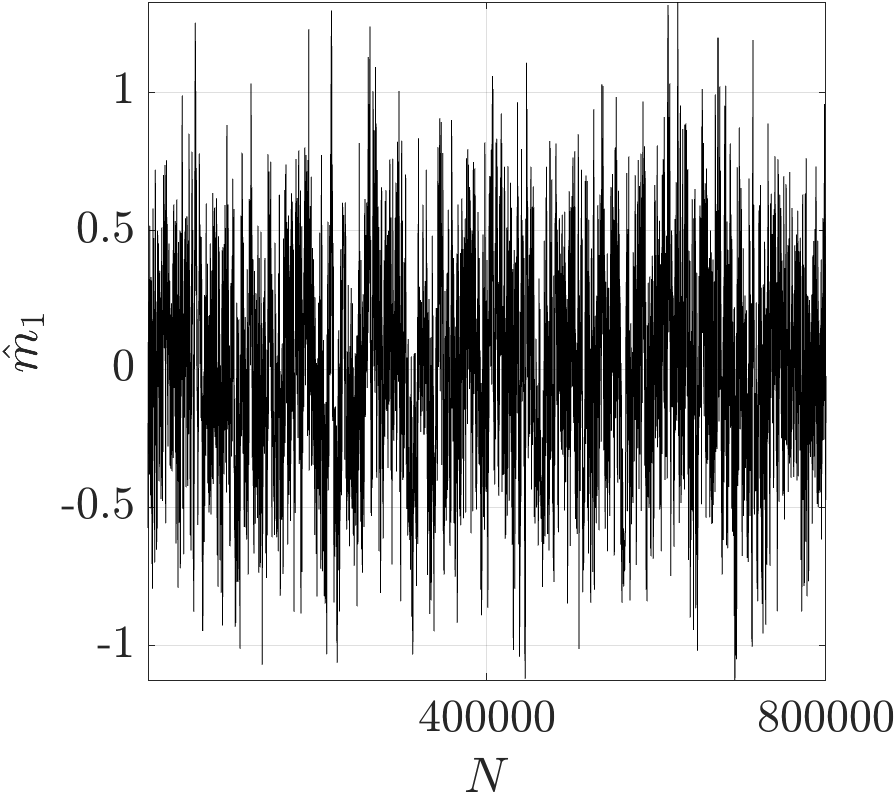}
\hfill
\includegraphics[height=0.3\textwidth]{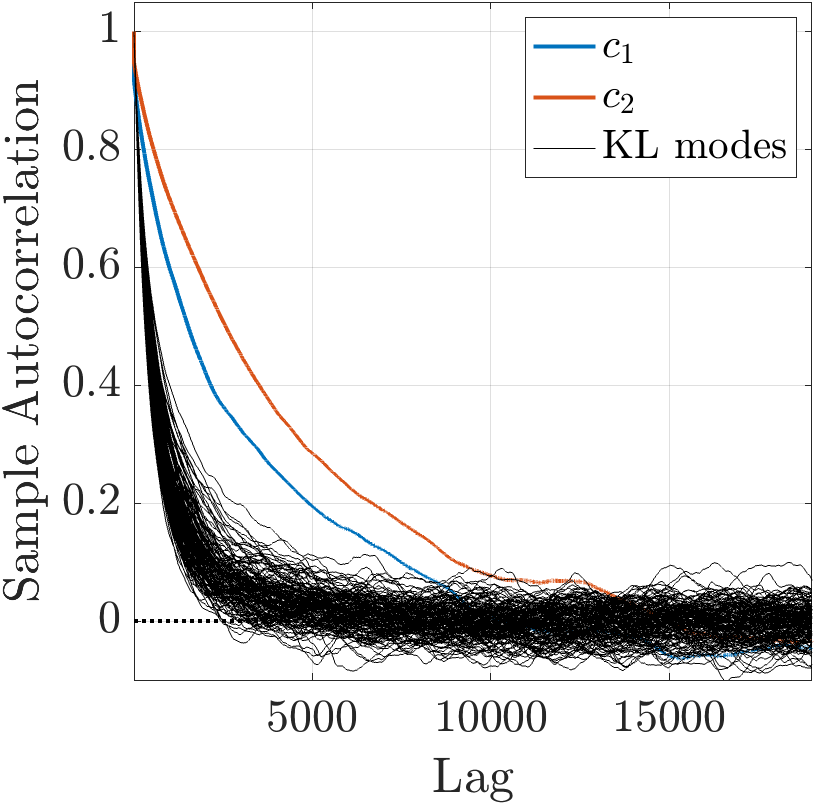}
\caption{MCMC diagnostics for the correlation parameters 
for the KL modes $\hat\z$ and $\hat\w$ in Example 2 using independent and joint inference. On the top row we show the trace plot of the first KL mode for the log-permeability $\hat\z_1$ using independent inversion (left), the trace plot of the first KL mode for the log-recharge $\hat\w_1$ using independent inversion and the sample autocorrelation for all KL modes using independent inversion. On the bottom we show the corresponding plots when using joint inversion and also show the sample autocorrelation for $c_1$ and $c_2$ in the autocorrelation plot (right).}\label{fig:Prob2Traces}
\end{figure}

\subsection{Computational costs}\label{sec: CompCost}

In this section we briefly consider the computational costs for Example 1 and Example 2. Specifically, the costs associated with characterising the posterior distributions using independent inference and joint inference. The key results are summarised in Table \ref{table}. For reference, all experiments were executed using a 10-core Apple M1 Pro with 32 GB RAM, with the forward model in Example 2 taking in the order of $1.0\times 10^{-3}$ seconds to evaluate. 

The additional cost of evaluating the posterior when using the proposed joint approach is negligible compared to the costs of evaluating the (forward models in the) likelihood, which is identical for the joint and independent approaches. That is to say, the cost of evaluating the posterior is (essentially) the same for the joint and independent approaches. Table \ref{table} gives the key details on the results and the computational costs associated  with carrying out the independent and joint inversion for both of the computational examples. As discussed in Sections \ref{sec: res2} and \ref{sec: res1}, for the examples considered here, employing the proposed joint approach results in reductions in the posterior uncertainty and generally improves the CM estimates. To compare the computational costs of the proposed joint approach and independent inversion, in Table \ref{table} we report the effective sample sizes (ESS) for the correlations ($c$ for Example 1 and $c_1$ and $c_2$ for Example 2) as well as the median ESS for $\hat\z$ and the median ESS for $\hat\w$. The effective sample size is given by
\begin{align}\label{eq: ESS}
    \text{ESS}=\frac{M}{1+2\sum_{w=1}^\infty r(w)}
\end{align}
where $M$ is the total number of retained samples (after burn-in is discarded) and $r(w)$ is the estimated autocorrelation function for a lag $w>0$. In practice the autocorrelation typically becomes noisy for large lags, and the summation in~(\ref{eq: ESS}) is truncated -- here we truncate the summation at the first lag for which the autocorrelation goes negative. Since the ESS is defined for each scalar parameter individually, we evaluate it for the correlation parameter(s) and for the coefficients of the first KL modes of both parameters.

For Example 1, as both forward problems are linear, characterising the posterior using independent inference does not require MCMC and we thus do not report an ESS. On the other hand, using the proposed joint approach leads to a non-Gaussian posterior, which we characterise using MCMC. We compute a chain of length $M=9.9\times 10^4$ retained samples which has as an ESS for the correlation $c$ of 14818. The estimated autocorrelation function for $c$ is shown in Figure~\ref{fig:C1}.

For Example 2, we use MCMC for both the independent and joint approaches. For both approaches we compute a chain of length $M=8\times10^5$ retained samples. Using the independent approach, the MCMC chain for coefficient of the first KL mode of $\z$ (resp. $\w$) , i.e., $\hat\z_1$ (resp. $\hat\w_1$), has an ESS of 904 (resp. 804). Using the joint approach resulted in ESS of 455 and 493 for $\hat\z_1$ and $\hat\w_1$, respectively, while for $c_1$ and $c_2$ the ESS are 153 and 106, respectively. The trace plots for $\hat\z_1$ and $\hat\w_1$ and the estimated autocorrelation function for all of the KL-mode coefficients using both approaches are shown in Figure~\ref{fig:Prob2Traces}. The reduction in ESS for the joint approach indicates that, although explicitly accounting for uncertainty in the correlation can lead to more accurate uncertainty quantification, it incurs an additional computational cost, likely due to the increased complexity of the posterior distribution.

\begin{table}[t!]
\caption{Key metrics for Example 1 and 2 when solving independently and jointly using the proposed approach. Note for Numerical Example 1 both forward models are linear, and thus the associated independent inverse problems can be carried out directly (i.e., without sampling). Reported are the relative error for $\z$ and $\w$ ($\mathtt{E}(\z)$ and $\mathtt{E}(\w)$), the relative uncertainty reduction for $\z$ and $\w$ 
 ($\mathtt{U}(\z)$ and $\mathtt{U}(\w)$), the total dimension of the parameters to be inferred (Dim.), the number of MCMC samples after discarding the burn-in (M), the effective sample size for the correlation coefficient(s) (ESS -- $\CC$), and the median effective sample size for $\z$ and $\w$ (ESS -- $\z$ and ESS -- $\w$, respectively). Note for Example 2 ESS -- $\z$ and ESS -- $\w$ are the median effective sample sizes of the KL mode coefficients $\hat\z$ and $\hat\w$.}\label{table}
\begin{center}
\begin{tabular}{ |c|c|c|c|c|c|c|c|c|c|c| }
  \hline 
  & $\mathtt{E}(\z)$ & $\mathtt{U}(\z)$ &  $\mathtt{E}(\w)$ & $\mathtt{U}(\w)$ & Dim. &  $M$ & \multicolumn{2}{|c|}{ESS -- $\CC$} & ESS -- $\z$ & ESS -- $\w$
\\ 
 \hline
 & \multicolumn{10}{c|}{Numerical Example 1}\\
 \hline
  Ind. $\z$ & 0.852 & 0.401 &  - & -  & 1250  & - & \multicolumn{2}{c|}{-} & - & -
\\ 
\hline
  Ind. $\w$  & - & - &  0.629 & 0.411  & 1250  & - & \multicolumn{2}{c|}{-} & - & -
\\ 
\hline
  Joint & 0.513 & 0.313  &   0.589 & 0.291 & 2501  & $9\times 10^4$ &  \multicolumn{2}{c|}{14818} & 29538 & 31965
\\ 
\hline
 & \multicolumn{10}{c|}{Numerical Example 2}\\
 \hline
  Ind. & 0.597 & 0.311 &  0.975 & 0.494 & 150  & $8\times 10^5$ & \multicolumn{2}{c|}{-} & 904 & 804
\\ 
\hline
  Joint  & 0.567 & 0.314 &  0.790 & 0.459 & 152  & $8\times 10^5$ & 153 & 106 & 455 & 493
\\
\hline
\end{tabular}
\end{center}
\end{table}

\section{Discussion and Conclusion}
\label{sec: Conc} In the present work, we proposed a novel approach for carrying out joint inference in the Bayesian framework. The approach is applicable to cases where the marginal prior distribution for each of the parameters is Gaussian and fixed, and we wish to have a jointly normal prior. As such the construction of the joint covariance is the key task. We presented a possible covariance structure which is optimal in the canonical correlation sense, and can easily encode cross-correlation between the two parameters. Before carrying out joint inversion we considered several examples of generating random fields with given cross-correlations, including spatially varying cross-correlation, and random fields in different spatial dimensions.  

We considered joint inversion on a toy Monod model example to highlight the effects of miss-specifying the cross-correlation. This example also served to illustrate that the proposed approach is not, at least in its current form, immediately amenable to optimisation-based approaches such as computing the (joint) MAP estimate. In our case, that is, simultaneous estimation of the correlation, we advocate for the use of sampling-based approaches such as MCMC. Our approach can easily accommodate dimension reduction of the two (or more) parameters of interest via truncated KL expansions which can go some way to reducing the computational overheads associated with carrying out sampling-based approaches.

Finally, we considered two higher-dimensional problems, one of which is governed by a PDE, where we estimated two random fields as well as the cross-correlation. In both examples, allowing for uncertainty in, and estimation of, the cross-correlation led to improved CM estimates and more concentrated marginal posteriors. We stress, however, that the aim is not simply to reduce posterior variance, but to accurately model and quantify the remaining uncertainty in the unknown parameters and their correlations.

The overall conclusion of this work is a reaffirmation of a key component of the Bayesian approach to inverse problems: {\em anything unknown should be treated as a random variable.} In particular, if the correlation between the two parameters is unknown, we should (at least attempt to) model and account for this uncertainty.

\paragraph{Limitations} 
Of course, the approach we have presented also has some limitations. The two main limitations being related to computational overheads: 
\begin{enumerate}
    \item First, as illustrated in the Monod example (see Section \ref{sec:MONOD}) and further discussed in Appendix \ref{sec: appMAP}, the approach proposed here is not immediately amenable to common optimisation-based approaches, i.e., the MAP estimate and Laplace approximation represent the posterior particularly poorly. As such we advocate for using sampling-based methods such as MCMC as a means to accurately characterise the posterior. However it is well known that such methods may in some cases be computationally infeasible. A possible means to reconcile the Laplace approximation approach may be to consider a suitable weighting on the correlation (hyper-) parameters, see for example ~\cite{calvetti2019hierachical} and the references therein.
    \item Secondly, as alluded to in the Example 1 (see Section \ref{sec ex1}), even when the joint prior is Gaussian and the forward models are linear, using the proposed approach to jointly estimate the parameters and the correlation $\CC$ results in a non-Gaussian posterior. Thus, sampling-based methods are required to characterise the posterior which can be computationally prohibitive. This represents a considerable computational cost increase, even when using the proposed MwG approach.
\end{enumerate}

Other than addressing the limitations noted above, future work directions could include using the proposed approach to generate two (or more) correlated level-set functions as a means or encoding structural similarity (rather than correlation) in the Bayesian framework.

\bibliographystyle{unsrt}
\bibliography{jointRefs}

\appendix

\section{Illustration of the pitfall of optimisation when estimating the correlation}\label{sec: appMAP}
Here we provide some further insight into why optimisation-based approaches to estimating both parameters $\z$ and $\w$ as well as the correlation $\CC$ may be misleading. The issue is (possibly un-surprising, as the likelihood is independent of $\CC$) related to the prior. To illustrate this we set $n_1=n_2=1$, i.e., both parameters are scalar $\z,\w\in\mathbb{R}$, and assume  $\pi(\z)=\pi(\w)=\mathcal{N}(0,1)$, with $C=c\in\mathbb{R}$ satisfying $|c|<1$. If no other prior knowledge of the correlation between $\z$ and $\w$ is available, it is natural to postulate the uniform prior $\pi(c)=\mathcal{U}(-1,1)$. In this case the joint prior can be written
\begin{align}
\pi(\z,\w,c)\propto\mathds{1}_{(-1,1)}(c)\exp\left\{-\frac{1}{2}
\begin{pmatrix}
    \z & \w 
\end{pmatrix}
\begin{pmatrix}
    1 & c\\ c &1 
\end{pmatrix}^{-1}
\begin{pmatrix}
    \z \\ \w 
\end{pmatrix}-\frac{1}{2}\ln(1-c^2)\right\},
\end{align}
where $\mathds{1}_{(-1,1)}$ denotes the indicator function over the interval $(-1,1)$. Consider now the log-prior;
\[V(\z,\w,c):=\ln\left(\pi(\z,\w,c)\right)\propto-\frac{1}{2}
\begin{pmatrix}
    \z & \w 
\end{pmatrix}
\begin{pmatrix}
    1 & c\\ c &1 
\end{pmatrix}^{-1}
\begin{pmatrix}
    \z \\ \w 
\end{pmatrix}-\frac{1}{2}\ln(1-c^2),\quad |c|<1,\]
which has gradient
\begin{align}
    \nabla V=
    \frac{1}{(c^2-1)}\begin{pmatrix}
    (\z-c\w) \\ 
    (\w-c\z)\\
(c^3-c+(\z - c\w)(\w - c\z))/(c^2-1)
\end{pmatrix},\quad |c|<1.
\end{align}
The only stationary point of $V$ is at $(0,0,0)$ which is a saddle point, the Hessian being:
\begin{align}
\nabla^2 V(0,0,0)=\begin{pmatrix}
    -1 & 0 & 0 \\ 
    0 & -1 & 0 \\
0 & 0 & 1
\end{pmatrix}.
\end{align}
For any value of $|c|<1$, the level curves of $V$ (and/or $\pi$) are ellipses, while fixing $\z$ and $\w$ we can always increase $V$ (and thus $\pi$) by taking $c\to\pm 1$. Thus, if the data of the forward model(s) is not informative {\em enough}, an optimisation-based approach will return $c=\pm 1$ (or arbitrarily close to it).

\section{Invariance of the marginal posteriors to sign of correlation}\label{app:sign}

To clarify the behaviour observed in Example 1 in the linear-Gaussian setting (see Section~\ref{sec: res1}), we examine how the sign of the cross-correlation matrix 
$\CC$ affects the joint prior precision matrix. In particular, we show that replacing 
$\CC$ by $-\CC$ leaves the diagonal blocks of the prior precision matrix unchanged, while reversing the sign of the off-diagonal blocks. This symmetry underlies the corresponding invariance of the marginal posterior covariance matrices in the separable linear setting with block-diagonal noise covariance.
\begin{lemma}\label{lem: pri}
    Let $\z$ and $\w$ be jointly Gaussian with joint covariance matrix
\begin{align}\label{eq: appGam}
    \mathbb{S}^+_n\ni\GGamma=\begin{pmatrix}
\GGamma_\z & \LL_\z^{-1}\CC\LL_\w^{-T}\\
 \LL_\w^{-1}\CC^{T}\LL_\z^{-T} & \GGamma_\w
 \end{pmatrix},\quad C\in\mathbb{A}^{n_1\times n_2}
.
\end{align}
Then the diagonal blocks of the prior precision operator are invariant under the transformation $\CC\mapsto-\CC$, while the off-diagonal blocks reverse sign.
\end{lemma}

\begin{proof}
First, let the joint prior covariance matrix be written as $\GGamma$ in (\ref{eq: appGam}). Then the joint prior precision matrix is given by
\[
\GGamma^{-1}=\begin{pmatrix}
\LL_\z^T(I-\CC\CC^{T})^{-1}\LL_\z & -\LL_\z^T\CC(I-\CC^{T}\CC)^{-1}\LL_\w \\
 -\LL_\w^T(I-\CC^{T}\CC)^{-1}\CC^T\LL_\z & \LL_\w^T(I-\CC^{T}\CC)^{-1}\LL_\w
 \end{pmatrix}.
\]
 Then it follows immediately that 
\begin{align}
    \LL_\z^T(I-\CC\CC^{T})^{-1}\LL_\z&=\LL_\z^T(I-(-\CC)(-\CC)^{T})^{-1}\LL_\z,\nonumber\\
 \LL_\w^T(I-\CC^{T}\CC)^{-1}\LL_\w&=\LL_\w^T(I-(-\CC)^{T}(-\CC))^{-1}\LL_\w,\nonumber
 \end{align}
 i.e., the diagonal blocks of the prior precision operator are invariant to the sign of $\CC$, and 
 \begin{align}
 -\LL_\z^T\CC(I-\CC^{T}\CC)^{-1}\LL_\w&=\LL_\z^T(-\CC)(I-(-\CC)^{T}(-\CC))^{-1}\LL_\w,\nonumber\\
  -\LL_\w^T(I-\CC^{T}\CC)^{-1}\CC^T\LL_\z&=\LL_\w^T(I-(-\CC)^{T}(-\CC))^{-1}(-\CC)^T\LL_\z,\nonumber 
\end{align}
i.e, the off-diagonal blocks of the prior precision operator change sign.
\end{proof}

Lemma~\ref{lem: pri} shows that under the transformation 
$\CC\mapsto-\CC$ only the off-diagonal blocks of the prior precision matrix change sign. In the separable linear setting with block-diagonal noise covariance, it follows that the posterior precision matrix inherits the same symmetry. As an immediate consequence, we have the following corollary that shows that this implies invariance of the marginal posterior covariance matrices.
\begin{corollary}
Let $s=(\z,\w)\in\mathbb{R}^n$ be normally distributed with covariance matrix given by $\Gamma$ in (\ref{eq: appGam}), and let $G\in\mathbb{R}^{q\times n}$ be the block diagonal matrix with diagonal blocks $G_1\in\mathbb{R}^{q_1\times n_1}$ and $G_2\in\mathbb{R}^{q_2\times n_2}$. Furthermore,  assume $e\in\mathbb{R}^q$ is Gaussian with block diagonal covariance matrix  $\Gamma_e\in\mathbb{R}^{q\times q}$ with diagonal blocks $\Gamma_{e_1}\in\mathbb{R}^{q_1\times q_1}$ and $\Gamma_{e_2}\in\mathbb{R}^{q_2\times q_2}$. Then under the assumptions of Lemma~\ref{lem: pri}, in the linear setting, i.e, 
\begin{align}
    \data=Gs+e,\label{eq: appLin}
\end{align}
then the marginal posterior covariance matrices of $\z$ and $\w$ are invariant under the sign change $\CC\mapsto-\CC$.
\end{corollary}

\end{document}